%% file: main_iid.tex
\documentclass[12pt]{article}
\usepackage{amsmath}
\usepackage{graphicx}
\usepackage{enumerate}
\usepackage{natbib}
\usepackage{url} 

\RequirePackage[OT1]{fontenc}
\RequirePackage{amsthm}
\RequirePackage[colorlinks,citecolor=blue,urlcolor=blue]{hyperref}
\RequirePackage{hypernat}
\usepackage{colortbl}
\usepackage{longtable}
\usepackage{numprint}
\usepackage[utf8]{inputenc} 
\usepackage{booktabs}       
\usepackage{amsfonts}       
\usepackage{amssymb}
\usepackage{nicefrac}       
\usepackage{microtype}      
\usepackage{lipsum}
\usepackage{comment}
\usepackage{setspace}
\usepackage{float}
\graphicspath{ {./images/} }
\usepackage{quoting}
\usepackage{wrapfig}
\quotingsetup{font={itshape}}
\usepackage{multirow}
\usepackage{bm}
\usepackage{subfigure}

\newcommand{\blind}{1}

\numberwithin{equation}{section}
\theoremstyle{definition}

\newtheorem{theorem}{Theorem}[section]
\newtheorem{proposition}{Proposition}[section]

\newtheorem{remark}{Remark}[section]
\newtheorem{assumption}{Assumption}[section]

\def\0{{\text{\boldmath $0$}}}

\def\R{{\text{\boldmath $R$}}}

\def\V{{\text{\boldmath $V$}}}
\def\W{{\text{\boldmath $W$}}}

\def\Y{{\text{\boldmath $Y$}}}

\DeclareMathOperator{\Var}{Var}
\DeclareMathOperator{\Cov}{Cov}

\newcommand{\argmin}{\mathop{\rm arg~min}\limits}
\def\diag{{\rm diag}}

\allowdisplaybreaks

\begin{document}

\bibliographystyle{agsm}

\def\spacingset#1{\renewcommand{\baselinestretch}%
{#1}\small\normalsize} \spacingset{1}


\if1\blind
{
  \title{\bf Local-Polynomial Estimation for Multivariate Regression Discontinuity Designs.}
  \author{Masayuki Sawada\thanks{
    We thank Yoichi Arai, Hidehiko Ichimura, Hiroaki Kaido, Toru Kitagawa, anonymous associate editor and referees, and the participants of the third Tohoku-ISM-UUlm workshop at Tohoku University, Summer Econometrics Forum at the University of Tokyo, and the seminar at Hitotsubashi University and Kobe University for their valuable comments.
    This work was supported by the JSPS KAKENHI Grant Number 22K13373 (ISHIHARA), the JSPS KAKENHI Grant Number 23K12456 (KURISU), the Grant-in-Aid for Scientific Research (B) 21H03400 (MATSUDA) and the JSPS KAKENHI Grant Number 21K13269 (SAWADA). First version: 2024-Feb-15.}\hspace{.2cm}\\
    Institute of Economic Research, Hitotsubashi University, \\
    Takuya Ishihara \\
    Graduate School of Economics and Management, Tohoku University, \\
    Daisuke Kurisu \\
    Center for Spatial Information Science, The University of Tokyo \\
    and \\
    Yasumasa Matsuda \\
    Graduate School of Economics and Management, Tohoku University}
  \maketitle
} \fi

\if0\blind
{
  \bigskip
  \bigskip
  \bigskip
  \begin{center}
    {\LARGE\bf Local-Polynomial Estimation for Multivariate Regression Discontinuity Designs.}
\end{center}
  \medskip
} \fi

\bigskip
\begin{abstract}
We study a multivariate regression discontinuity design in which treatment is assigned by crossing a boundary in the space of multiple running variables. We document that the existing bandwidth selector is suboptimal for a multivariate regression discontinuity design when the distance to a boundary point is used for its running variable, and introduce a multivariate local-linear estimator for multivariate regression discontinuity designs. Our estimator is asymptotically valid and can capture heterogeneous treatment effects over the boundary. We demonstrate that our estimator exhibits smaller root mean squared errors and often shorter confidence intervals in numerical simulations. We illustrate our estimator in our empirical applications of multivariate designs of a Colombian scholarship study and a U.S. House of representative voting study and demonstrate that our estimator reveals richer heterogeneous treatment effects with often shorter confidence intervals than the existing estimator.
\end{abstract}

\noindent%
\vfill


\newpage
\spacingset{1.9} 

\section{Introduction}
\label{sec:intro}

The regression discontinuity (RD) design takes advantage of a particular treatment assignment mechanism that is set by the running variables. \footnote{See \cite{Imbens.Lemieux2008}, \cite{Lee.Lemieux2010}, \cite{DiNardo.Lee2011}, and \citeauthor*{Cattaneo.Idrobo.Titiunik2019} (\citeyear{Cattaneo.Idrobo.Titiunik2019},\citeyear{Cattaneo.Idrobo.Titiunik2023}) for extensive surveys of RD literature} An example of such a mechanism is a scholarship that is awarded to applicants whose scores are above a threshold. The eligibility sometimes involves an additional requirement. For example, the applicants' poverty scores must be below another threshold to be eligible. These RD designs are \textit{multivariate} in their running variables because a student must exceed a policy boundary in the space of multivariate running variables to be treated. 

Existing approaches often handle multivariate designs as if they are univariate designs. \footnote{There are a few studies which tackled the multivariate problem as multivariate. For example, \cite{Papay.Willett.Murnane2011} and \cite{reardonRegressionDiscontinuityDesigns2012} are early exception which consider extensions of the classical polynomial based estimation of \cite{Imbens.Lemieux2008}.} The most popular approach \textit{aggregates} observations over the boundary to handle multivariate RD designs. For example, \cite{Matsudaira2008} considers participation in a program based on either a failure in language or math exams. \cite{Matsudaira2008} reduces the multivariate design by aggregating the language-passing students who are at the boundary of the math exam.\footnote{\cite{wongAnalyzingRegressionDiscontinuityDesigns2013} consider a decomposition of the boundary average effects into a weighted average of the boundary specific estimate of the similar strategy.} While there is no theoretical issue with the \textit{aggregation} strategy, one may wish to estimate heterogeneous treatment effects across the policy boundary. \footnote{If we segment the boundary into a few intervals, then we may estimate heterogeneous effects separately for each segment.
Nevertheless, finding an appropriate set of segments can be challenging and one cannot easily take its limit of this strategy to estimate the heterogeneous effect at each boundary point.}

To estimate heterogeneous treatment effects over the boundary, another popular approach constructs a running variable as the Euclidean \textit{distance} from a boundary point. For example, \cite{Keele.Titiunik2015} propose a procedure to conduct the ordinary univariate regression discontinuity estimation with the Euclidean distance from a particular boundary point. \footnote{The distance approach dates back to \cite{Black1999}, for example, which computes the closest boundary point for each unit and compares units of the same closest boundary point to achieve the mean effect across the boundary. In this paper, we focus on estimating the heterogeneous effects across the boundary points.} The \textit{distance} approach produces a valid estimate with a valid inference under the procedure of \cite{Calonico.Cattaneo.Titiunik2014} because of its self-normalizing property of the t-statistic. \footnote{We thank an anonymous referee for this point.} The estimator is straightforward to implement, and available as Stata and R packages, \textit{rdrobust} or its wrapper \textit{rdmulti} (\citealp*{Cattaneo.Titiunik.Vazquez-Bare2020}).

However, the \textit{distance} strategy selects bandwidth for the \textit{incorrect} rate of convergence for the underlying multivariate design: the existing estimators select the optimal bandwidth for a univariate problem, but the underlying design is multivariate. As a result, the existing bandwidth selectors are \textit{suboptimal} and hence their estimations are \textit{inefficient}. 

In this study, we document that the existing bandwidth selectors including \cite{Calonico.Cattaneo.Titiunik2014} are suboptimal when they are applied to a multivariate design with the distance from a boundary point as a running variable. We further propose a multivariate RD estimator with a Mean-Squared Error (MSE) optimal bandwidth selector. We demonstrate preferable properties of our estimator in simulation and empirical analyses.

Our estimator demonstrates favorable performances with smaller MSEs and shorter confidence intervals in most of designs. We demonstrate our estimator in two empirical contexts to compare with \textit{rdrobust}. First, we apply our estimates to the multivariate RD design data of \citet*{Londono-Velez.Rodriguez.Sanchez2020} who study the impact of a Colombian scholarship program on the college attendance rate. Second, we consider a pseudo-multivariate RD design for the \cite{Lee2008} data with continuous covariates to study heterogeneous treatment effects across different values of the covariates. In the first application, our estimates exhibit shorter or comparable confidence intervals and better stability in the choice of scaling in two runnnig variables. In the second application, our estimates reveal shorter confidence intervals or richer heterogeneity.

We contribute to the literature on the estimation for RD designs. For a scalar running variable, the local-linear estimation of \cite{Calonico.Cattaneo.Titiunik2014} is the first choice for estimating treatment effects. Its statistical package, \textit{rdrobust} (\citealp*{Calonico.Cattaneo.Titiunik2014a}, \citealp*{Calonico.Cattaneo.Farrell.Titiunik2017}, \citealp*{Calonico.Cattaneo.Farrell2022}), is the dominant and reliable package for a uni-variate RD design with a large sample. However, we demonstrate that the \textit{rdrobust} bandwidth selector is suboptimal for a multivariate RD design when the distance from a boundary point is its univariate running variable. We further provide an alternative local-linear estimator with an optimal bandwidth selector.

We note that multivariate estimations are also available in a non-kernel \textit{bias-aware} procedure such as \cite{Imbens.Wager2019} and \cite{Kwon.Kwon2020a} which are derived from \cite{Armstrong.Kolesar2018}, for example. These bias-aware methods may fully adapt the underlying distribution of the running variable. The bias-aware approach is a valid but different alternative to the kernel procedure because they employ the worst-case second derivative as the tuning parameter instead of the bandwidth. Given that the two approaches are in different principle, we contribute to fill a missing piece in the kernel procedure for a multivariate RD design with an optimal bandwidth selector.

The most related study is the recent work by \cite{cattaneoEstimationInferenceBoundary2025} which has reported an important boundary bias in the \textit{distance} approach when the evaluation point is at the corner of the policy boundary. Combined with our arguments about the suboptimality of the \textit{distance} approach, both contributions jointly alert that the univariate \textit{distance} approach should not be used for the multivariate designs not just at the corner but also at any boundary points. Contributions in our estimator is also complementary. On the one hand, we allow for selecting dimension specific bandwidths which can differ substantially when the scaling of the running variables differ as demonstrated in our simulation. On the other hand, \cite{cattaneoEstimationInferenceBoundary2025} provide a uniform inference across evaluation points. Hence, both contributions are complementary both in terms of studying the \textit{distance} approach as well as developing an appropriate estimator.

The remainder of the paper is organized as follows. We document the problem of the existing approach and introduce our estimator in Section 2. In Section 3, we evaluate the proposed estimator in Monte Carlo simulations and in empirical studies by \cite{Londono-Velez.Rodriguez.Sanchez2020} and a modification of \cite{Lee2008}. Finally, we conclude the paper and discuss future challenges in Section 4.

\section{Methods}

\subsection{Set up and identification}
Consider a multivariate RD design for a student with a pair of test scores $(R_1,R_2)$. For example, we consider a program that accepts students whose scores exceed their corresponding thresholds $(c_1,c_2)$. In this program, the eligibility is set by a treatment region $\mathcal{T} = \{(R_1,R_2) \in \mathbb{R}^2: R_1 \geq c_1, R_2 \geq c_2 \}$ (Figure \ref{fig:treatment} (a)). For another example, consider a program that accepts students whose total score exceeds a single threshold $c_1 + c_2$. The eligibility is set by another region $\mathcal{T} = \{(R_1,R_2) \in \mathbb{R}^2: R_1 + R_2 \geq c_1 + c_2 \}$ (Figure \ref{fig:treatment} (b)). In general, we consider a binary treatment $D \in \{0,1\}$ and associated pair of potential outcomes $\{Y(1),Y(0)\}$ such that 
$
 Y = D Y(1) + (1 - D)Y(0)
$
for an observed outcome $Y \in \mathbb{R}$. We consider a sharp RD design with a \textit{vector} of running variables $R \in \mathcal{R} \subseteq \mathbb{R}^d$ for some integer $d \geq 1$. Specifically, let $\mathcal{T}$ be the treatment region, which is an open subset of the support, $\mathcal{R}$. Let $\mathcal{T}^C$ be the complement of the closure of $\mathcal{T}$. This $\mathcal{T}^C$ is the control region, and both $\mathcal{T}$ and $\mathcal{T}^C$ have non-zero Lebesgue measures, and $D = 1\{R \in \mathcal{T}\}$.

\begin{figure}[H]
    \centering
    \begin{minipage}{0.45\hsize}
    \includegraphics[width=\columnwidth]{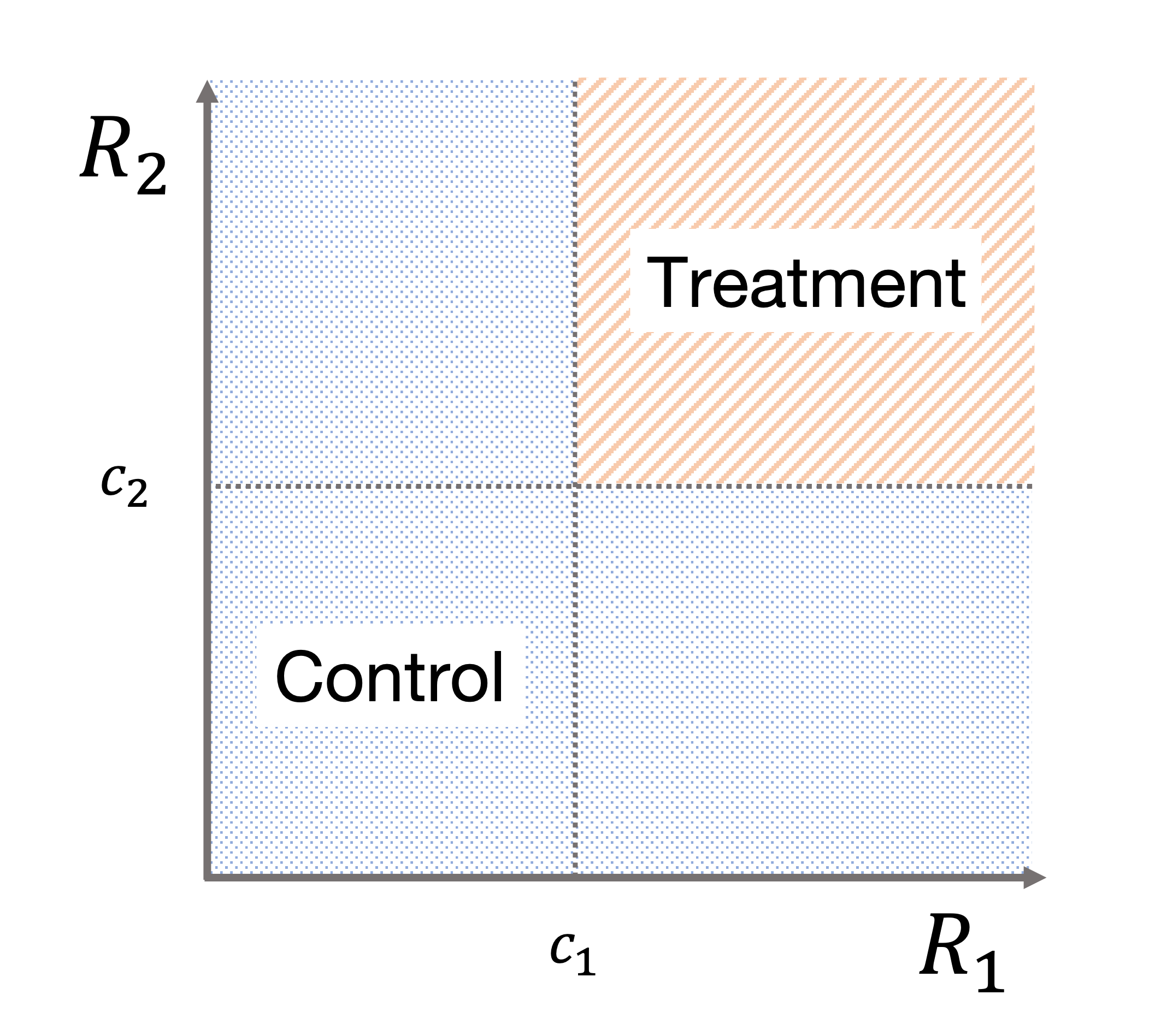}
    \centering (a) $\{(R_1,R_2) \in \mathbb{R}^2: R_1 \geq c_1, R_2 \geq c_2 \}$
    \end{minipage}    
    \begin{minipage}{0.45\hsize}
    \includegraphics[width=\columnwidth]{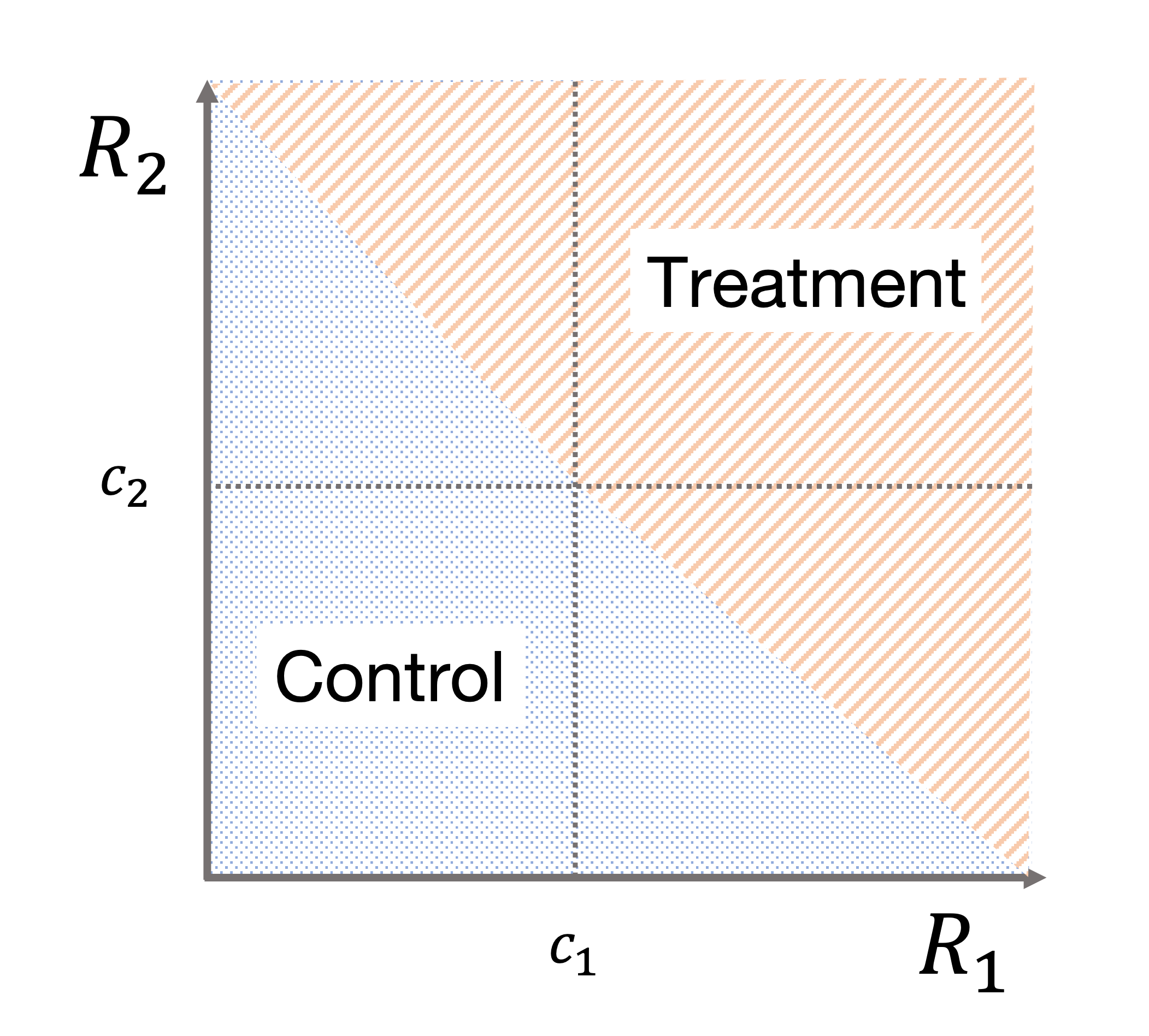}
    \centering (b) $\{(R_1,R_2) \in \mathbb{R}^2: R_1 + R_2 \geq c_1 + c_2 \}$
    \end{minipage}
    \caption{Illustration of $\mathcal{T}$.}\label{fig:treatment}
\end{figure}

We consider the i.i.d. sample of $(Y,D,R)$, $(Y_i,D_i,R_i)_{i \in \{1,\ldots, n\}}$, where $R_i = (R_{i,1},R_{i,2})$ and $R = (R_1,R_2)$. Let $c$ be a particular point on the boundary of the closure of $\mathcal{T}$. Our target parameter is $\theta(c) := \lim_{r \rightarrow c, r \in \mathcal{T}} E[Y(1)-Y(0)|R=r] - \lim_{r \rightarrow c, r \in \mathcal{T}^C} E[Y(1)-Y(0)|R=r]$. In the following section, we focus on the issues in estimating the given identified parameter $\theta(c)$. Under the following assumptions (\citealp*{Hahn.Todd.Klaauw2001}; \citealp{Keele.Titiunik2015}), $\theta(c)$ is the average treatment effect (ATE) at each point of the boundary $c$:
\begin{proposition}(\citealp{Keele.Titiunik2015}, Proposition 1)
 If $E[Y(1)|R=r]$ and $E[Y(0)|R=r]$ are continuous in $r$ at all points $c$ of the boundary of the closure of $\mathcal{T}$; $P(D_i = 1) = 1$ for all $i$ such that $R_i \in \mathcal{T}$; $P(D_i = 1) = 0$ for all $i$ such that $R_i \in \mathcal{T}^C$, then, $\theta(c) = E[Y(1) - Y(0)|R = c]$
 for all $c$ in the boundary.
\end{proposition}

\subsection{Issues in Conventional Estimators}

To estimate heterogeneous treatment effects over the boundary points, one often employs the \textit{distance} strategy which explicitly reduces a multivariate running variable to a scalar distance measure. A frequent choice is the Euclidean distance from a point or the closest boundary \citep{Keele.Titiunik2015}. The \textit{distance} strategy can be easily implemented in most designs via the local-linear estimation \citep[for example]{Fan.Gijbels1992} for the uni-variate RD designs with a MSE optimal bandwidth selection. However, the existing bandwidth selector is not rate optimal when it uses for the \textit{distance} strategy for a multivariate design.


Our first observation is the property of the density function of the \textit{distance} running variable at a boundary point. Let $Z_i$ be the scalar running variable as a distance from a boundary point $c$. Then its density $f_Z(z)$ shrinks to zero as it approaches the boundary when the distance $\tilde{d}$ bounds the Euclidean distance with some constant:
\begin{proposition}\label{prop.zero_density}
    Let $\tilde{d}(\cdot, \cdot)$ be a distance on $\mathbb{R}^d$ such that $\bar{c} \|a - b\| \leq \tilde{d}(a,b)$ for any $a, b \in \mathbb{R}^d$ and some constant $\bar{c} > 0$. Here $\|a - b\|$ is the Euclidean distance between $a = (a_1,\ldots, a_d)'$ and $b = (b_1,\ldots, b_d)'$. Define $Z_i = \tilde{d}(R_i,c)$ with $c = (0,\ldots, 0)'$ and assume that $R_i$ and $Z_i$ have density functions $f$ and $f_Z$, respectively. \footnote{The boundary point value $c$ is set to zero for illustration. The same argument applies in general by normalizing the running variables with respect to the boundary point.
    The distance $\tilde{d}$ includes the Euclidean norm $\|a - b\|$, $\ell^{\infty}$-norm $\|a - b\|_{\infty} = \max_{1 \leq j \leq d} |a_j - b_j| \geq (1/d) \|a - b\|$, and $\ell^1$-norm $\|a - b\|_1 = \sum_{j=1}^d |a_j - b_j| \geq \|a - b\|$.}
    
    Assume that $f$ and $f_Z$ are continuous. Then we have $f_Z(z) \rightarrow 0$ as $z \rightarrow 0$.
\end{proposition}
\begin{proof}
    By construction of $Z_i$, for  $z > 0$,
      \begin{align*}
        \int_0^z f_Z(r) dr =& P(Z_i \leq z) = P(\tilde{d}(R_i,0) \leq z)
        \leq P(\bar{c} \|R_i\| \leq z) = P(\|R_i\| \leq z/\bar{c})\\
        =& \int_0^{z/\bar{c}} t\left(\int_0^{2\pi} f(t \cos \theta, t \sin \theta) d\theta \right)dt\\
        =& \int_0^z (1/\bar{c})^{2} r \left(\int_0^{2\pi} f(\bar{c}^{-1} r \cos \theta , \bar{c}^{-1} r \sin \theta) d\theta \right) dr
    \end{align*}
    where the last equality uses the change of variable $r = \bar{c} t$. If $f$ is continuous, $f_Z(0) = 0$ by using the above inequality. Since $f_Z$ is continuous, the statement follows.
\end{proof}

To illustrate the proposition in an example, consider $R_i = (R_{1i}, R_{2i})$ where $R_{1i}$ and $R_{2i}$ independent each other, and $R_{1i} \sim U[-1,1]$ and $R_{2i} \sim U[0,1]$. The distribution function of $Z_i = \|R_i\|$ is $P(Z_i \leq z) = P(R_{1i}^2 + R_{2i}^2 \leq z^2) = (\pi/4) z^2$. The half-circle area shrinks to zero at the order of $z^2$ as $z$ approaches the value $0$ at the boundary point $(0,0)$. 


This zero-density problem itself may appear to be not an immediate problem for the \cite{Calonico.Cattaneo.Titiunik2014} (henceforth, CCT or \textit{rdrobust}) estimator as its bandwidth selector does not estimate the density directly. Nevertheless, it leads to another problem that makes CCT bandwidth selector suboptimal when it is used for the \textit{distance} strategy. 

To demonstrate its mechanism, first we consider the simpler \cite{Imbens.Kalyanaraman2012} (IK) bandwidth selector with the Euclidean distance running variable $Z_i = D_i \|R_i\| - (1-D_i) \|R_i\|$. The IK bandwidth selector for the \textit{distance} strategy takes the following form
\[
 \hat{h}_{IK} = C \cdot \left( \frac{\hat{V}_{IK} / \hat{f}_Z(0)}{\hat{B}_{IK}} \right)^{1/5}  n^{-1/5}
\]
where $\hat{B}_{IK}$ depends on the regularization term and the estimator of the second derivative of $E[Y_i|Z_i=z]$, $\hat{V}_{IK}$ depends on the estimator of the conditional variance $V(Y_i|Z_i=z)$, $\hat{f}_Z(0) = \frac{1}{nh_{\mathrm{pilot}}} \sum_{i=1}^n K(Z_i/h_{\mathrm{pilot}})$ with some kernel function $K$ and the pilot bandwidth $h_{\mathrm{pilot}}$. In Online Appendix \ref{sec.density}, we show that $\hat{f}_{Z}(0)$ converges to zero while $h_{\mathrm{pilot}}^{-1} \cdot \hat{f}_{Z}(0)$ converges to a positive constant when $f(0)$ is strictly positive. Hence, if $\hat{V}_{IK}$ and $\hat{B}_{IK}$ converge to strictly positive constants, then the \textit{variance term} $\hat{V}_{IK}/\hat{f}_{Z}(0)$ diverges while $h_{pilot} \hat{V}_{IK}/\hat{f}_{Z}(0)$ converges to a strictly positive constant. Hence, we obtain 
$
\hat{h}_{IK} = O_p(h_{\mathrm{pilot}}^{-1/5} n^{-1/5}),
$
and the rate of $\hat{h}_{IK}$ depends on the pilot bandwidth. For instance, if $h_{\mathrm{pilot}} = O_p(n^{-1/5})$, then $\hat{h}_{IK} = O_p(n^{-4/25})$, which is suboptimal in a two-dimensional estimation problem.

This \textit{diverging variance term} problem arises in CCT bandwidth selector as well even though it avoids the density estimation directly. The CCT bandwidth selector has the form
\begin{equation}
    \hat{h}_{CCT} = C \cdot \left( \frac{\tilde{V}_{CCT}}{\tilde{B}_{CCT}}  \right)^{1/5} n^{-1/5} \label{CCT_band_distance}
\end{equation}
where $\tilde{B}_{CCT}$ depends on the regularization term and the estimator of the second derivative of $E[Y_i|Z_i=z]$ and $\tilde{V}_{CCT}$ depends on the estimators of the conditional variance $V(Y_i|Z_i=z)$ and $f_Z(0)$ but they are estimated in a sandwich form so that the density estimation does not arise explicitly. Specifically, in Online Appendix F, we study the variance term $
    \tilde{V}_{CCT} = n h_{\mathrm{initial}} \left\{ \hat{V}_{+}(h_{\mathrm{initial}}) + \hat{V}_{-}(h_{\mathrm{initial}}) \right\}
$ for $h_{\mathrm{initial}} = O_p(n^{-1/5})$ where the variance term is the sum of the elements of sandwich forms: $\hat{V}_{+}(h) = e_1' \Gamma_{+}(h)^{-1} \Psi_{+}(h) \Gamma_{+}(h)^{-1} e_1 /n$, and $\hat{V}_{-}(h) = e_1' \Gamma_{-}(h)^{-1} \Psi_{-}(h) \Gamma_{-}(h)^{-1} e_1 /n$ for $r_1(z) = (1,z)'$ and $e_1 = (1,0)'$. \footnote{See Online Appendix F for the formal definitions for $\Gamma_+, \gamma_-, \Psi_+,$ and $\Psi_-$. We consider a simplified version for the $\Psi_+$ and $\Psi_{-}$ matrices which take a known variance function instead of the original formula with the plug-in estimates of the residual variance function.} As Proposition F.1, we show that $h^{-1}\Gamma_+(h)$ and $\Psi_+(h)$ converges to constant matrices, instead of $\Gamma_+(h)$ and $h \Psi_+(h)$ convergence as required in the original procedure. Hence, the positive side of the original variance term $n h_{\mathrm{initial}} \hat{V}_{+}(h)$ diverges while
\[
    n h_{\mathrm{initial}}^2 \hat{V}_{+}(h_{\mathrm{initial}}) = e_1' \left\{ h_{\mathrm{initial}}^{-1}\Gamma_{+}(h_{\mathrm{initial}}) \right\}^{-1} \Psi_{+}(h_{\mathrm{initial}}) \left\{ h_{\mathrm{initial}}^{-1}\Gamma_{+}(h_{\mathrm{initial}}) \right\}^{-1} e_1
\]
converges to a positive constant. If $\tilde{B}_{CCT}$ also converges to a positive constant, we have
\[
\hat{h}_{CCT} = C \cdot \left( \frac{h_{\mathrm{initial}} \tilde{V}_{CCT}}{\tilde{B}_{CCT}} \right)^{1/5} h_{\mathrm{initial}}^{-1/5} n^{-1/5} = O_p\left( h_{\mathrm{initial}}^{-1/5} n^{-1/5} \right).
\]
Hence, if $h_{\mathrm{initial}} = O_{p}(n^{-1/5})$, the convergence rate of $\hat{h}_{CCT}$ is $n^{-4/25}$ which is the same suboptimal rate as the IK bandwidth for the multivariate \textit{distance} bandwidth selector. See the complete discussion for the Online Appendix F.

\subsection{Our Estimator} \label{sec.estimator}

Given the suboptimality of the \textit{distance} strategy, we propose a multivariate RD estimator for the heterogeneous treatment effect over the boundary with the MSE optimal bandwidths.

We demonstrate our estimator in a special case of two-dimensional running variables. Consider the following local-linear estimator $\hat{\beta}^+(c)= (\hat{\beta}^+_0(c), \hat{\beta}^+_1(c),\hat{\beta}^+_2(c))'$
\[
\hat{\beta}^+(c) = \argmin_{(\beta_0,\beta_1,\beta_2)' \in \mathbb{R}^{3}}\sum_{i=1}^{n}(Y_i - \beta_0 - \beta_1(R_{i,1} - c_1) - \beta_{2}(R_{i,2}-c_2))^2K_h\left(R_i - c\right)1\{R_i \in \mathcal{T}\}
\]
where
$
K_{h}(R_i - c) = K\left(({R_{i,1} - c_1) / h_1},{(R_{i,2} - c_2) / h_2}\right)
$
and each $h_j$ is a sequence of positive bandwidths such that $h_j \to 0$ as $n \to \infty$. Similarly, let $\hat{\beta}^-(c)$ be the estimator using $1\{R_i \in \mathcal{T}^c\}$ subsample. Hence, our multivariate RD estimator at $c$ is $\hat{\beta}^+_0(c) - \hat{\beta}^-_0(c)$. Our estimator uses the theoretical results of \cite{Ruppert.Wand1994}, \cite{Masry1996}, and \cite{guMultivariateLocalPolynomial2015} for the multivariate RD designs. Specifically, we employ \cite{guMultivariateLocalPolynomial2015} with a slightly extended result such as non-product kernels and explicit higher-order expressions to allow us to conduct the \cite{Calonico.Cattaneo.Titiunik2014} style bias-correction procedure.  

As we consider a random sample, the treated sample is independent of the control sample. Without the loss of generality, we consider the following nonparametric regression models for each sample: $Y_i = m_+(R_i) + \varepsilon_{+,i},\ E[\varepsilon_{+,i}|R_i] = 0,\ i \in \{1,\dots, n: R_i \in \mathcal{T}\}$ and $Y_i = m_-(R_i) + \varepsilon_{-,i},\ E[\varepsilon_{-,i}|R_i] = 0,\ i \in \{1,\dots, n: R_i \in \mathcal{T}^C\}.$

For the asymptotic normality, we impose the following regularity conditions that are standard in kernel regression estimations. We provide the conditions under its general possible form. In Online Appendix \ref{sec.asymptotic}, we present the general results for $p$th order local-polynomial estimation with $d$-dimensional running variables. The general results in the Online Appendix are the basis of the bias correction procedure of our estimator.

In Assumption \ref{Ass1}, we assume the existence of a continuous density function for the running variable $R$. Assumption \ref{Ass2} is the regularity conditions for a kernel function. We select a particular set of kernel functions for our subsequent analysis. Assumption \ref{Ass3} imposes a set of smoothness conditions for the conditional mean functions $m_+$ and $m_-$ and for the conditional moments of residuals $\varepsilon_{+,i}$ and $\varepsilon_{-,i}$. Assumption \ref{Ass4} specifies the rate of convergence of the vector of bandwidths $\{h_1,\ldots, h_d\}$ relative to the sample size $n$.
\begin{assumption}\label{Ass1}
Let $U_{r}$ be a neighborhood of $r=(r_1,\dots, r_d)' \in \mathcal{R}$. 
\begin{itemize}
\item[(a)] The vector of random variables $R_i$ has a probability density function $f$. 
\item[(b)] The density function $f$ is continuous on $U_r$ and $f(r)>0$. 
\item[(c)] For each $r \in \mathcal{R}$ on the boundary of the treated region $\mathcal{T}$, \textcolor{black}{there exists $\delta>0$ such that $\mathcal{T} \cap \prod_{j=1}^d[r_j -\delta, r_j + \delta] = [r_2,r_2+\delta] \times \prod_{j \neq 2} [r_j -\delta, r_j + \delta]$.} 
\end{itemize}
\end{assumption}

Condition (c) means that, in a sufficiently small neighborhood of the point $r$ of interest on the boundary of $\mathcal{T}$, the boundary is linearly separated. \footnote{This condition (c) excludes the evaluation of the corner point. The following implementation and the actual numerical simulation and empirical analyses avoid the evaluation exactly at the corner of the boundary. Adapting the finding of \cite{cattaneoEstimationInferenceBoundary2025}, the same statement should follow with a relaxed condition (c) which allows for the corner point. One may allow for more complex boundary structures in the neighborhood of $r$ in a different setting such as those in \cite{cattaneoEstimationInferenceBoundary2025}; however, extending our framework to their setting is beyond the scope of this paper.}

\begin{assumption}\label{Ass2}
Let $K :\mathbb{R}^{d} \to [0,\infty)$ be a kernel function such that
\textcolor{black}{
\begin{itemize}
\item[(a)]  $\int K_\pm(z)dz = 1$ where $K_+(z)=K(z)1\{z_2\geq 0\}$, and $K_-(z) = K(z)1\{z_2 <0\}$.
\item[(b)] The kernel function $K$ is bounded and there exists a constant $C_{K}>0$ such that $K$ is supported on $[-C_{K},C_{K}]^d$. 
\item[(c)] Define $\check{z}:= (1,(z)'_1,\dots,(z)'_p)',\ (z)_L = \left(\prod_{\ell=1}^{L}z_{j_\ell}\right)'_{1\leq j_1 \leq \dots \leq j_L\leq d},\ 1\leq L \leq p. $
The matrix $S_\pm=\int K_\pm(\bm{z})\left(
\begin{array}{c}
1 \\
\check{\bm{z}}
\end{array}
\right)
(1\ \check{\bm{z}}')d\bm{z}$ is non-singular.
\end{itemize}
}
\end{assumption}

\begin{assumption}\label{Ass3}
Let $U_{r}$ be a neighborhood of $r \in \mathcal{R}$. 
\begin{itemize}
\item[(a)] The mean function $m_+$ is $(p+1)$-times continuously partial differentiable on $U_r$ and define $\partial_{j_1\dots j_L}m_+(r):= \frac{\partial m_+(r)}{\partial r_{j_1}\dots r_{j_L}}$, $1 \leq j_1,\dots, j_L \leq d$, $0 \leq L \leq p+1$. When $L=0$, we set $\partial_{j_1 \dots j_L}m_+(r) = \partial_{j_0}m_+(r)= m_+(r)$. The parallel restriction holds for the mean function $m_-$.
\item[(b)] The variance function $\sigma_+^2(z) = E[\varepsilon_{+,i}^2|R_i = z]$ and  $\sigma_-^2(z) = E[\varepsilon_{-,i}^2|R_i = z]$ are continuous at $r$.
\item[(c)] There exists a constant $\delta>0$ such that $\sup_{z \in U_r}E[|\varepsilon_{+,1}|^{2+\delta}|R_1=z]\leq U(r) < \infty$ and $\sup_{z \in U_r}E[|\varepsilon_{-,1}|^{2+\delta}|R_1=z]\leq U(r) < \infty$
\end{itemize}
\end{assumption}

\begin{assumption}\label{Ass4}
As $n \to \infty$, 
\begin{itemize}
\item[(a)] $h_j \to 0$ for $1 \leq j \leq d$, 
\item[(b)] $nh_1 \cdots h_d \times h_{j_1}^2 \dots h_{j_p}^2 \to \infty$ for $1 \leq j_1 \leq \dots \leq j_{p} \leq d$,
\item[(c)] $nh_1 \cdots h_d \times h_{j_1}^2 \dots h_{j_p}^2h_{j_{p+1}}^2 \to c_{j_1\dots j_{p+1}} \in [0,\infty)$ for $1 \leq j_1 \leq \dots \leq j_{p+1} \leq d$.
\end{itemize} 
\end{assumption}

Under these assumptions, we establish the asymptotic normality of our estimator $\hat{\beta}^+_0(c) - \hat{\beta}^-_0(c)$. \textcolor{black}{The result follows from Theorem \ref{thm: LP-CLT} in Appendix \ref{sec.asymptotic}.}
\begin{theorem}[Asymptotic normality of local-linear estimators]\label{thm: LP-CLT-main-d2}
Under Assumptions \ref{Ass1}, \ref{Ass2}, \ref{Ass3} and \ref{Ass4} for $r=c$ with $d = 2$ and $p = 1$, as $n \to \infty$, we have
\begin{align*}
\sqrt{nh_1 h_2}
&\left[\right.H^{ll}\left((\hat{\beta}^+_0(c) - \hat{\beta}_0^-(c)) - e'_1(M_+(c) - M_-(c))\right)\\
-& \textcolor{black}{e'_1 \{S_+^{-1}B^{(2,1)}M_{+,n}^{(2,1)}(c) - S_-^{-1}B^{(2,1)}M_{-,n}^{(2,1)}\}}\left.\right] \\
&\textcolor{black}{\stackrel{d}{\to} N\left( \0 , e'_1\left\{\frac{\sigma_{+}^2(c)}{f(c)} S_+^{-1}\mathcal{K}_+S_+^{-1} +\frac{\sigma_{-}^2(c)}{f(c)} S_-^{-1}\mathcal{K}_-S_-^{-1}\right\}e_1 \right), }
\end{align*}
for $e_1 = (1,0,0)'$ and $H^{ll} = \diag(1,h_1,h_2) \in \mathbb{R}^{3 \times 3}$ where
\begin{align*}
M_+(c) &= \left(m_+(c),\partial_1 m_+(c), \partial_2 m_+(c)\right)', M_-(c) = \left(m_-(c),\partial_1 m_-(c), \partial_2 m_-(c)\right)',\\
M_{+,n}^{(2,1)}(r) 
&= \left(\frac{\partial_{11}m_+(c)}{2}h_1^{2},\partial_{12}m_+(c)h_1 h_2,\frac{\partial_{22}m_+(c)}{2}h_2^{2}\right)',\\
M_{-,n}^{(2,1)}(r) 
&= \left(\frac{\partial_{11}m_-(c)}{2}h_1^{2},\partial_{12}m_-(c)h_1 h_2,\frac{\partial_{22}m_-(c)}{2}h_2^{2}\right)',\mbox{ and }\\
B^{(2,1)} &= \int \left(
\begin{array}{c}
1 \\
\check{\bm{z}}
\end{array}
\right)
(\bm{z})'_{2}d\bm{z},\ \textcolor{black}{\mathcal{K}_\pm= \int K_\pm^2(\bm{z})\left(
\begin{array}{c}
1 \\
\check{\bm{z}}
\end{array}
\right)
(1\ \check{\bm{z}}')d\bm{z}.}
\end{align*}
\end{theorem}
Given the bias and variance expressions in Theorem \ref{thm: LP-CLT-main-d2}, we may find the common bandwidth $h = h_1 = h_2$ that minimizes the following asymptotic expansion of the mean-squared error (MSE) of $\hat{m}_+(c) - \hat{m}_{-}(c)$: for $e_1 = (1,0,0)'$,
\textcolor{black}{
\begin{align*}
&\underbrace{
 \left[e'_1 S_+^{-1} B^{(2,1)} \begin{pmatrix}
 \partial_{11}m_+(c)/2\\
 \partial_{12}m_+(c)\\
 \partial_{22}m_+(c)/2 
 \end{pmatrix}
 - e'_1 S_-^{-1} B^{(2,1)}\begin{pmatrix}
 \partial_{11}m_-(c))/2\\
 \partial_{12}m_-(c)\\
 \partial_{22}m_-(c))/2
 \end{pmatrix}\right]^2 h^4}_{\text{Bias term}}\\
&\quad  + \underbrace{\frac{1}{nh^2}\left\{\frac{\sigma_+^2(c)}{f(c)} e'_1 S_+^{-1} \mathcal{K}_+S_+^{-1}e_1+\frac{\sigma_-^2(c)}{f(c)} e'_1 S_-^{-1} \mathcal{K}_-S_-^{-1}e_1\right\}}_{\text{Variance term}}.
\end{align*}
}
 In general, it would be more intuitive and reasonable to consider heterogeneous bandwidth $h_1 \neq h_2$, and our main numerical illustrations are based on the heterogeneous bandwidths. In one of our empirical analysis dataset, two running variables take quite different ranges of values because one of the running variable has twice or three times larger scale than the other. If we use a common bandwidth, a possibly awkward squared area will be used for the estimation while it may be too large for one dimension and too small for the other dimension. One can avoid such an awkward situation by rescaling the running variables appropriately, but the results may change substantially by rescaling. The heterogeneous bandwidths allows users to avoid such a difficult rescaling task and use the original scaling for the estimation. See Appendix \ref{sec.hetero_band} for the details for the heterogeneous case.

\section{Numerical Results}\label{sec.numerical}
We demonstrate the numerical properties of our estimator in Monte Carlo simulations and empirical applications. Numerical simulations use the first empirical context of a Colombian scholarship, \citet*{Londono-Velez.Rodriguez.Sanchez2020,londono_data_2020}. Specifically, we evaluate the performances of our estimator in simulations which take higher-order approximations of the Colombian data as \textit{true} data generating processes, and in empirical application with the actual dataset. In their application, the scholarship of interest is primarily determined by two thresholds: merit-based and need-based. As a result, a policy \textit{boundary} exists instead of a single cutoff. 
Figure \ref{fig:score_map} is a scatter plot of two running variables with $30$ boundary points. 
The $30$ points are selected from taking $15$ points from the maximum value across the boundary to $0$. In the simulations below, we found that a largest point in SISBEN boundary is challenging for all methods evaluated, and we evaluate $28$ points denoted as red filled circles after removing the extreme boundary points denoted as blank black circles in the empirical application later. We explain the institutional details further in Section \ref{sec.application}.
\begin{figure}[H]
    \centering
    \includegraphics[width=0.75\hsize]{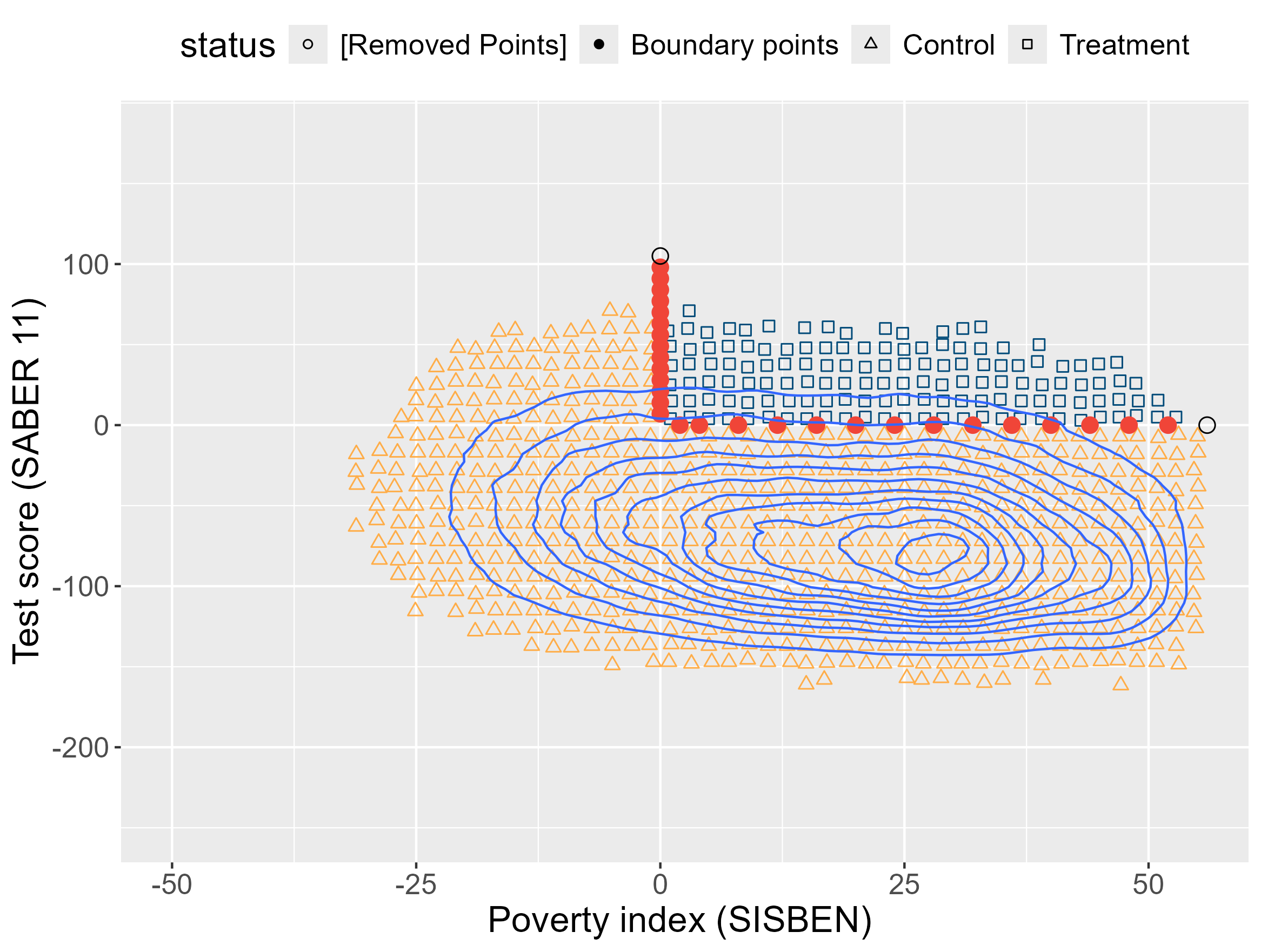}
    \caption{A binned scatter plot with joint density estimates in solid contour plot curves. The $x$-axis represents the SISBEN score minus the policy cutoff; the $y$-axis represents the SABER11 score minus the policy cutoff. Each bin has length $2$ in $x$-axis and $11$ in $y$-axis and its bin-wise median values in each axis are shown in the plot, excluding bins which have fewer than 20 observations in each bin. Circles over the boundary represent 30 points to evaluate in the simulation, where we use the filled $28$ points for the empirical analysis later. Positive scores in both measures imply that the requirements are satisfied. (Source: our calculation using \citealp{londono_data_2020})}
    \label{fig:score_map}
\end{figure}

\subsection{Simulation Results}\label{sec.simul}

Given the dataset, we constructed four designs which are all two-dimensional saturated higher-order polynomial approximations of the conditional expectation functions at four boundary points. Specifically, we use the fully saturated polynomials up to fourth orders plus the fifth order terms for $X$ and $Y$ each. The four boundary points are at a higher SISBEN (need-based) boundary $(7)$, an intermediate SIBEN boundary $(13)$, an intermediate SABER11 (merit-based) boundary $(19)$ and a higher SABER11 (merit-based) boundary $(25)$. Figures \ref{fig:shapes} show the two-dimensional plots of the mean functions. For each draw of a simulation sample, we draw a random sample of two-dimensional running variables as $R_1 \sim U[-1,1]$ and $R_2 \sim 2 \times Beta(2,4)-1$ independent of each other over a rescaled rectangular support, and generate the outcome variable as $m(R_{i1},R_{i2}) + \epsilon_i$ where $\epsilon_i \sim N(0,0.1295^2)$.
\begin{figure}[H]
    \centering 
    \begin{minipage}{0.49\hsize}
    (a) Design 1
    \includegraphics[width=\columnwidth]{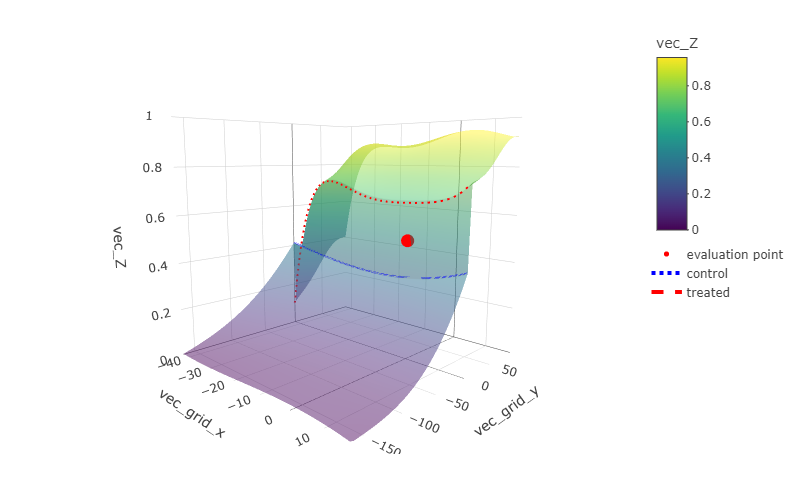}
    \centering 
    \end{minipage}    
    \begin{minipage}{0.49\hsize}
    (b) Design 2
    \includegraphics[width=\columnwidth]{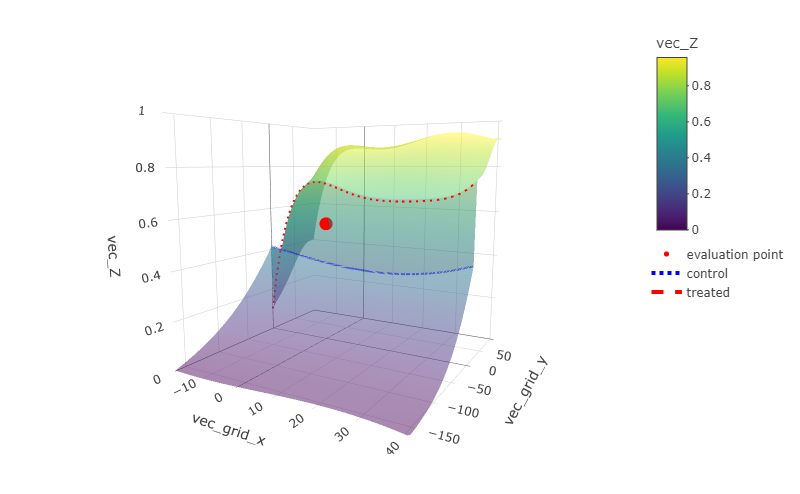}
    \centering
    \end{minipage}
    \begin{minipage}{0.49\hsize}
    (c) Design 3
    \includegraphics[width=\columnwidth]{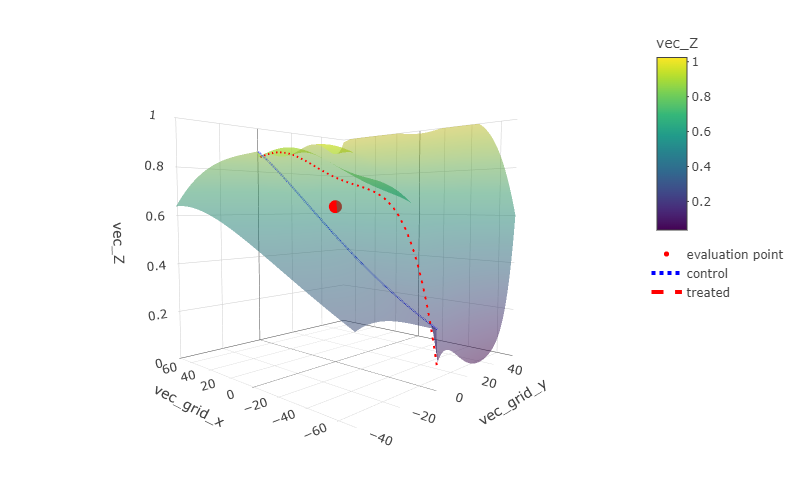}
    \centering 
    \end{minipage}
    \begin{minipage}{0.49\hsize}
    (d) Design 4
    \includegraphics[width=\columnwidth]{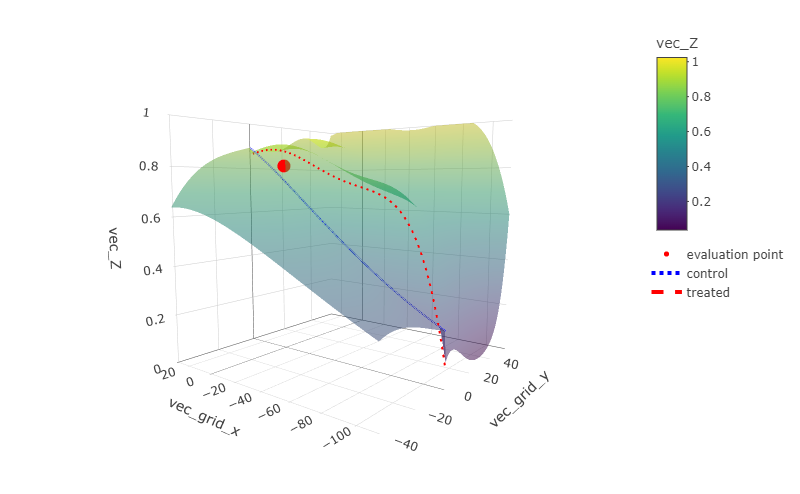}
    \centering 
    \end{minipage}
    \caption{3D plots of the mean functions at four boundary points. The horizontal line is the boundary; the center circle is the evaluation point. We rotate the axes so that the $X$-axis aligns with the boundary and the sign of $Y$-axis value determines the treatment status. See Appendix \ref{sec.polynomialShapes} for the exact polynomial shapes used and supports for each design.}
    \label{fig:shapes}
\end{figure}

We compare the quality of our estimator relative to the \textit{rdrobust} estimation. Figure \ref{fig:shape_2d} shows histograms of realized estimates of $10,000$ times replications for the primary data generating process. The light-colored histograms of our \textit{2D local poly} estimates tend to have thinner shapes than the dark-colored histograms of \textit{rdrobust} estimates.
\begin{figure}[H]
    \centering
    \begin{minipage}{0.49\hsize}
     (a) Design 1
    \includegraphics[width=\columnwidth]{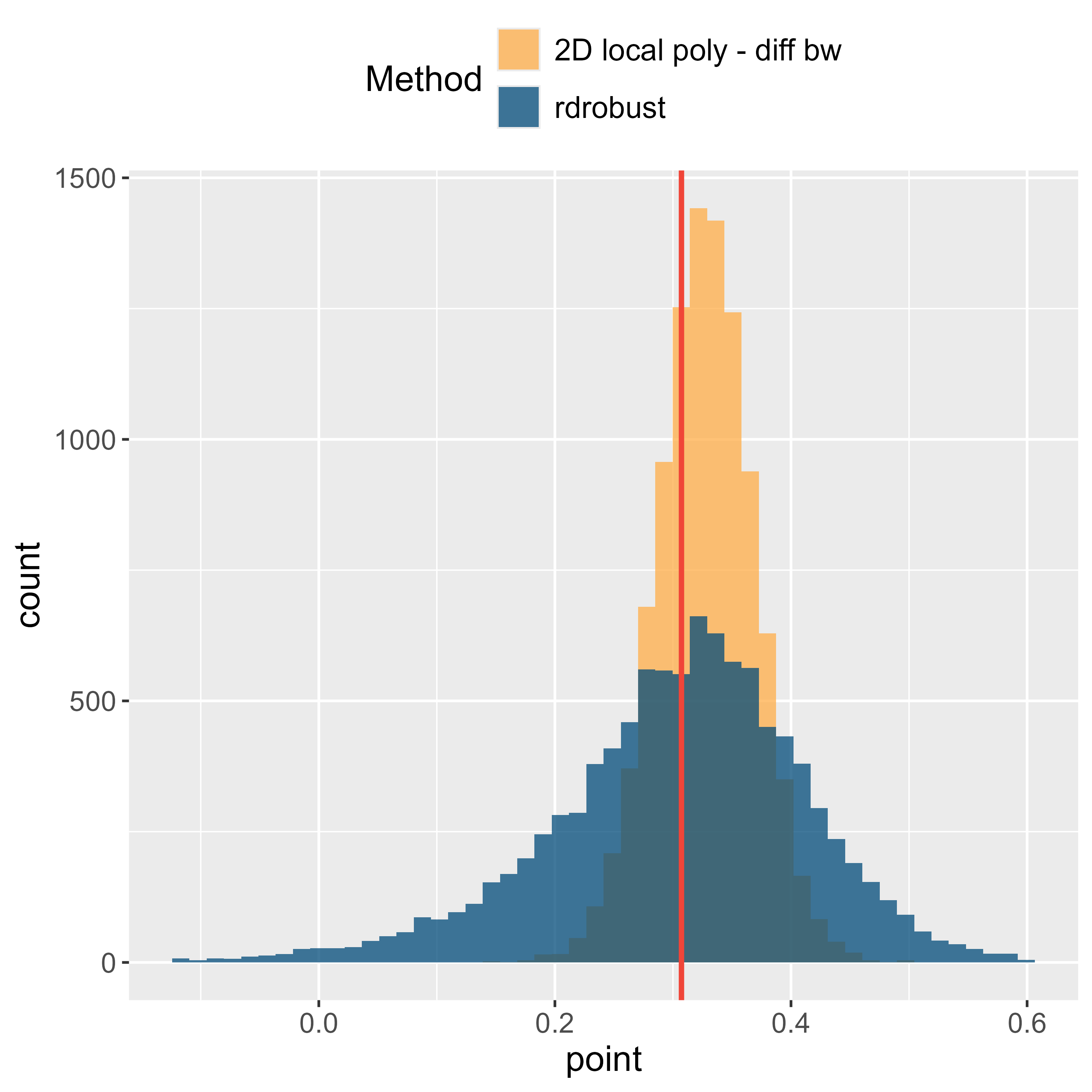}
    \centering
    \end{minipage}    
    \begin{minipage}{0.49\hsize}
    (b) Design 2
    \includegraphics[width=\columnwidth]{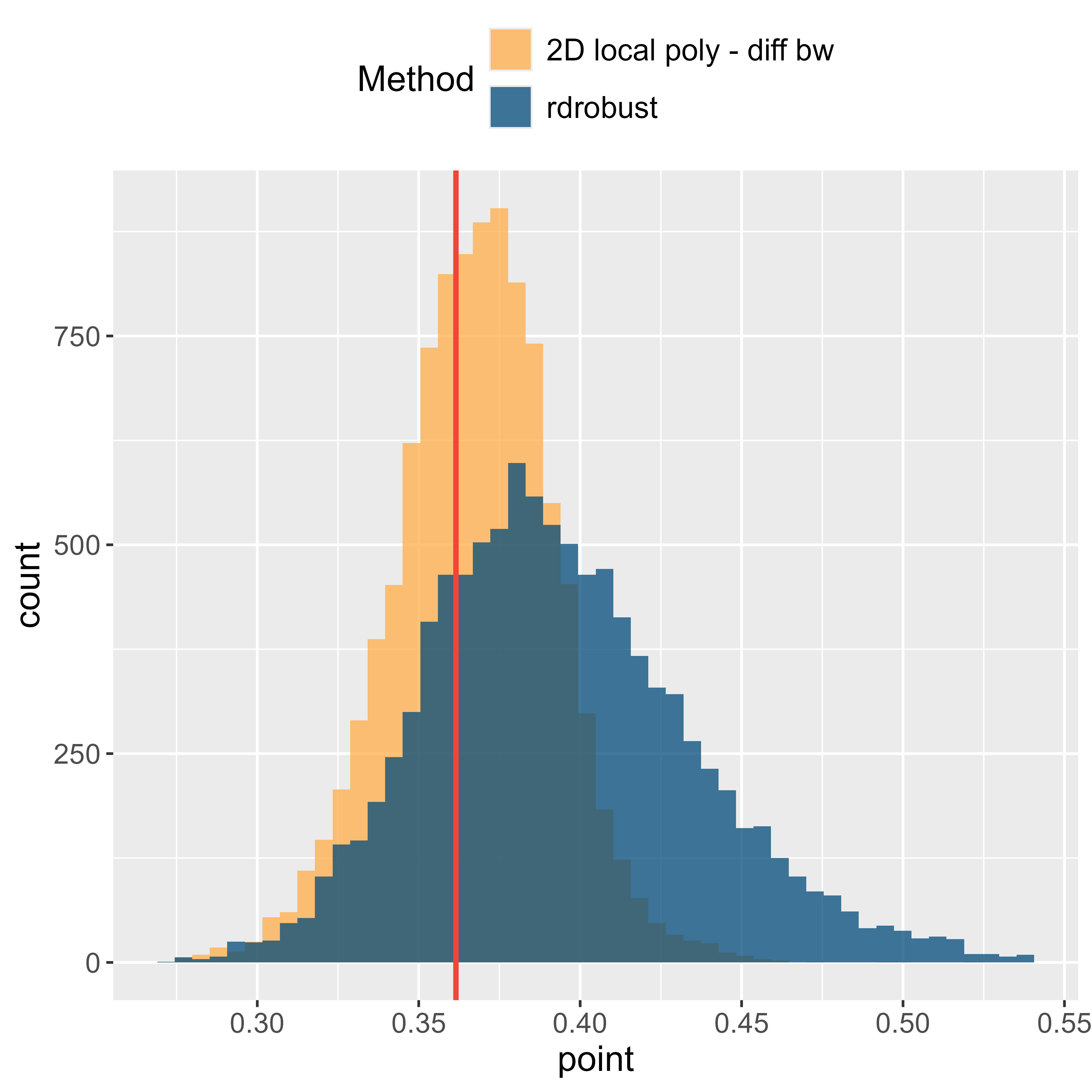}
    \centering
    \end{minipage}
    \begin{minipage}{0.49\hsize}
    (c) Design 3
    \includegraphics[width=\columnwidth]{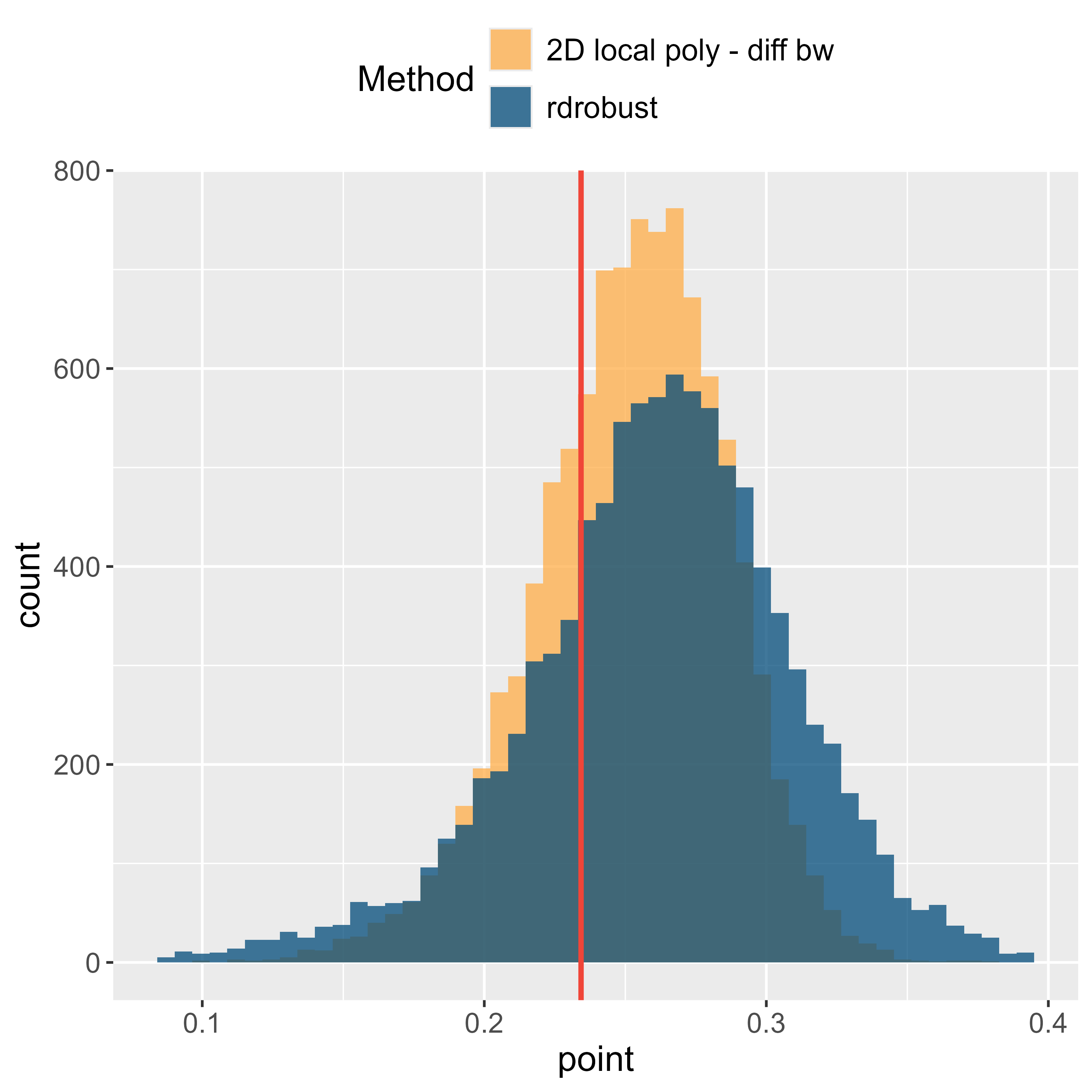}
    \centering
    \end{minipage}
    \begin{minipage}{0.49\hsize}
    (d) Design 4
    \includegraphics[width=\columnwidth]{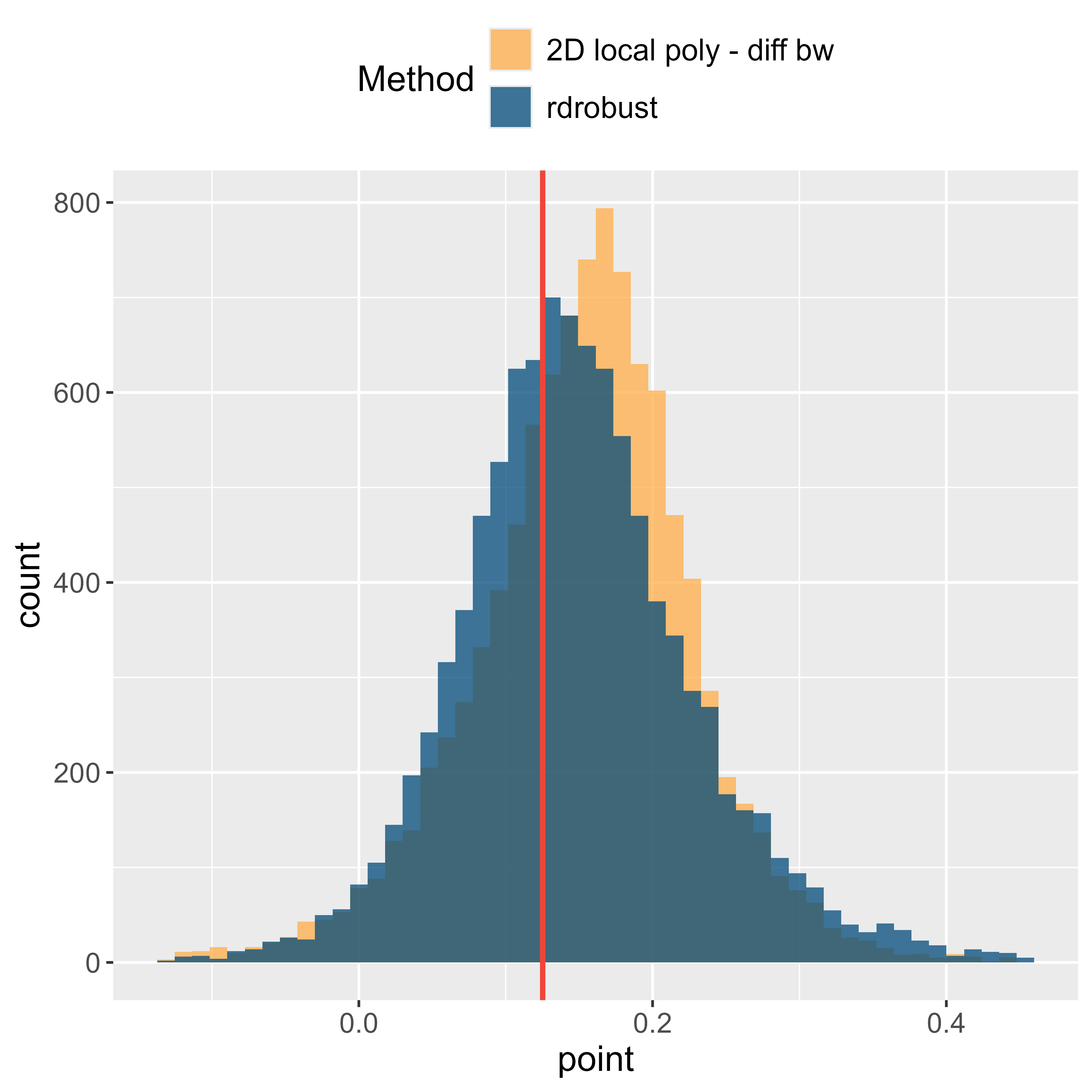}
    \centering
    \end{minipage}
    \caption{Histograms of point estimates with trimming of 1\% tail realizations. Light-colored distributions are of our estimator; dark-colored distributions are of the \textit{rdrobust}.}
    \label{fig:shape_2d}
\end{figure}

\begin{table}[H]
\centering
\caption{Simulation Results For Four Designs.}
\input{figures/rd2dim_revision/table_simulation_5000.tex}

\label{tab:MC_result}
\vspace{0.5cm}
\begin{minipage}{0.9\hsize}\footnotesize
 \textit{Notes:} Results are from $10,000$ replication draws of $5,000$ observation samples. \textit{rdrobust} is the estimator with the Euclidean distance from the boundary point as the running variable using \textit{rdrobust}; \textit{2D local poly} refers to our preferred different bandwidth estimator \textit{diff bw} and with imposing common bandwidth \textit{common bw}. All the implementations are in \textit{R}. \textit{length} and \textit{coverage} are of generated confidence interval length and coverage rate.
\end{minipage}
\end{table}

We report the detailed results in Table \ref{tab:MC_result}. Our first observation is that estimation with heterogeneous bandwidths $h_1 \neq h_2$ matters. The \textit{common bw} estimator is a version of our \textit{2D local poly} estimator that imposes $h_1 = h_2$. For all designs, our preferred \textit{2D local poly - diff bw} has smaller or approximately equal bias than \textit{common bw}. The better bias correction with heterogeneous bandwidths selection appears to induce smaller root MSEs for most designs while our \textit{2D local poly - diff bw} estimator is stable and maintaining the coverage rates above $95\%$ in all four designs.

Greater differences appear in comparison of our preferred estimator with \textit{rdrobust}. The RMSE of our estimator is smaller than that of the \textit{rdrobust} for all designs. In particular, the RMSE is less than the half of the RMSE in \textit{rdrobust} estimates for the Designs 1 and 2. Furthermore, the confidence intervals of our estimator are also shorter than that of the \textit{rdrobust} for most designs. Hence, our estimates are more efficient than the \textit{rdrobust} estimates and the efficiency conveys its greater performance in the inferences. Interestingly, the bias can be smaller in \textit{rdrobust} than in our estimator while its RMSE is always greater than in our estimator and their coverages are always below $95\%$. This result of the \textit{rdrobust} estimator is consistent with our earlier methodological analyses. The \textit{rdrobust} estimator chooses its bandwidth as if it is a univariate design; hence, their bandwidth selector chooses a suboptimal bandwidth which overly reduce bias relative to variance. \footnote{We report the summary statistics of the bandwidths used in Table \ref{tab:MC_result_bands}. We also conduct a parallel simulation study with a binary response via a linear probability model of the same polynomial in Table \ref{tab:MC_result_LPM} and \ref{tab:MC_result_LPM_bands}.}

Finally, we conduct a parallel exercise across all $30$ points. \footnote{Note that the underlying sampling supports are different from the earlier simulation results for the four points. Unlike the four points, which are relatively center in the support, some of the $30$ boundary points are outside of the originally constructed rectangular supports for the four designs.} Table \ref{tab:MC_result_summary_all} summarizes the performance comparisons across $30$ points. Except for an extreme behavior that appears in the max among $30$ points, our estimator performs favorably relative to \textit{rdrobust}. \footnote{The low performing one is at point 1 where all three estimators are poorly performed. We realize that the largest boundary points (points 1 and 30) are too extreme. Hence, we exclude the extreme points from the boundary points to evaluate in the empirical analysis. See Online Appendix \ref{sec.standardized_plots} Table \ref{tab:MC_result_all_30_points} for the results of all $30$ points.}

\begin{table}[H]
\centering
\addtocounter{table}{-1}
\caption{Summary of Simulation Results At All 30 Points.}
\input{figures/rd2dim_revision/table_simulation_rmse_5000_all_points} \vspace{-0.2cm}
\input{figures/rd2dim_revision/table_simulation_coverage_5000_all_points}\vspace{-0.2cm}
\input{figures/rd2dim_revision/table_simulation_length_5000_all_points}\vspace{-0.2cm}
\input{figures/rd2dim_revision/table_simulation_bias_5000_all_points}

\label{tab:MC_result_summary_all}
\vspace{0.5cm}
\begin{minipage}{0.9\hsize}\footnotesize
 \textit{Notes:} Results are from $5,000$ replication draws of $5,000$ observation samples. Four tables report rmse, coverage, length, and bias results summarized across 30 simulations points. Each column report minimum (min) of $30$ results, 25\%-tile among $30$ results, median of $30$ results, 75\%-tile of $30$ results and max of $30$ results, respectively. All the implementations are in \textit{R}.
\end{minipage}
\end{table}

\subsection{Applications} \label{sec.application}

We first illustrate our estimator through an empirical application of a Colombian scholarship, \citet*{Londono-Velez.Rodriguez.Sanchez2020, londono_data_2020}. From 2014 to 2018, the Colombian government operated a large-scale scholarship program called Ser Pilo Paga (SPP). The scholarship loan covers ``the full tuition cost of attending \textit{any} four-year or five-year undergraduate program in \textit{any} government-certified `high-quality' university in Colombia'' \citep[pp.194]{Londono-Velez.Rodriguez.Sanchez2020}. The eligibility of the SPP program is based on two thresholds. The first threshold is merit-based, determined by the nationally standardized high school graduation exam, SABER 11. In 2014 of \cite{Londono-Velez.Rodriguez.Sanchez2020}'s study period, the cutoff was the top 9\% of the score distribution. The second threshold is need-based, and is determined by the eligibility of the social welfare program, SISBEN. SISBEN-eligible families are roughly the poorest 50 percent. \footnote{Students must be also accepted by an eligible college in Colombia to receive the scholarship. Hence, the impact of exceeding both thresholds is not the impact of the program itself owing to noncompliance. The estimand is the impact of the program eligibility, which is the intention-to-treat effect.} The sample consists of $347,673$ observations of the control units and $15,423$ observations of the treated units.

The \textit{aggregation} approach is the empirical strategy of \cite{Londono-Velez.Rodriguez.Sanchez2020}. They run \textit{rdrobust} separately for two boundaries: the merit-based criterion (SAVER11) and the need-based criterion (SISBEN) as in Figure \ref{fig:score_map}. They report the effect of exceeding the merit-based (SABER11) threshold on enrollment in any eligible college is $0.32$ with a standard error of $0.012$ for the need-based (SISBEN) eligible subsample, and the effect of exceeding the need-based (SISBEN) threshold on enrollment in any eligible college is $0.274$ with a standard error of $0.027$ for the merit-based (SABER11) eligible subsample. Students with the need eligibility in the $x$-axis boundary of Figure \ref{fig:score_map} have a slightly higher effect than students with the merit eligibility in the $y$-axis boundary of Figure \ref{fig:score_map}. Indeed, their strategy captures certain heterogeneity in the sub-populations, albeit with richer heterogeneity within. The SISBEN threshold students are heterogeneous in their SABER11 scores; the SABER11 threshold students are heterogeneous in their SISBEN scores. 

We estimate the heterogeneous effects over the entire boundary. We summarize our results in Figure \ref{fig:empirical_results} with Panel (a) of the SABER = 0 boundary and Panel (b) of the SISBEN = 0 boundary. The dark-colored intervals are the pointwise $95\%$ confidence intervals from our estimates at each boundary point value, and the
 light-colored intervals are the pointwise $95\%$ confidence intervals from the \textit{rdrobust} estimates. For most points, the two estimates show similar patterns across the boundary points with a notable difference in the length of the confidence intervals. Our estimates exhibit shorter confidence intervals than \textit{rdrobust} when there are enough neighboring observations around the boundary points (such as SISBEN values from 2 through 24 in (a) and SABER values from 7 through 21 in (b)). On the other hand, our confidence intervals widen when there are only a few neighboring observations around the boundary points (such as SISBEN values at 44 and 48 in (a) and SABER values from 70 or more in (b)). Hence, our estimates are more stable for various designs and efficient at least when the effective sample size is large enough.

Both estimates suggest substantial heterogeneity in the effects among the merit-eligible students (Panel (b)) but not among the need-eligible students (Panel (a)). Specifically, the program has similar effects among the majority of students, but has no impact on extremely capable students. The null effect for extremely capable students is reasonable because they would have received other scholarships to attend college anyway. Consequently, the program could have benefited from accepting a larger number of students with higher household incomes because their impact is expected to 
be similar.
 
\begin{figure}[H]
    \centering
    \subfigure[Boundary at SABER = 0]{
        \includegraphics[width=0.47\hsize]{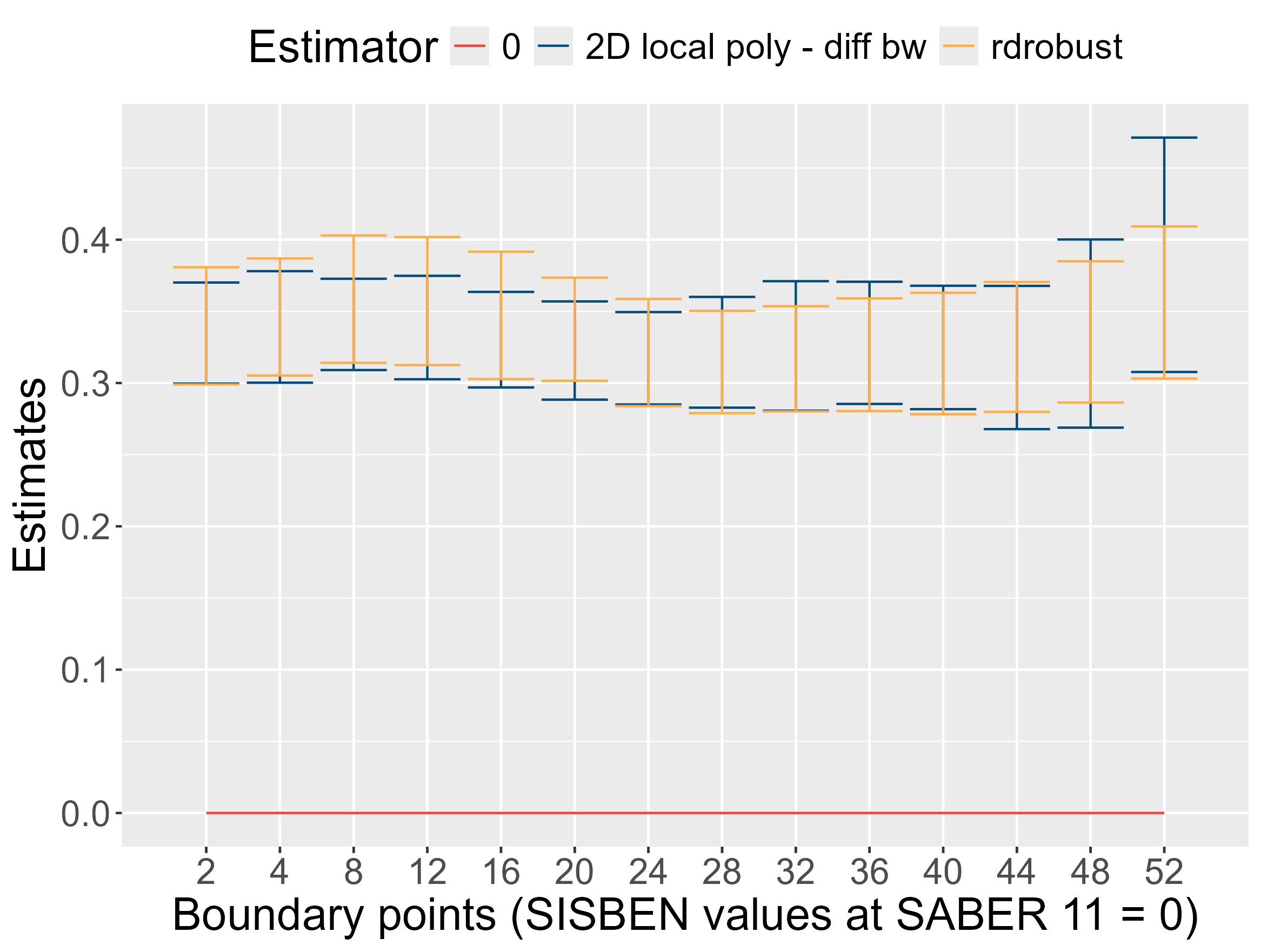}
    }
    \subfigure[Boundary at SISBEN = 0]{
        \includegraphics[width=0.47\hsize]{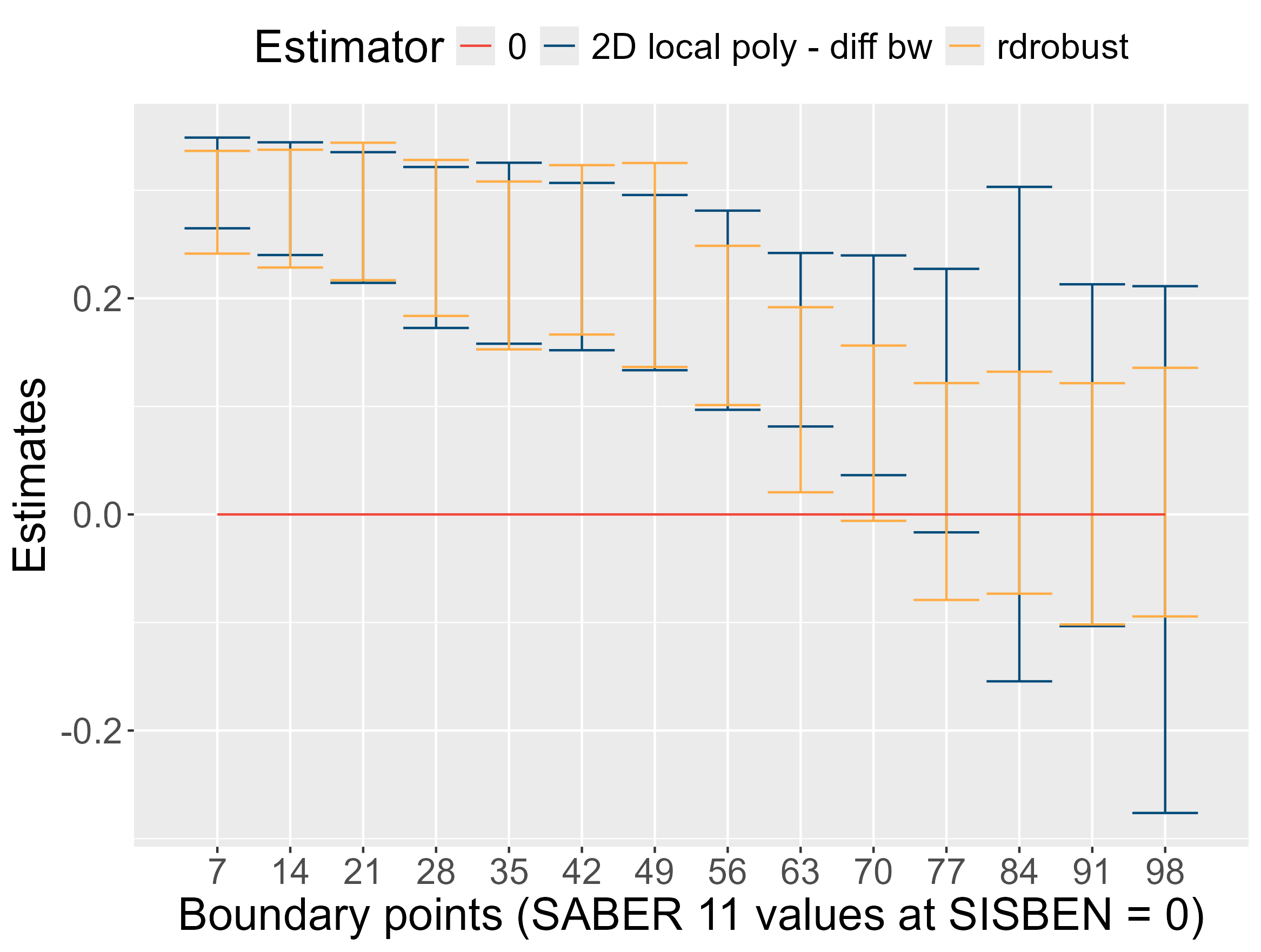}
    }
    \caption{$95\%$ confidence intervals over the boundary points. Dark-colored ranges are of our estimates. Light-colored ranges are of \textit{rdrobust} estimates. The Left panel (a) is for exceeding the merit threshold among the need-eligible students; the right panel (b) is for exceeding the need threshold among the merit-eligible students.}
    \label{fig:empirical_results}
\end{figure}

We further assess the stability of our estimates relative to \textit{rdrobust} by changing the scalings of the two running variables. Figure \ref{fig:empirical_results_scaled_distance} compares estimates with and without scaling by the absolute maximum values of each axis. Compared with the Panel (a) and (c) which exhibit substantial changes in the estimated confidence intervals of \textit{rdrobust}, our estimates in Panel (b) and (d) show the stability in the underlying (relative) scale of the running variables. An appropriate relative scaling of the two axes is hardly known. Hence, our approach is superior in handling the relative scaling of the two-dimensional data as is because our estimator is more robust against the choice of scaling.

\begin{figure}[H]
    \centering
    \subfigure[SABER = 0, rdrobust with and without scaling]{
        \includegraphics[width=0.47\hsize]{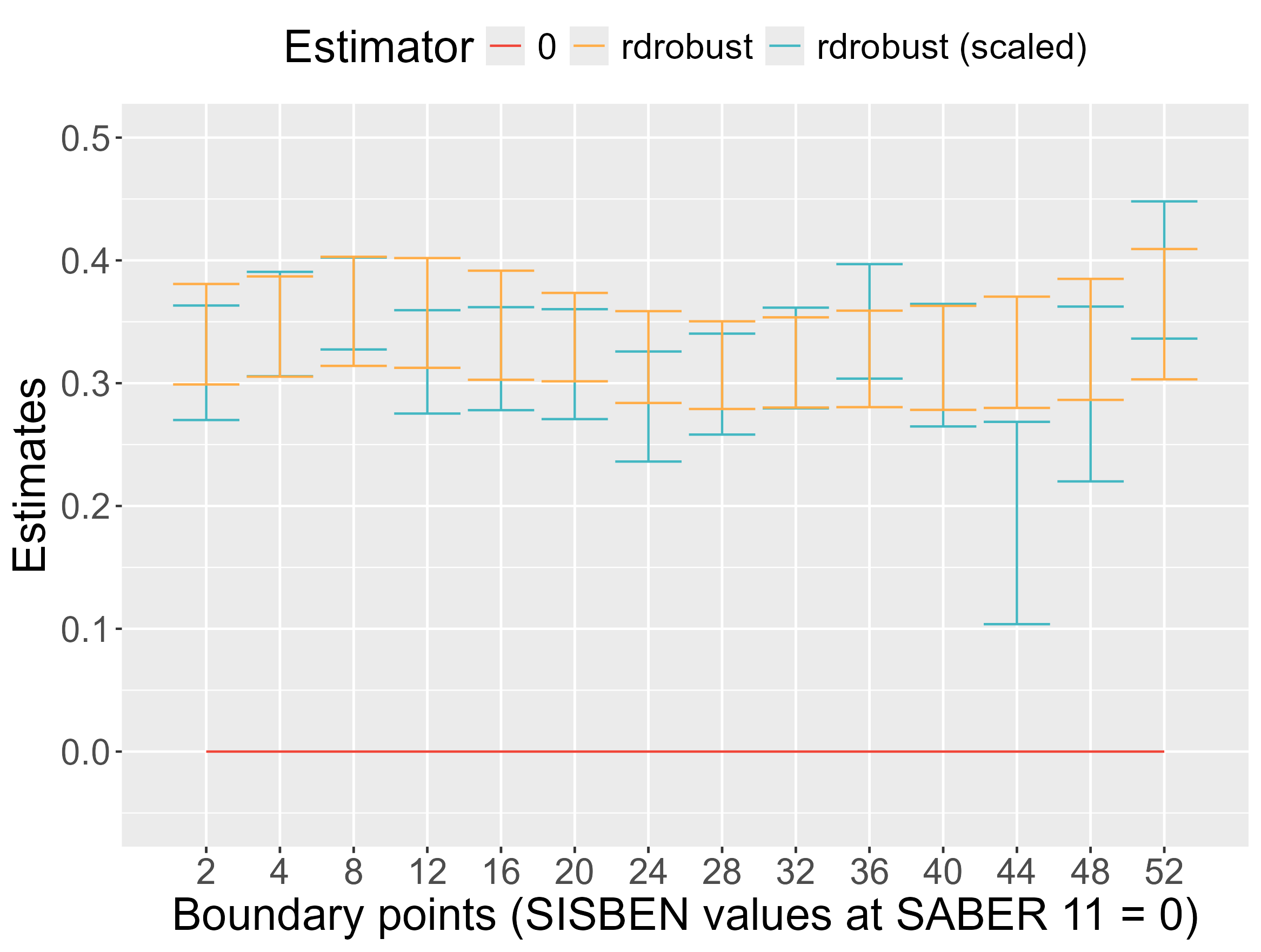}
    }
    \subfigure[SABER = 0, ours with and without scaling]{
        \includegraphics[width=0.47\hsize]{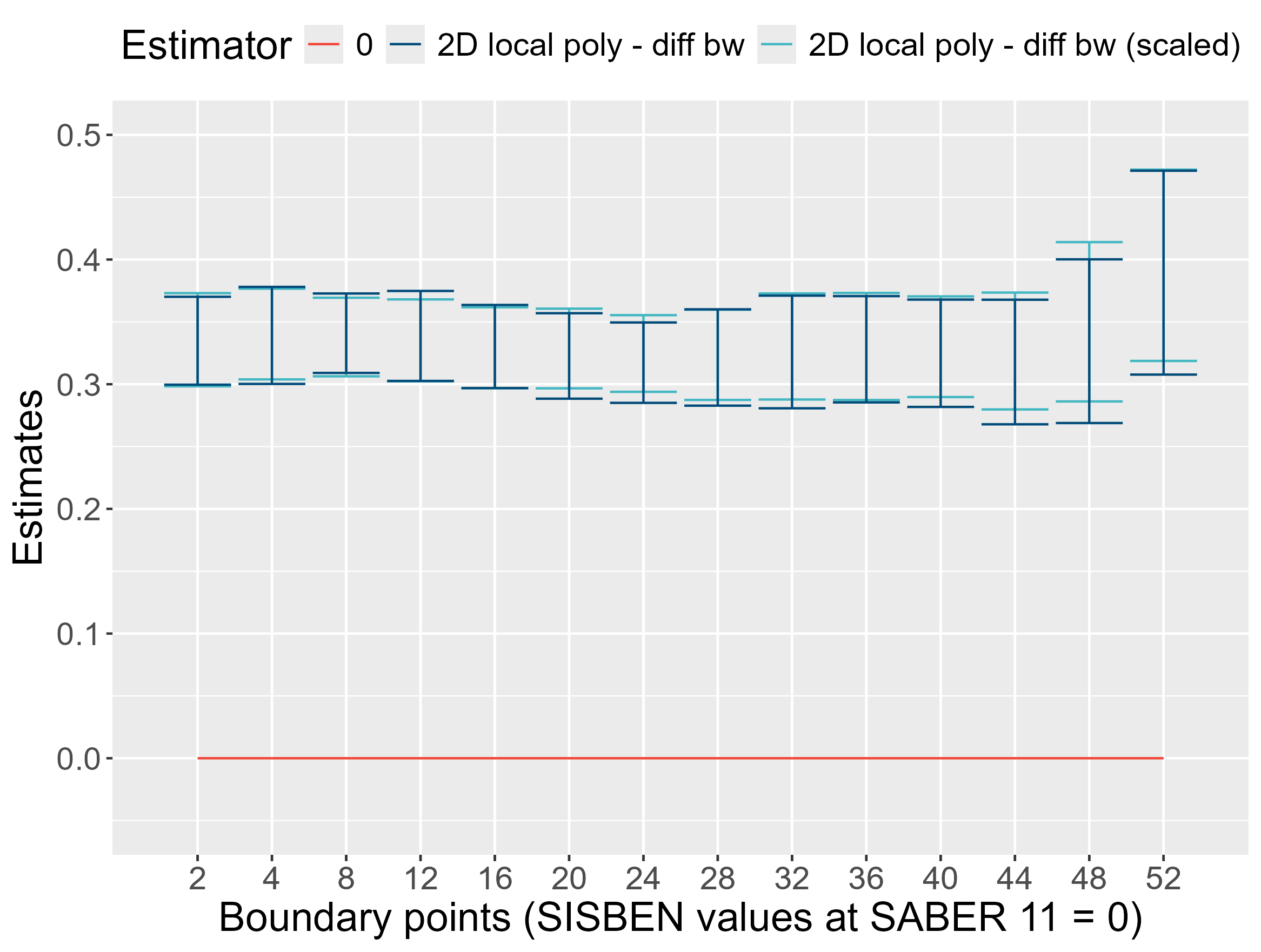}
    }
    \subfigure[SISBEN = 0, rdrobust with and without scaling]{
        \includegraphics[width=0.47\hsize]{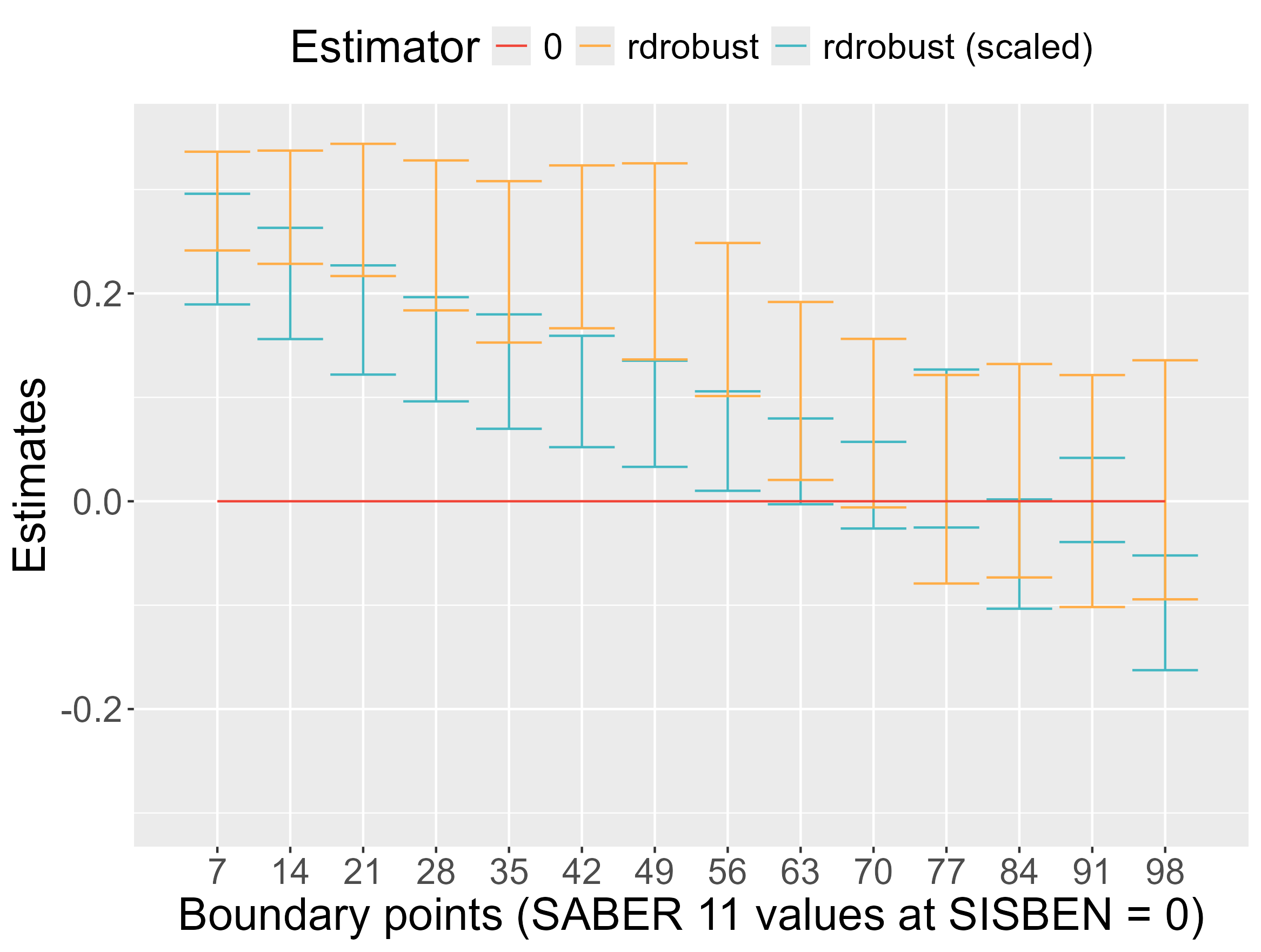}
    }
    \subfigure[SISBEN = 0, ours with and without scaling]{
        \includegraphics[width=0.47\hsize]{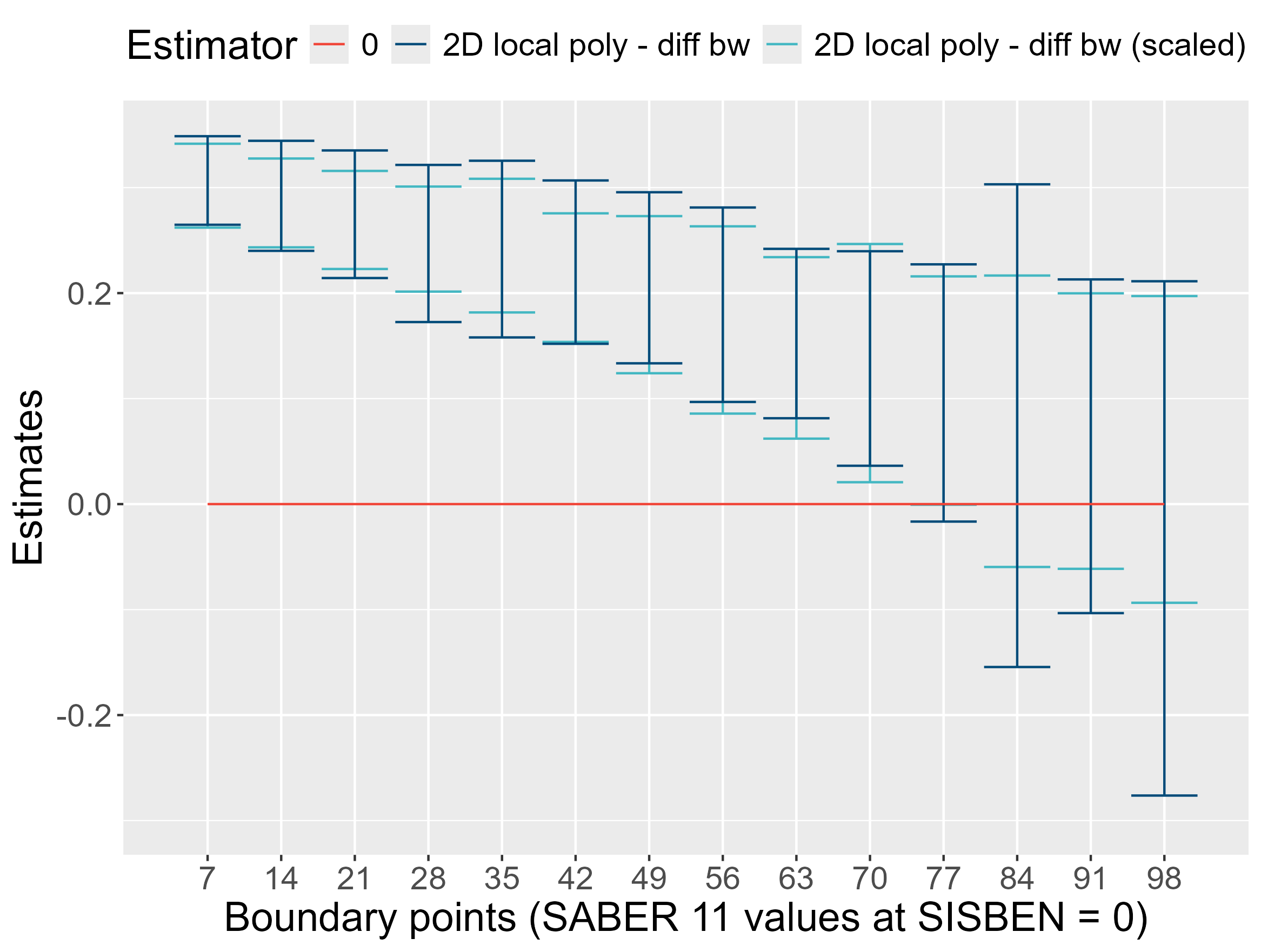}
    }
    \caption{Estimation results over the 28 boundary points comparing two \textit{rdrobust} estimates with and without normalizing scales by their maximum values for each two axes (Panel (a) and (c)) and our estimator (Panel (b) and (d)).}
    \label{fig:empirical_results_scaled_distance}
\end{figure}


We further apply our procedure to the dataset used in \cite{Lee2008} (also in  \citeauthor{caugheyDataCitation}, \citeyear{caugheyDataCitation}, \citeyear{Caughey.Sekhon2011}) that studies the U.S. House Elections and finds the positive significant incumbent margin. There are a few baseline covariates with continuous variations as reported in \cite{Caughey.Sekhon2011}. We use four baseline covariates: percentages of black voters, foreign born voters, government worker voters, and of urban areas for each electoral district. See their scatter plots and evaluation points in Online Appendix \ref{sec.standardized_plots} Figure \ref{fig:shapes_lee_data}. Among the four covariates, three covariate designs exhibit shorter confidence intervals of our estimates relative to \textit{rdrobust}. The confidence intervals were larger among the Government Worker Percentage design, however, our estimates capture more distinct heterogeneity that the higher government worker is related to higher incumbent margin.

\begin{figure}[H]
    \centering 
    \begin{minipage}{0.49\hsize}
    (a) Black Percentage
    \includegraphics[width=\columnwidth]{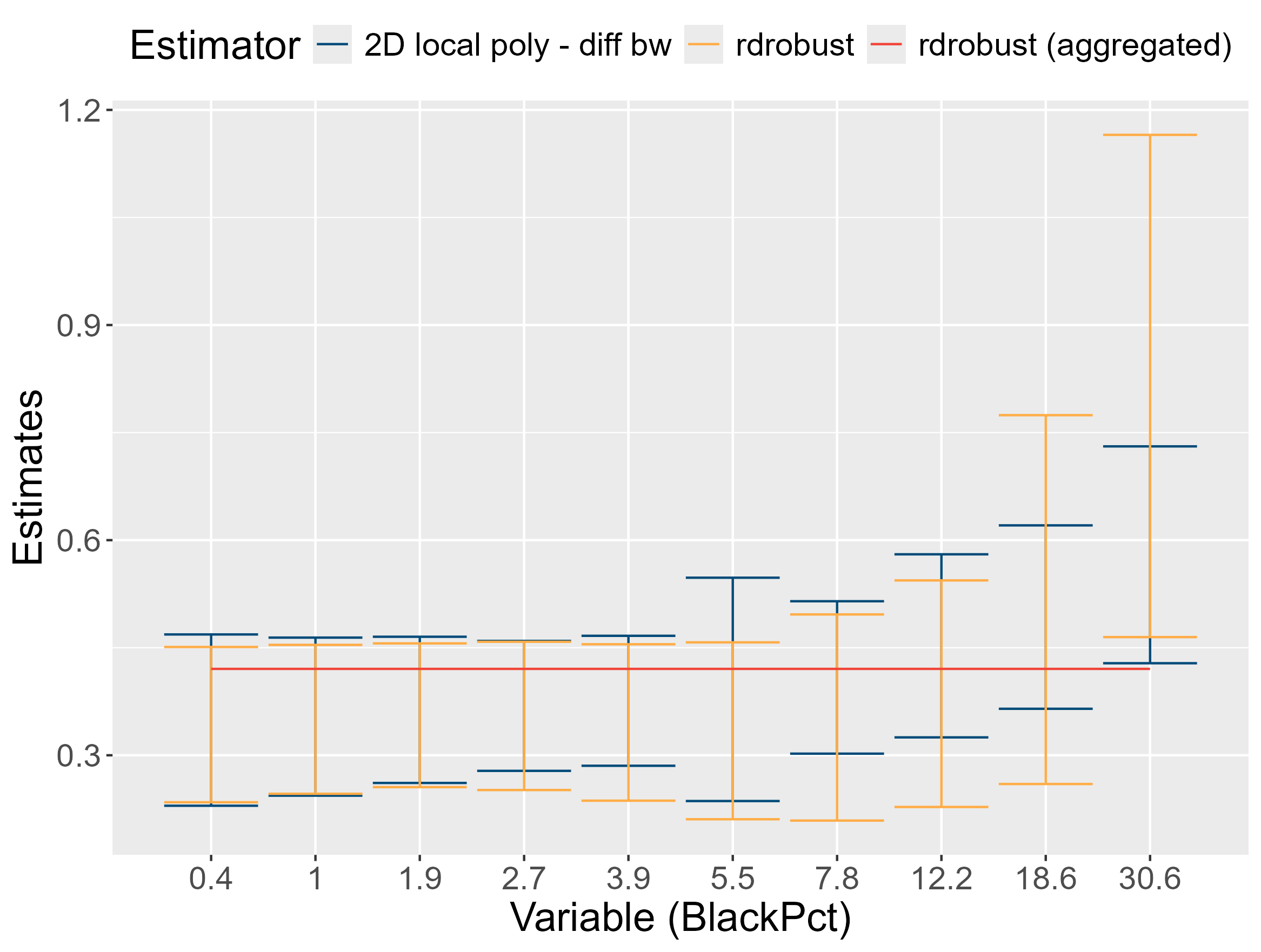}
    \centering 
    \end{minipage}    
    \begin{minipage}{0.49\hsize}
    (b) Foreign Born Percentage
    \includegraphics[width=\columnwidth]{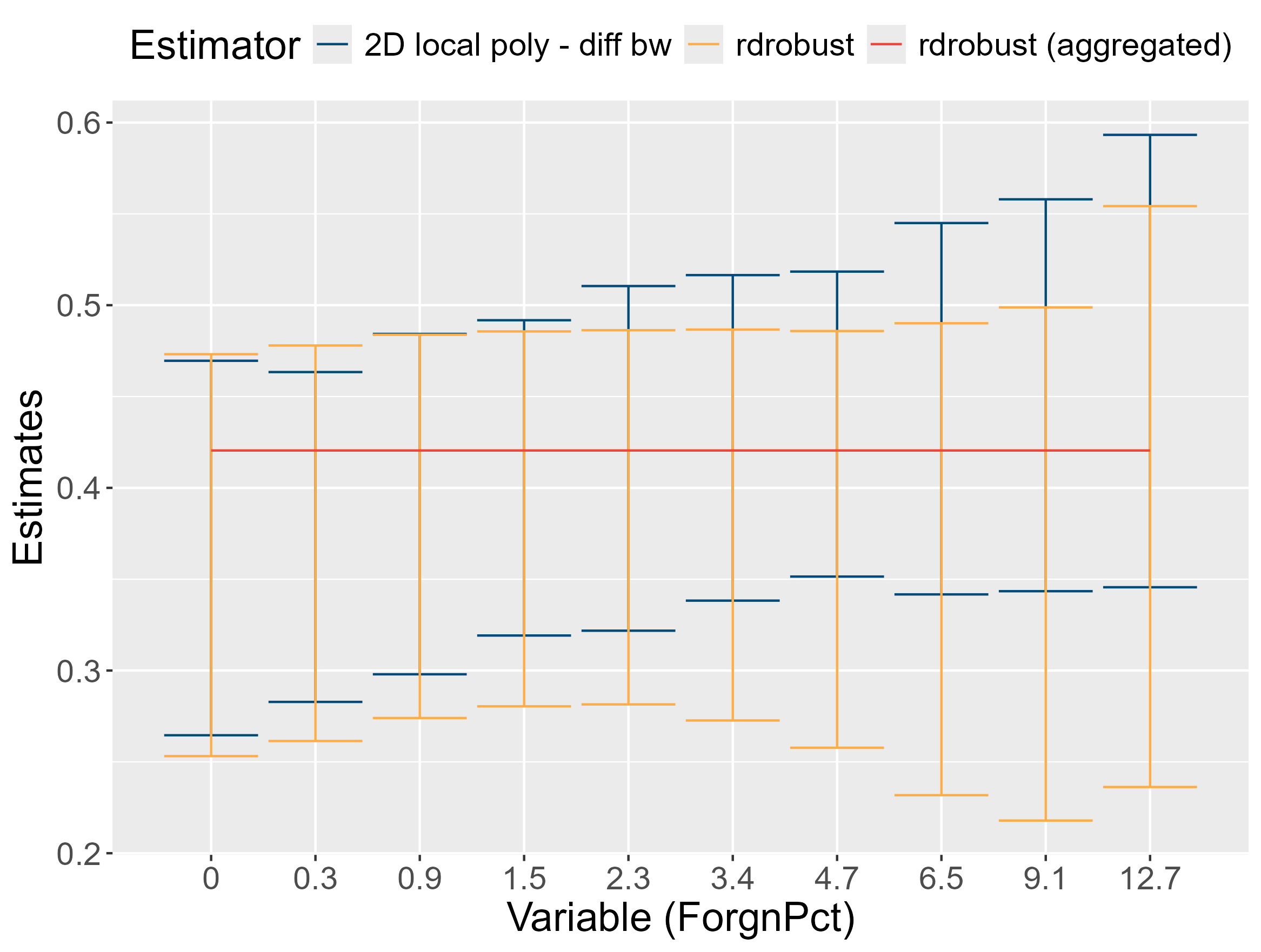}
    \centering
    \end{minipage}
    \begin{minipage}{0.49\hsize}
    (c) Government Worker Percentage
    \includegraphics[width=\columnwidth]{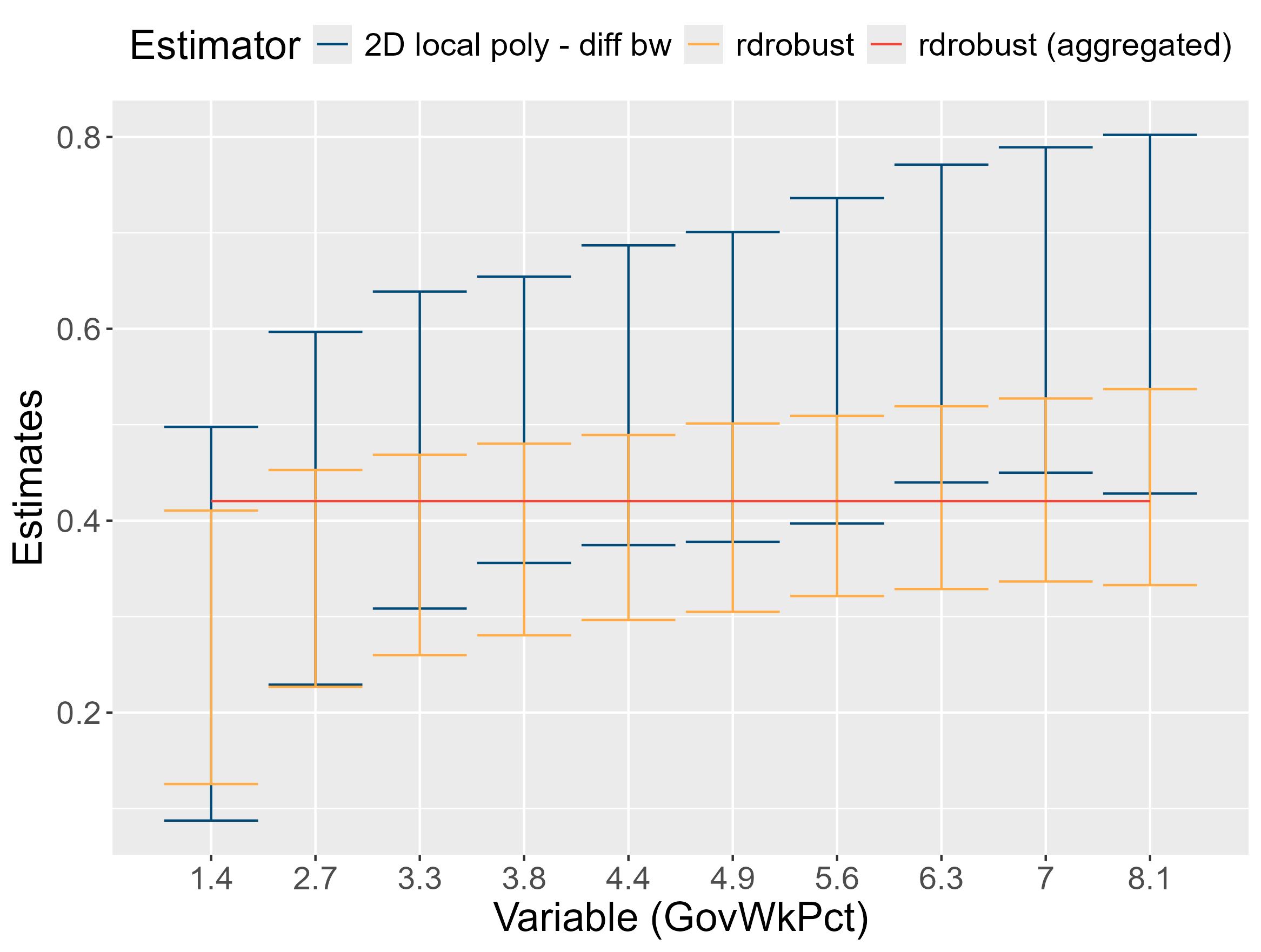}
    \centering 
    \end{minipage}
    \begin{minipage}{0.49\hsize}
    (d) Urban Percentage
    \includegraphics[width=\columnwidth]{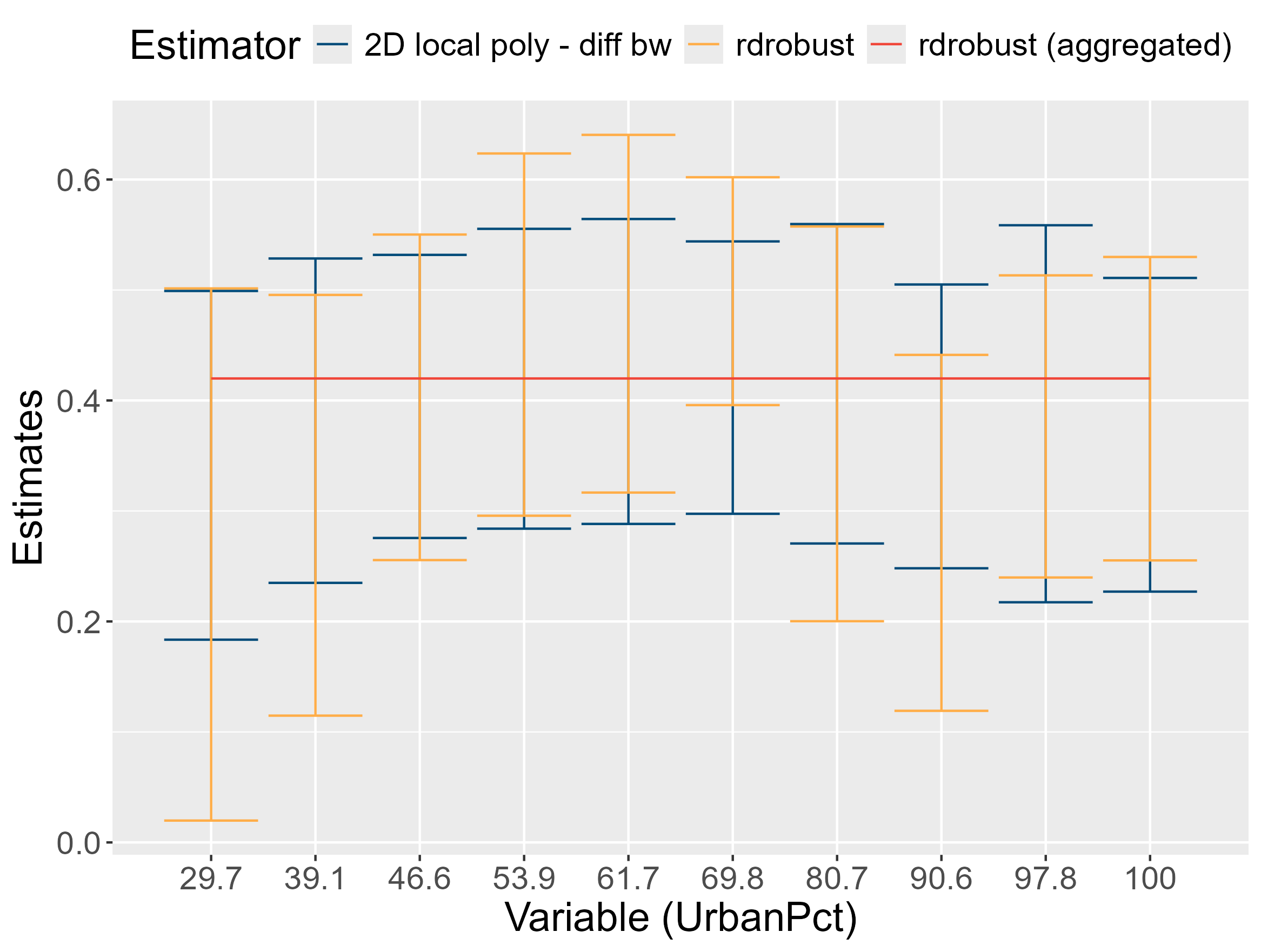}
    \centering 
    \end{minipage}
    \caption{Heterogeneous incumbent margin estimates across different covariate values. The red straight lines in \textit{rdrobust (aggregated)} represents the original univariate estimate.}
    \label{fig:estimates_lee_data}
\end{figure}

\section{Conclusion}

We document that the existing bandwidth selectors are suboptimal when they are used for a multivariate RD design when they take the distance from a boundary point as the running variable. We further provide an alternative estimator for a multivariate RD design to estimate the heterogeneous treatment effects. In numerical simulations, we demonstrate the favorable performance of our estimator against a frequently used \textit{rdrobust} procedure with the distance from a point as the scalar running variable. We apply our estimator to the study of \cite{Londono-Velez.Rodriguez.Sanchez2020} who study the impact of a scholarship program that has two eligibility requirements and a quasi-multivariate design for \cite{Lee2008} dataset with a baseline covariate to study the heterogeneous effects across the covariate values. In these application, our estimates are consistent with the original estimates, often produce shorter confidence intervals, and reveal a richer heterogeneity in the program impacts over the policy boundary than the original estimates.

Hence, we contribute to the RD estimation literature in two ways. We provide a detailed argument that the distance approach is suboptimal for a multivariate design and we provide a remedy for the problem with a dimension-specific bandwidths selector. Combined with the recent work by \cite{cattaneoEstimationInferenceBoundary2025} which documents another problem of the distance approach for designs with a corner or kink and provides an alternative estimator with a uniform inference, we provide the reason why the distance from a boundary point should not be used for a multivariate RD design to estimate heterogeneous effects across the boundary as well as an appropriate estimator to remedy the estimation problem.


Some theoretical and practical issues remain. First, our consideration is limited to a random sample; hence, spatial RD designs are excluded from our consideration. We defer our focus to spatial design because of its theoretical and conceptual complexity. Nevertheless, we aim to propose a spatial RD estimation based on newly developed asymptotic results of \cite{Kurisu.Matsuda2022} in a separated study. Second, our theoretical results can be applied to any finite-dimensional RD design; however, the practical performance of estimators with more than two dimensions is limited. Third, our approach requires a sufficiently large sample over the boundary, and its performance with an extremely small sample size is limited. For a smaller sample, an explicit randomization approach is a compelling alternative. \citet*{Cattaneo.Frandsen.Titiunik2015}, \citet*{Cattaneo.Titiunik.Vazquez-Bare2016} and \citet*{Cattaneo.Titiunik.Vazquez-Bare2017} propose the concepts and a randomization inference. Their approach requires a substantially stronger assumption but is applicable to a geographical RD design as well \citep*{Keele.Titiunik.Zubizarreta2015}. Fourth, covariates are often incorporated in the estimation procedures in RD designs. For the efficiency gain, \cite{Frolich.Huber2019} propose a method with a multi-dimensional non-parametric estimation; \cite{Calonico.Cattaneo.Farrell.Titiunik2019} develop an easy-to-implement augmentation; and recently \citet*{Noack.Olma.Rothe2021} considers flexible and efficient estimation including machine-learning devices and several studies such as \citet*{Kreiss.Rothe2021} and \citet*{Arai.Otsu.Seo2021} explore augmentation with high-dimensional covariates. We defer these analyses to theoretical and conceptual complications for a companion study for a geographic RD design. Fifth, we provided the optimal bandwidths for multivariate RD estimation; however, the optimal kernel for this class of estimators is unknown. Exploring the optimal kernel for a multivariate estimator is a topic for future research. Finally, we do not provide any procedure to aggregate heterogeneous estimates over the set of boundary points. For example, a major feature of the rdmulti package, \citet*{Cattaneo.Titiunik.Vazquez-Bare2020}, is averaging over multiple boundary points;  \citet*{Cattaneo.Keele.Titiunik.Vazquez-Bare2016} offers a target pooling parameter; and \citet*{Cattaneo.Keele.Titiunik.Vazquez-Bare2021} uses a different policy in Columbia with multiple cutoffs to extrapolate the missing part of the support. These ideas can be a benchmark to consider averaging and extrapolation when the support has \textit{holes} in the boundary. 

\bibliography{bibtex_hardcopied_20260107} 

\appendix

\section{The heterogeneous bandwidths case}\label{sec.hetero_band}

We follow the standard bandwidth selection procedure in RD designs to find the pair of $(h_1,h_2)$ that minimizes the asymptotic expansion of the MSE that we derive from Theorem \ref{thm: LP-CLT-main-d2}. Specifically, we derive the asymptotic expansion of the mean-squared error of \textcolor{black}{$\hat{m}_+(c)$}: 
\textcolor{black}{
\begin{align*}
&\underbrace{
 \left[e'_1 S_+^{-1} B^{(2,1)} \begin{pmatrix}
 \partial_{11}m_+(c)h_1^2/2\\
 \partial_{12}m_+(c)h_1h_2\\
 \partial_{22}m_+(c)h_2^2/2 
 \end{pmatrix}\right]^2}_{\text{Bias term}}   + \underbrace{\frac{1}{nh_1h_2}\frac{\sigma_+^2(c)}{f(c)} e'_1 S_+^{-1} \mathcal{K}_+S_+^{-1}e_1}_{\text{Variance term}},
\end{align*}
}
for $e_1 = (1,0,0)'$. However, this general expression is too complex to obtain an analytical formula for the optimal bandwidths when all coefficients of the partial derivatives $\partial_{11}m_+(c), \partial_{12}m_+(c)$ and $\partial_{22}m_+(c)$ in the bias term are non-zero. We simplify this expression by choosing the kernels as follows:
\textcolor{black}{
\begin{equation}
\kappa_{1,\pm}^{(1,1)} = \kappa_{1,2,\pm}^{(1,1,1)}=\kappa_{1,\pm}^{(1,2)}=\kappa_{1,2,\pm}^{(1,1,2)}=\kappa_{1,2,\pm}^{(1,2,1)}=0. \label{eq.kernel_restriction}
\end{equation}
where $\kappa_{j_1,\dots, j_M,\pm}^{(r_1,\dots,r_M,v)}:= \int \prod_{\ell=1}^{M}z_{j_\ell}^{r_\ell}K_\pm^v(z)dz$ for integer $v$.}
Among the product kernels of the form \textcolor{black}{$K_\pm( z_1,z_2) = K_1( z_1) K_2(\pm z_2)$}, the above restriction amounts to \textit{rotate} the space so that the boundary becomes either the $x$ or $y$-axis. Among the product kernels, the following kernels satisfy the above restrictions:
\begin{align*}
K_1(z) &= 
\begin{cases}
(1-|z|)1_{\{|z|\leq 1\}} & (\text{two-sided triangular kernel}),\\
{3 \over 4}(1-z^2)1_{\{|z|\leq 1\}} & (\text{Epanechnikov kernel}),
\end{cases}\\
K_2(z) &= 2(1-|z|)1_{\{0\leq z \leq 1\}}\ (\text{one-sided triangular kernel}).
\end{align*}
The same restriction is possible without a product kernel. For example, a cone kernel \textcolor{black}{$K_\pm(z_1,z_2) = K(z_1,\pm z_2)$ where} $
K(z_1,z_2) = {6 \over \pi}\left(1-\|z\|\right)1_{\{\|z\| \leq 1, z_2\geq 0\}}.
$ for $z = (z_1,z_2)$ and $\|z\| = \sqrt{z_1^2 + z_2^2}$ satisfy \eqref{eq.kernel_restriction}. Given the flexibility of the kernel choice, this simplifying restriction for the admissible kernel is innocuous.

In the subsequent analysis, we assume that $K_1$ is the two-sided triangular kernel and $K_2$ is the one-sided triangular kernel. For example, the design with $\mathcal{T} = \{(R_1,R_2) \in \mathbb{R}^2: R_1 \geq c_1, R_2 \geq c_2 \}$ satisfies the restriction \eqref{eq.kernel_restriction} as is or with a $90$ degrees rotation; the design with $\mathcal{T} = \{(R_1,R_2) \in \mathbb{R}^2: R_1 + R_2 \geq c_1 + c_2 \}$ satisfies the restriction \eqref{eq.kernel_restriction} with a $45$ degrees rotation. \textcolor{black}{Define $\kappa_{0,\pm}^{(v)}:= \int K_\pm^v(z)dz$.} Under \eqref{eq.kernel_restriction}, $\text{MSE}(\hat{m}_+(c))$ is simplified as follows
\textcolor{black}{
\begin{align*}
\Bigl\{\Bigr. &{h_1^2 \over 2}\partial_{11}m_+(c)\left(\tilde{s}_{1,+} \kappa_{1,+}^{(2,1)} + \tilde{s}_{3,+} \kappa_{1,2,+}^{(2,1,1)}\right) + {h_2^2 \over 2}\partial_{22}m_+(c)\left(\tilde{s}_{1,+} \kappa_{2,+}^{(2,1)} + \tilde{s}_{3,+} \kappa_{2,+}^{(3,1)}\right)\Bigl.\Bigr\}^2\\
\quad +& {\sigma_+^2(c)  \over f(c)nh_1h_2} {\kappa_{0,+}^{(2)} \left(\kappa_{1,+}^{(2,1)}\kappa_{2,+}^{(2,1)}\right)^2 - 2\kappa_{2,+}^{(1,2)}\left(\kappa_{1,+}^{(2,1)}\right)^2\kappa_{2,+}^{(2,1)}\kappa_{2,+}^{(1,1)} + \kappa_{1,+}^{(2,2)}\left(\kappa_{1,+}^{(2,1)}\kappa_{2,+}^{(1,1)}\right)^2 \over \left(\kappa_{0,+}^{(1)}\kappa_{1,+}^{(2,1)}\kappa_{2,+}^{(2,1)} - \left(\kappa_{2,+}^{(1,1)}\right)^2 \kappa_{2,+}^{(2,1)}\right)^2}
\end{align*}
}
where 
\textcolor{black}{
\begin{align*}
\begin{pmatrix}
 \tilde{s}_{1,+}\\ \tilde{s}_{2,+} \\\tilde{s}_{3,+}
\end{pmatrix}
 &:= S_+^{-1}e_1={1 \over \kappa_{0,+}^{(1)}\kappa_{1,+}^{(2,1)}\kappa_{2,+}^{(2,1)} - \left(\kappa_{2,+}^{(1,1)}\right)^2 \kappa_{2,+}^{(2,1)}}
 \begin{pmatrix}
\kappa_{1,+}^{(2,1)}\kappa_{2,+}^{(2,1)}\\
0\\
-\kappa_{1,+}^{(2,1)}\kappa_{2,+}^{(1,1)} 
 \end{pmatrix}.
\end{align*}
One can also see that $(\tilde{s}_{1,-},\tilde{s}_{2,-},\tilde{s}_{3,-})':=S_-^{-1}e_1=(\tilde{s}_{1,+},\tilde{s}_{2,+},-\tilde{s}_{3,+})'$, $\kappa_{1,2,-}^{(2,1,1)}=-\kappa_{1,2,+}^{(2,1,1)}$, $\kappa_{2,-}^{(3,1)}=-\kappa_{2,+}^{(3,1)}$, and $e'_1S_+^{-1}\mathcal{K}_+S_+^{-1}e_1=e'_1S_-^{-1}\mathcal{K}_-S_-^{-1}e_1$.}

\textcolor{black}{To simplify the notation, define 
$e'_1S^{-1}\mathcal{K}S^{-1}e_1 \equiv e'_1S_\pm^{-1}\mathcal{K}_\pm S_\pm^{-1}e_1$,$ \tilde{s}_1 \kappa_1^{(2,1)} \equiv \tilde{s}_{1,\pm} \kappa_{1,\pm}^{(2,1)},$
 $\tilde{s}_{3} \kappa_{2}^{(3,1)} \equiv \tilde{s}_{3,\pm} \kappa_{2,\pm}^{(3,1)},$ and $
\tilde{s}_{3} \kappa_{2}^{(3,1)} \equiv  \tilde{s}_{3,\pm} \kappa_{2,\pm}^{(3,1)}.
$} The MSE of the estimator $\hat{m}_+(c) - \hat{m}_+(c)$ is
\textcolor{black}{
\begin{align*}
\Biggl\{&{h_1^2 \over 2}\left(\partial_{11}m_+(c) -\partial_{11}m_-(c)\right) \left(\tilde{s}_1 \kappa_1^{(2,1)} + \tilde{s}_3 \kappa_{1,2}^{(2,1,1)}\right)\\
&+ {h_2^2 \over 2}\left(\partial_{22}m_+(c) - \partial_{22}m_-(c)\right)\left(\tilde{s}_{1} \kappa_{2}^{(2,1)} + \tilde{s}_{3} \kappa_{2}^{(3,1)}\right)\Biggr\}^2 + {(\sigma_+^2(c) + \sigma_-^2(c)) \over f(c)nh_1h_2}e'_1 S^{-1} \mathcal{K} S^{-1}e_1
\end{align*}
}
when the same kernels are used for both the treatment and control sides.

We consider the optimal pair of bandwidths $(h_1,h_2)$ that minimizes the above asymptotic MSE. In minimizing asymptotic MSE, the bias term may disappear when the second derivatives of the treatment and control mean functions are equal. Nevertheless, the second derivatives match exactly only in an extreme scenario. Following \cite{Imbens.Kalyanaraman2012}, we assume the second derivatives $\partial_{11}m_+(c)$ and $\partial_{11}m_{-}(c)$ as well as $\partial_{22}m_{+}(c)$ and $\partial_{22}m_{-}(c)$ are different. Under the following restrictions
\textcolor{black}{
$
\partial_{11}m_+(c) \neq \partial_{11}m_-(c), \partial_{22}m_+(c) \neq \partial_{22}m_-(c)$, and $sgn\left\{(\partial_{11}m_+(c) - \partial_{11}m_-(c)) \tilde{s}_{11}\right\}
= sgn\left\{(\partial_{22}m_+(c) - \partial_{22}m_-(c))\tilde{s}_{22}\right\}$
where $\tilde{s}_{11} \equiv \tilde{s}_{1} \kappa_{1}^{(2,1)} + \tilde{s}_{3} \kappa_{1,2}^{(2,1,1)}$ and $\tilde{s}_{22} \equiv \tilde{s}_{1} \kappa_{2}^{(2,1)} + \tilde{s}_{3} \kappa_{2}^{(3,1)}$,}
the pair of optimal bandwidths \footnote{These bandwidths are not optimal when the signs of the bias terms differ. A similar issue arises in the single-variable RD estimation with heterogeneous bandwidths with the treatment and control mean functions (\citealp*{Imbens.Kalyanaraman2012}). \cite{Arai.Ichimura2018} derive the higher-order expansion of the bias terms for the single-variable RD estimation. In Online Appendix \ref{sec.asymptotic_higher_order}, we derive the higher-order expansion of the bias terms. Nevertheless, we do not follow \cite{Arai.Ichimura2018}'s approach because estimating higher-order bias correction terms is unreliable for multivariate RD estimations.} is
\textcolor{black}{
\begin{align*}
 \frac{h_1}{h_2} = \sqrt{\frac{B_2(c)}{B_1(c)}} &\mbox{ and } h_1^6 = \frac{(\sigma^2_+(c) + \sigma_-^2(c))}{2n} e'_1 S^{-1}\mathcal{K}S^{-1}e_1 (B_1^{-5/2}(c) B_2^{1/2}(c))\\
 \text{where } B_1(c) =& \sqrt{\left\{(\partial_{11}m_+(c) - \partial_{11}m_-(c)) \left(\tilde{s}_{1} \kappa_{1}^{(2,1)} + \tilde{s}_{3} \kappa_{1,2}^{(2,1,1)}\right)\right\}^2}, \mbox{ and } \\
 B_2(c) =& \sqrt{\left\{(\partial_{22}m_+(c) - \partial_{22}m_-(c))\left(\tilde{s}_{1} \kappa_{2}^{(2,1)} + \tilde{s}_{3} \kappa_{2}^{(3,1)}\right)\right\}^2}.    
\end{align*}
}
Furthermore, we follow \cite{Imbens.Kalyanaraman2012}'s regularization approach to prevent the bandwidths from blowing up when the bias terms are zero by following \cite{Calonico.Cattaneo.Titiunik2014}'s approach to estimate the variances of the bias term $B_1(c)$ and $B_2(c)$ estimations and the variances are added to the bias term estimates which appear in the denominator for the optimal bandwidth formulas so that the denominator will not equal to zero even if the bias terms are zero. We also follow \cite{Calonico.Cattaneo.Titiunik2014} for a bias correction to obtain appropriate inference. We propose a plug-in bias correction with a multivariate local-quadratic estimation. See Online Appendix \ref{sec.implementation} for further implementation details.
\newpage

\setcounter{page}{1}

\centering \Huge 
\textbf{Online Appendices}

\Large
\textbf{for Local-Polynomial Estimation for Multivariate Regression Discontinuity Designs.}

\normalsize
\flushleft


\section{Asymptotic Theory for multivariate Local-Polynomial Regressions} \label{sec.asymptotic}

\textcolor{black}{Throughout this section, we write the kernel function $K_+$ as $K$ for simplicity.}
\subsection{Local-polynomial estimator}\label{sec.asymptotic_estimator}
Consider the following nonparametric regression model: 
\[
Y_i = m(R_i) + \varepsilon_i,\ E[\varepsilon_i|R_i] = 0,\ i=1,\dots, n,
\]
where $\{(Y_i,R_i)\}_{i=1}^{n}$ is a sequence of i.i.d. random vectors such that $Y_i \in \mathbb{R}$, $R_i=(R_{i,1},\dots, R_{i,d})' \in \mathbb{R}^{d}$. 

Define 
\begin{align*}
D &= \#\{(j_1,\dots ,j_L): 1 \leq j_1 \leq \dots \leq j_L \leq d, 0 \leq L \leq p\},\\ 
\bar{D} &=\#\{(j_1,\dots,j_{p+1}): 1 \leq j_1 \leq \dots \leq j_{p+1} \leq d\},
\end{align*}  
and $(s_{j_1\dots j_L1},\dots, s_{j_1\dots j_Ld}) \in \mathbb{Z}_{\geq 0}^d$ such that $s_{j_1\dots j_Lk} = \#\{j_\ell : j_\ell = k, 1\leq \ell \leq L\}$. Further, define $
\bm{s}_{j_1\dots j_L}!=s_{j_1\dots j_L1}!\dots s_{j_1\dots j_Ld}!$.
When $L = 0$, we set $(j_1,\dots,j_L) = j_0 = 0$, $\bm{s}_{j_1\dots j_L}! = 1$. Note that $\sum_{j=1}^{d}s_{j_1\dots j_L \ell} = L$. 
The local-polynomial estimator 
\begin{align*}
\hat{\beta}(r) &= (\hat{\beta}_{j_1,\dots j_L}(r))'_{1\leq j_1\leq \dots \leq j_L\leq d, 0 \leq L \leq p}\\
&:= (\hat{\beta}_0(r), \hat{\beta}_1(r),\dots, \hat{\beta}_d(r),\hat{\beta}_{11}(r),\dots \hat{\beta}_{dd}(r),\dots,\hat{\beta}_{1\dots1}(r),\dots,\hat{\beta}_{d\dots d}(r))'.
\end{align*}
of 
\begin{align*}
M(r) &=\left({1 \over \bm{s}_{j_1\dots j_L}!}\partial_{j_1,\dots j_L}m(r)\right)'_{1\leq j_1 \leq \dots \leq j_L\leq d, 0 \leq L \leq p}\\
&:= \left(m(r),\partial_1 m(r),\dots, \partial_d m(r), {\partial_{11}m(r) \over 2!}, {\partial_{12}m(r) \over 1!1!}, \dots, {\partial_{dd}m(r) \over 2!}, \right. \\ 
&\left. \quad \quad \dots, {\partial_{1\dots1}m(r) \over p!}, {\partial_{1\dots 2}m(r) \over (p-1)!1! }\dots, {\partial_{d\dots d}m(r) \over p!} \right)' 
\end{align*} 
is given as a solution of the following problem:  
\begin{align}\label{LL-minimize}
\hat{\beta}(r) &= \argmin_{\beta \in \mathbb{R}^D}\sum_{i=1}^{n}\left(Y_i - \sum_{L=0}^{p}\sum_{1 \leq j_1 \leq \dots \leq j_L \leq d}\beta_{j_1\dots j_L}\prod_{\ell=1}^{L}(R_{i,j_\ell} - r_{j_\ell})\right)^2K_h\left(R_i - r\right)
\end{align}
where $\beta = (\beta_{j_1\dots j_L})'_{1 \leq j_1 \leq \dots \leq j_L \leq d, 0 \leq L \leq p}$, 
\[
K_{h}(R_i - r) = K\left({R_{i,1} - r_i \over h_1}, \dots, {R_{i,d} - r_d \over h_d}\right)
\]
and each $h_j$ is a sequence of positive constants (bandwidths) such that $h_j \to 0$ as $n \to \infty$. For notational convenience, we interpret $\sum_{1 \leq j_1 \leq \dots \leq j_L \leq d}\beta_{j_1\dots j_L}\prod_{\ell=1}^{L}(R_{i,j_\ell} - r_{j_\ell}) = \beta_0$ when $L=0$. We introduce some notations: 
\begin{align*}
\Y &:=
\left(
\begin{array}{c}
Y_1  \\
\vdots \\
Y_n
\end{array}
\right),\ \W := \diag\left(K_h\left(R_1 - r \right),\dots,K_h\left(R_n - r\right)\right),\\
\R &:= (\R_1,\dots,\R_n) = 
\left(
\begin{array}{ccc}
1 & \cdots & 1  \\
\left(R_{1} - r\right)_1 & \cdots &\left(R_{n} - r \right)_1\\
\vdots & \dots & \vdots \\
\left(R_{1} - r\right)_p & \cdots &\left(R_{n} - r\right)_p
\end{array}
\right)= 
\left(
\begin{array}{ccc}
1 & \dots & 1\\
\check{\R}_1 &\dots & \check{\R}_n
\end{array}
\right),
\end{align*}
where 
\begin{align*}
\left(R_{i} - r\right)_L &= \left(\prod_{\ell = 1}^{L}(R_{i,j_{\ell}} - r_{j_{\ell}})\right)'_{1 \leq j_1 \leq \dots \leq j_L \leq d}.
\end{align*}
The minimization problem (\ref{LL-minimize}) can be rewritten as
\begin{align*}
\hat{\beta}(r) &= \argmin_{\beta \in \mathbb{R}^D}(\Y - \R'\beta)'\W(\Y - \R'\beta) = \argmin_{\beta \in \mathbb{R}^D}Q_n(\beta).
\end{align*}
Then the first order condition of the problem (\ref{LL-minimize}) is given by
\begin{align*}
{\partial \over \partial \beta}Q_n(\beta) &= -2\R \W \Y + 2\R \W \R'\beta = 0.
\end{align*}
Hence the solution of the problem (\ref{LL-minimize}) is given by
\begin{align*}
\hat{\beta}(r) &= (\R \W \R')^{-1}\R \W \Y\\
&= \left[\sum_{i=1}^{n}K_h\left(R_i - r \right)\R_i \R'_i\right]^{-1}\sum_{i=1}^{n}K_h\left(R_i - r\right)\R_i Y_i.
\end{align*}

Define
\[
H := \diag(1,h_1,\dots, h_d,h_1^2, h_1h_2, \dots,h_d^2,\dots, h_1^p, h_1^{p-1}h_2,\dots, h_d^p) \in \mathbb{R}^{D \times D}.
\]

\begin{theorem}[Asymptotic normality of local-polynomial estimators]\label{thm: LP-CLT}
Under Assumptions \ref{Ass1}, \ref{Ass2}, \ref{Ass3} and \ref{Ass4}, as $n \to \infty$, we have
\begin{align*}
&\sqrt{nh_1 \cdots h_d}\left(H\left(\hat{\beta}(r) - M(r)\right) - S^{-1}B^{(d,p)}M_n^{(d,p)}(r)\right)\\ 
&\stackrel{d}{\to} N\left( \left(
\begin{array}{c}
0 \\
\vdots \\
0
\end{array}
\right), {\sigma^2(r) \over f(r)}S^{-1}\mathcal{K}S^{-1}\right), 
\end{align*}
where 
\begin{align*}
M_n^{(d,p)}(r) &= \left({\partial_{j_1\dots j_{p+1}}m(r) \over \bm{s}_{j_1 \dots j_{p+1}}!}\prod_{\ell=1}^{p+1}h_{j_{\ell}}\right)'_{1 \leq j_1 \leq \dots \leq j_{p+1}\leq d}\\
&= \left({\partial_{1\dots 1}m(r) \over (p+1)!}h_1^{p+1},{\partial_{1\dots 2}m(r) \over p!}h_1^{p}h_2,\dots, {\partial_{d\dots d}m(r) \over (p+1)!}h_d^{p+1}\right)' \in \mathbb{R}^{\bar{D}}, 
\end{align*}
\begin{align*}
B^{(d,p)} &= \int \left(
\begin{array}{c}
1 \\
\check{\bm{z}}
\end{array}
\right)
(\bm{z})'_{p+1}d\bm{z} \in \mathbb{R}^{D \times \bar{D}},\ \mathcal{K}= \int K^2(\bm{z})\left(
\begin{array}{c}
1 \\
\check{\bm{z}}
\end{array}
\right)
(1\ \check{\bm{z}}')d\bm{z}.
\end{align*}
\end{theorem}

\begin{proof}
\textcolor{black}{Define $\kappa_0^{(v)}:= \int K_+^v(z)dz$, $\kappa_{j_1,\dots, j_M}^{(v)}:= \int \prod_{\ell=1}^{M}z_{j_\ell}K_+^v(z)dz$ for integer $v$.} We also define $h:= (h_1,\dots, h_d)'$ and for $r, y \in \mathbb{R}^{d}$, let $r \circ y = (r_1y_1 ,\cdots, r_dy_d)'$ be the Hadamard product. Considering Taylor's expansion of $m(r)$ around $r=(r_1,\dots,r_d)'$, 
\begin{align*}
m(R_i) &= (1, \check{\R}'_i)M(r) + {1 \over (p+1)!}\sum_{1 \leq j_1 \leq \dots \leq j_{p+1}\leq d}{(p+1)! \over \bm{s}_{j_1\dots j_{p+1}}!}\partial_{j_1,\dots,j_{p+1}}m(\tilde{R}_i)\\
&\quad \times \prod_{\ell=1}^{p+1}(R_{i,j_\ell} - r_{j_\ell}),
\end{align*}
where $\tilde{R}_i = r + \theta_i(R_i - r)$ for some $\theta_i \in [0,1)$. Then we have
\begin{align*}
&\hat{\beta}(r) - M(r) \\
&= (\R \W \R')^{-1}\R \W (\Y - \R'M(r))\\
&= \left[\sum_{i=1}^{n}K_h\left(R_i - r\right)
\left(
\begin{array}{c}
1 \\
\check{\R}_i 
\end{array}
\right)
(1\ \check{\R}'_i)\right]^{-1}
\sum_{i=1}^{n}K_h\left(R_i - r\right)
\left(
\begin{array}{c}
1 \\
\check{\R}_i
\end{array}
\right)\\
&\quad \times \left(\varepsilon_i + \sum_{1 \leq j_1 \leq \dots \leq j_{p+1}\leq d}{1 \over \bm{s}_{j_1\dots j_{p+1}}!}\partial_{j_1,\dots,j_{p+1}}m(\tilde{R}_i)\prod_{\ell=1}^{p+1}(R_{i,j_\ell} - r_{j_\ell})\right).
\end{align*}
This yields
\begin{align*}
\sqrt{nh_1 \cdots h_d}H(\hat{\beta}(r) - M(r)) &= S_n^{-1}(V_n(r) + B_n(r)),
\end{align*}
where
\begin{align*}
S_n(r) &= {1 \over nh_1 \cdots h_d}\sum_{i=1}^{n}K_h\left(R_i - r\right)H^{-1}
\left(
\begin{array}{c}
1 \\
\check{\R}_i
\end{array}
\right)
(1\ \check{\R}'_i)H^{-1},\\
V_n(r) &= {1 \over \sqrt{nh_1 \cdots h_d}}\sum_{i=1}^{n}K_h\left(R_i - r\right)H^{-1}
\left(
\begin{array}{c}
1 \\
\check{\R}_i
\end{array}
\right)\varepsilon_i\\
&=: (V_{n,j_1\dots j_L}(r))'_{1 \leq j_1\leq \dots \leq j_L \leq d, 0 \leq L \leq p},\\
B_n(r) &= {1 \over \sqrt{nh_1 \cdots h_d}}\sum_{i=1}^{n}K_h\left(R_i - r\right)H^{-1}
\left(
\begin{array}{c}
1 \\
\check{\R}_i
\end{array}
\right)\\
&\quad \times \sum_{1 \leq j_1 \leq \dots \leq j_{p+1}\leq d}{1 \over \bm{s}_{j_1\dots j_{p+1}}!}\partial_{j_1,\dots,j_{p+1}}m(\tilde{R}_i)\prod_{\ell=1}^{p+1}(R_{i,j_\ell} - r_{j_\ell})\\ 
&=: (B_{n,j_1\dots j_L}(\tilde{R}_i))'_{1 \leq j_1\leq \dots \leq j_L \leq d, 0 \leq L \leq p}.
\end{align*}

\noindent
(Step 1) Now we evaluate $S_n(r)$. For $1 \leq j_{1,1} \leq \dots \leq j_{1,L_1}\leq d,1\leq j_{2,1}\leq \dots \leq j_{2,L_2} \leq d, 0 \leq L_1,L_2 \leq p$, we define
\begin{align*}
&I_{n,j_{1,1}\dots j_{1,L_1},j_{2,1}\dots j_{2,L_2}}\\ 
&:= {1 \over nh_1 \cdots h_d}\sum_{i=1}^{n}K_h\left(R_i - r \right)\prod_{\ell_1=1}^{L_1}\left({R_{i,j_{\ell_1}} - r_{j_{\ell_1}} \over h_{j_{\ell_1}}}\right)\prod_{\ell_2=1}^{L_2}\left({R_{i,j_{\ell_2}} - r_{j_{\ell_2}} \over h_{j_{\ell_2}}}\right). 
\end{align*}
Observe that 
\begin{align*}
&E\left[I_{n,j_{1,1}\dots j_{1,L_1},j_{2,1}\dots j_{2,L_2}} \right] \\
&= {1 \over h_1 \cdots h_d}E\left[K_h\left(R_i - r \right)\prod_{\ell_1=1}^{L_1}\left({R_{i,j_{\ell_1}} - r_{j_{\ell_1}} \over h_{j_{\ell_1}}}\right)\prod_{\ell_2=1}^{L_2}\left({R_{i,j_{\ell_2}} - r_{j_{\ell_2}} \over h_{j_{\ell_2}}}\right)\right]\\
&= \int \left(\prod_{\ell_1=1}^{L_1}z_{j_{\ell_1}}\right)\left(\prod_{\ell_2=1}^{L_2}z_{j_{\ell_2}}\right)K(z)f(r + h \circ z)dz\\
&= f(r)\kappa_{j_{1,1}\dots j_{1,L_1}j_{2,1}\dots j_{2,L_2}}^{(1)} + o(1).
\end{align*}
For the last equation, we used the dominated convergence theorem. 
\begin{align*}
&\Var(I_{n,j_{1,1}\dots j_{1,L_1},j_{2,1}\dots j_{2,L_2}})\\ 
&\quad = {1 \over n(h_1 \cdots h_d)^2}\Var\left(K_h\left(R_1 - r \right)\prod_{\ell_1=1}^{L_1}\left({R_{i,j_{\ell_1}} - r_{j_{\ell_1}} \over h_{j_{\ell_1}}}\right)\prod_{\ell_2=1}^{L_2}\left({R_{i,j_{\ell_2}} - r_{j_{\ell_2}} \over h_{j_{\ell_2}}}\right)\right)\\
&\quad = {1 \over nh_1 \cdots h_d}\left\{\int \prod_{\ell_1=1}^{L_1}\left({R_{i,j_{\ell_1}} - r_{j_{\ell_1}} \over h_{j_{\ell_1}}}\right)^2\prod_{\ell_2=1}^{L_2}\left({R_{i,j_{\ell_2}} - r_{j_{\ell_2}} \over h_{j_{\ell_2}}}\right)^2K^2(z)f(r + h \circ z)dz \right. \\ 
&\left. \quad \quad - h_1 \cdots h_d\left(\int\prod_{\ell_1=1}^{L_1}\left({R_{i,j_{\ell_1}} - r_{j_{\ell_1}} \over h_{j_{\ell_1}}}\right)\prod_{\ell_2=1}^{L_2}\left({R_{i,j_{\ell_2}} - r_{j_{\ell_2}} \over h_{j_{\ell_2}}}\right)K(z)f(r + h \circ z)dz\right)^2\right\}\\
&\quad = {1 \over nh_1 \cdots h_d}\left(f(r)\kappa_{j_{1,1}\dots j_{1,L_1}j_{2,1}\dots j_{2,L_2}j_{1,1}\dots j_{1,L_1}j_{2,1}\dots j_{2,L_2}}^{(2)} + o(1)\right)\\
&\quad \quad   - {1 \over n}(f(r)\kappa_{j_{1,1}\dots j_{1,L_1}j_{2,1}\dots j_{2,L_2}}^{(1)} + o(1))^2\\
&\quad = {f(r)\kappa_{j_{1,1}\dots j_{1,L_1}j_{2,1}\dots j_{2,L_2}j_{1,1}\dots j_{1,L_1}j_{2,1}\dots j_{2,L_2}}^{(2)} \over nh_1 \cdots h_d} + o\left({1 \over nh_1 \cdots h_d}\right).
\end{align*}
For the third equation, we used the dominated convergence theorem. Then for any $\rho>0$, 
\begin{align*}
&P\left(|I_{n,j_{1,1}\dots j_{1,L_1},j_{2,1}\dots j_{2,L_2}} - f(r)\kappa_{j_{1,1}\dots j_{1,L_1}j_{2,1}\dots j_{2,L_2}}^{(1)}|>\rho \right)\\ 
&\leq \rho^{-1}\left\{\Var(I_{n,j_{1,1}\dots j_{1,L_1},j_{2,1}\dots j_{2,L_2}}) + \left(E[I_{n,j_{1,1}\dots j_{1,L_1},j_{2,1}\dots j_{2,L_2}}] - f(r)\kappa_{j_{1,1}\dots j_{1,L_1}j_{2,1}\dots j_{2,L_2}}^{(1)}\right)^2\right\}\\
&= O\left({1 \over nh_1 \cdots h_d}\right) + o(1) = o(1).
\end{align*}
This yields $I_{n,j_{1,1}\dots j_{1,L_1},j_{2,1}\dots j_{2,L_2}} \stackrel{p}{\to} f(r)\kappa_{j_{1,1}\dots j_{1,L_1}j_{2,1}\dots j_{2,L_2}}^{(1)}$. Hence we have $S_{n}(r) \stackrel{p}{\to} f(r)S$.

\noindent
(Step 2) Now we evaluate $V_n(r)$. For any $t = (t_0,t_1,\dots, t_d, t_{11},\dots,t_{dd},\dots,t_{1\dots1},\dots, t_{d\dots d})' \in \mathbb{R}^D$, we define
\begin{align*}
R_{n,i,j_1\dots j_L} &:= {1 \over \sqrt{nh_1 \cdots h_d}}K_h\left(R_i - r\right)\prod_{\ell=1}^{L}\left({R_{i,j_\ell} - r_{j_\ell} \over h_{j_\ell}}\right)\varepsilon_i,\ 1\leq j_1,\dots,j_L \leq d,\\
Z_{n,i} &:= \sum_{L=0}^{p}\sum_{1 \leq j_1 \leq \dots \leq j_L \leq d}t_{j_1\dots j_L}R_{n,i,j_1\dots j_L}.
\end{align*}
Observe that
\begin{align*}
\sigma_{n,j_1\dots j_L}^2 &:= \Var\left(\sum_{i=1}^{n}R_{n,i,j_1\dots j_L}\right) = {1 \over h_1 \cdots h_d}E\left[\varepsilon_i^2K_h^2\left(R_1 - r \right)\prod_{\ell=1}^{L}\left({R_{1,j_\ell} - r_{j_\ell} \over h_{j_\ell}}\right)^2\right]\\
&= {1 \over h_1 \cdots h_d}E\left[\sigma^2(R_i)K_h^2\left(R_1 - r \right)\prod_{\ell=1}^{L}\left({R_{1,j_\ell} - r_{j_\ell} \over h_{j_\ell}}\right)^2\right]\\ 
&= \int \sigma^2(r + h \circ z)\left(\prod_{\ell=1}^{L}z_{j_\ell}^2\right)K^2(z)f(r + h \circ z)dz\\
&= \sigma^2(r)f(r)\kappa_{j_1\dots j_L j_1 \dots j_L}^{(2)} + o(1).
\end{align*}
For the last equation, we used the dominated convergence theorem. Moreover, for $1 \leq j_{1,1} \leq \dots \leq j_{1,L_1} \leq d$ and $1 \leq j_{2,1} \leq \dots \leq j_{2,L_2} \leq d$, we have
\begin{align*}
&\Cov(V_{n,j_{1,1}\dots j_{1,L_1}}(r), V_{n,j_{2,1}\dots j_{2,L_2}}(r))\\ 
&=  {1 \over h_1 \cdots h_d}E\left[\sigma^2(R_i)K_h^2\left(R_i - r \right)\prod_{\ell_1=1}^{L_1}\left({R_{i,j_{1,\ell_1}} - r_{j_{1,\ell_1}} \over h_{j_{1,\ell_1}}}\right)\prod_{\ell_2=1}^{L_2}\left({R_{i,j_{2,\ell_2}} - r_{j_{2,\ell_2}} \over h_{j_{2,\ell_2}}}\right)\right]\\ 
&= \int \sigma^2(r + h \circ z)\left(\prod_{\ell_1=1}^{L_1}z_{j_{1,\ell_1}}\right)\left(\prod_{\ell_2=1}^{L_2}z_{j_{2,\ell_2}}\right) K^2(z)f(r + h \circ z)dz\\ 
&= \sigma^2(r)f(r)\kappa_{j_{1,1}\dots j_{1,L_1}j_{2,1}\dots j_{2,L_2}}^{(2)} + o(1).
\end{align*}
For the last equation, we used the dominated convergence theorem. For sufficiently large $n$, we have
\begin{align*}
&\sum_{i=1}^{n}E[|Z_{n,i}|^{2+\delta}] \\
&= {1 \over n^{\delta/2}(h_1 \cdots h_d)^{1 + \delta/2}}E\left[|\varepsilon_i|^{2+\delta}\left|K_h\left(R_i - r \right)\right|^{2+\delta} \right. \\
&\left. \quad \times \left|\sum_{L=0}^{p}\sum_{1 \leq j_1 \leq \dots \leq j_L \leq d}t_{j_1\dots j_L}\prod_{\ell=1}^{L}\left({R_{i,j_\ell} - r_{j_\ell} \over h_{j_\ell}}\right)\right|^{2+\delta}\right]\\
&\leq {U(r) \over (nh_1 \cdots h_d)^{\delta/2}}\int \left|\sum_{L=0}^{p}\sum_{1 \leq j_1 \leq \dots \leq j_L \leq d}t_{j_1\dots j_L}\prod_{\ell=1}^{L}z_{j_{\ell}}\right|^{2+\delta}|K(z)|^{2+\delta}f(r + h \circ z)dz\\
&= {U(r)f(r) \over (nh_1 \cdots h_d)^{\delta/2}}\int \left|\sum_{L=0}^{p}\sum_{1 \leq j_1 \leq \dots \leq j_L \leq d}t_{j_1\dots j_L}\prod_{\ell=1}^{L}z_{j_{\ell}}\right|^{2+\delta}|K(z)|^{2+\delta}dz + o(1)\\
&= o(1). 
\end{align*}
For the second equation, we used the dominated convergence theorem. Thus, Lyapounov's condition is satisfied for $\sum_{i=1}^{n}Z_{n,i}$. Therefore, by Cram\'er-Wold device, we have
\begin{align*}
V_n(r) &\stackrel{d}{\to} N\left( \left(
\begin{array}{c}
0 \\
\vdots \\
0
\end{array}
\right), \sigma^2(r)f(r)\mathcal{K}\right).
\end{align*}

\noindent
(Step 3) Now we evaluate $B_n(r)$. Decompose
\begin{align*}
B_{n,j_1\dots j_L}(\tilde{R}_i)&= \left\{B_{n,j_1\dots j_L}(\tilde{R}_i) - B_{n,j_1\dots j_L}(r) - E\left[B_{n,j_1\dots j_L}(\tilde{R}_i) - B_{n,j_1\dots j_L}(r)\right]\right\}\\
&\quad + E\left[B_{n,j_1\dots j_L}(\tilde{R}_i) - B_{n,j_1\dots j_L}(r)\right]\\
&\quad + \left\{B_{n,j_1\dots j_L}(r) - E\left[B_{n,j_1\dots j_L}(r)\right]\right\}\\
&\quad + E\left[B_{n,j_1\dots j_L}(r)\right]\\
&=: \sum_{\ell=1}^{4}B_{n,j_1\dots j_L\ell}. 
\end{align*}
Define $N_r(h):= \prod_{j=1}^{d}[r_j-C_Kh_j, r_j + C_Kh_j]$. For $B_{n,j_1\dots j_L1}$, 
\begin{align}\label{B_n01}
&\Var(B_{n,j_1\dots j_L1})  \nonumber \\ 
&\leq {1 \over \{(p+1)!\}^2h_1 \cdots h_d}E\left[K_h^2\left(R_i - r \right)\prod_{\ell=1}^{L}\left({R_{i,j_\ell} - r_{j_\ell} \over h_{j_\ell}}\right)^2 \right. \nonumber \\
&\left. \quad \times \sum_{1 \leq j_{1,1} \leq \dots \leq j_{1,p+1} \leq d, 1 \leq j_{2,1} \leq \dots \leq j_{2,p+1} \leq d}{1 \over \bm{s}_{j_{1,1} \dots j_{1,p+1}}! }{1 \over \bm{s}_{j_{2,1}\dots j_{2,p+1}}!} \right. \nonumber  \\
&\left. \quad \times (\partial_{j_{1,1}\dots j_{1,p+1}}m(\tilde{R}_i) - \partial_{j_{1,1}\dots j_{1,p+1}}m(r))(\partial_{j_{2,1}\dots j_{2,p+1}}m(\tilde{R}_i) - \partial_{j_{2,1}\dots j_{2,p+1}}m(r)) \right.  \nonumber \\
&\left. \quad  \times \prod_{\ell_1=1}^{p+1}(R_{i,j_{1,\ell_1}} - r_{j_{1,\ell_1}})\prod_{\ell_2=1}^{p+1}(R_{i,j_{2,\ell_2}}-r_{j_{2,\ell_2}})\right]  \nonumber \\ 
&\leq {1 \over \{(p+1)!\}^2}\max_{1 \leq j_1 \leq \dots \leq j_{p+1} \leq d}\sup_{y \in N_r(h)}|\partial_{j_1\dots j_{p+1}}m(y) - \partial_{j_1\dots j_{p+1}}m(r)|^2 \nonumber \\
&\quad \quad \times  \sum_{1 \leq j_{1,1} \leq \dots \leq j_{1,p+1} \leq d, 1 \leq j_{2,1} \leq \dots \leq j_{2,p+1} \leq d}\prod_{\ell_1=1}^{p+1}h_{j_{1,\ell_1}}\prod_{\ell_2=1}^{p+1}h_{j_{2,\ell_2}} \nonumber \\
&\quad \quad \times \int \left(\prod_{\ell=1}^{L}|z_{j_\ell}|\prod_{\ell_1=1}^{p+1}|z_{j_{1,\ell_1}}|\prod_{\ell_2=1}^{p+1}|z_{j_{2,\ell_2}}|\right)K^2(z)f(r + h \circ z)dz \nonumber \\
&= o\left( \sum_{1 \leq j_{1,1} \leq \dots \leq j_{1,p+1} \leq d, 1 \leq j_{2,1} \leq \dots \leq j_{2,p+1} \leq d}\prod_{\ell_1=1}^{p+1}h_{j_{1,\ell_1}}\prod_{\ell_2=1}^{p+1}h_{j_{2,\ell_2}}\right).
\end{align}
Then we have $B_{n,j_1\dots j_L1} = o_p(1)$. 

For $B_{n,j_1\dots j_L2}$, 
\begin{align}\label{B_n02}
&|B_{n,j_1\dots j_L2}| \nonumber \\ 
&\leq {1 \over (p+1)!}\max_{1 \leq j_1,\dots,j_{p+1} \leq d}\sup_{y \in N_r(h)}|\partial_{j_1\dots j_{p+1}}m(y) - \partial_{j_1\dots j_{p+1}}m(r)| \nonumber \\
&\quad \times \sqrt{nh_1 \cdots h_d}\sum_{1 \leq j_{1,1} \leq \dots \leq j_{1,p+1} \leq d}\prod_{\ell_1=1}^{p+1}h_{j_{1,\ell_1}}\int \left(\prod_{\ell=1}^{L}|z_{j_\ell}|\prod_{\ell_1=1}^{p+1}|z_{j_{1,\ell_1}}|\right)|K(z)|f(r + h \circ z)dz \nonumber \\
&= o(1). 
\end{align}

For $B_{n,j_1\dots j_L3}$, 
\begin{align}\label{B_n03}
&\Var(B_{n,j_1\dots j_L3}) \nonumber \\ 
&\leq {1 \over \{(p+1)!\}^2} \sum_{1 \leq j_{1,1} \leq \dots \leq j_{1,p+1} \leq d, 1 \leq j_{2,1} \leq \dots \leq j_{2,p+1} \leq d}\partial_{j_{1,1}\dots j_{1,p+1}}m(r)\partial_{j_{2,1}\dots j_{2,p+1}}m(r) \nonumber \\
&\quad \quad \times \prod_{\ell_1=1}^{p+1}h_{j_{1,\ell_1}}\prod_{\ell_2=1}^{p+1}h_{j_{2,\ell_2}} \int \left(\prod_{\ell=1}^{L}z_{j_\ell}^2\prod_{\ell_1}^{p+1}|z_{j_{1,\ell_1}}|\prod_{\ell_2=1}^{p+1}|z_{j_{2,\ell_2}}|\right)K^2(z)f(r + h \circ z)dz \nonumber  \\ 
&= o(1).
\end{align}
Then we have $B_{n,j_1\dots j_L3} = o_p(1)$. 

For $B_{n,j_1\dots j_L4}$, 
\begin{align}\label{B_n04}
&B_{n,j_1\dots j_L4} \nonumber \\ 
&= \sqrt{nh_1 \cdots h_d}\sum_{1 \leq j_{1,1} \leq \dots \leq j_{1,p+1} \leq d}{\partial_{j_{1,1}\dots j_{1,p+1}}m(r) \over \bm{s}_{j_{1,1} \dots j_{1,p+1}}!}\nonumber \\&\quad \times \prod_{\ell_1=1}^{p+1}h_{j_{1,\ell_1}} \int \left(\prod_{\ell=1}^{L}z_{j_\ell} \prod_{\ell_1=1}^{p+1}z_{j_{1,\ell_1}}\right) K(z)f(r + h \circ z)dz \nonumber \\
&= f(r) \sqrt{nh_1 \cdots h_d}\sum_{1 \leq j_{1,1} \leq \dots \leq j_{1,p+1} \leq d}{\partial_{j_{1,1}\dots j_{1,p+1}}m(r) \over \bm{s}_{j_{1,1} \dots j_{1,p+1}}!}\prod_{\ell_1=1}^{p+1}h_{j_{1,\ell_1}} \kappa_{j_1\dots j_Lj_{1,1}\dots j_{1,p+1}}^{(1)} + o(1). 
\end{align} 
Combining (\ref{B_n01})-(\ref{B_n04}), 
\begin{align*}
B_{n,j_1\dots j_L}(\tilde{R}_i) &= f(r)\sqrt{nh_1 \cdots h_d}\sum_{1 \leq j_{1,1} \leq \dots \leq j_{1,p+1} \leq d}{\partial_{j_{1,1}\dots j_{1,p+1}}m(r) \over \bm{s}_{j_{1,1} \dots j_{1,p+1}}!}\\
&\quad \times \prod_{\ell_1=1}^{p+1}h_{j_{1,\ell_1}} \kappa_{j_1\dots j_Lj_{1,1}\dots j_{1,p+1}}^{(1)} + o_p(1).
\end{align*}

\noindent
(Step 4) Combining the results in Steps1-3, we have
\begin{align*}
A_n(r) &:=V_n(r) + \left(B_n(r) - f(r)\sqrt{nh_1 \cdots h_d}\left(b_{n,j_1\dots j_L}(r)\right)'_{1 \leq j_1 \leq \dots \leq  j_L \leq d,0\leq L\leq p}\right) \\ 
&\stackrel{d}{\to} N\left( \left(
\begin{array}{c}
0 \\
\vdots \\
0
\end{array}
\right), \sigma^2(r)f(r)\mathcal{K}\right).
\end{align*} 
This yields the desired result. 
\end{proof}


\begin{remark}[General form of the MSE of $\widehat{\partial_{j_1\dots j_L}m(r)}$]
Define
\begin{align*}
\bm{b}_n^{(d,p)}(r) &:= B^{(d,p)}M_n^{(d,p)}(r)\\ 
&= \left(b_{n,0}(r), b_{n,1}(r),\dots,b_{n,d}(r), \right. \\
&\left. \quad \quad \quad b_{n,11}(r),b_{n,12}(r),\dots,b_{n,dd}(r), \dots, b_{n,1\dots,1}(r), b_{n,1\dots 2}(r),\dots,b_{n,d \dots d}(r)\right)'
\end{align*}
and let $e_{j_1\dots j_L}= (0,\dots,0,1,0,\dots,0)'$ be a $D$-dimensional vector such that $e_{j_1\dots j_L}'\bm{b}_n^{(d,p)}(r) = b_{j_1\dots j_L}(r)$. 
Theorem \ref{thm: LP-CLT} yields that 
\begin{align*}
b_{n,j_1,\dots,j_L}(r) &:= \sum_{1 \leq j_{1,1} \leq \dots \leq j_{1,p+1}\leq d}{\partial_{j_{1,1}\dots j_{1,p+1}}m(r) \over \bm{s}_{j_{1,1}\dots j_{1,p+1}}!}\prod_{\ell_1=1}^{p+1}h_{j_{1,\ell_1}}\kappa_{j_1\dots j_L j_{1,1} \dots j_{1,p+1}}^{(1)},\\
\end{align*}
for $1 \leq j_1 \leq \dots \leq  j_L \leq d$, $0 \leq L \leq p$ and 
\begin{align*}
&\text{MSE}(\widehat{\partial_{j_1\dots j_L}m(r)}) \\ 
&= \left\{\bm{s}_{j_1\dots j_L}!{(S^{-1}e_{j_1\dots j_L})'B^{(d,p)}M_n^{(d,p)}(r) \over \prod_{\ell=1}^{L}h_{j_\ell}}\right\}^2\\ 
&\quad + \left(\bm{s}_{j_1\dots j_L}!\right)^2{\sigma^2(r) \over nh_1\cdots h_d \times \left(\prod_{\ell=1}^{L}h_{j_\ell}\right)^2f(r)}e'_{j_1\dots j_L}S^{-1}\mathcal{K}S^{-1}e_{j_1\dots j_L}.
\end{align*}
\end{remark}

\subsection{Higher-order bias} \label{sec.asymptotic_bias}
In this section, we derive higher-order biases of local-polynomial estimators. Suppose that Assumptions \ref{Ass1}, \ref{Ass2}, \ref{Ass3} and \ref{Ass4} hold. Further, we assume that 
\begin{itemize}
\item the density function $f$ is continuously differentiable on $U_r$. 
\item the mean function $m$ is $(p+2)$-times continuously differentiable on $U_r$. 
\end{itemize}
Recall that 
\begin{align*}
\sqrt{nh_1 \cdots h_d}H(\hat{\beta}(r) - M(r)) &= S_n^{-1}(V_n(r) + B_n(r)),
\end{align*}
where
\begin{align*}
S_n(r) &= {1 \over nh_1 \cdots h_d}\sum_{i=1}^{n}K_h\left(R_i - r\right)H^{-1}
\left(
\begin{array}{c}
1 \\
\check{\R}_i
\end{array}
\right)
(1\ \check{\R}'_i)H^{-1},\\
V_n(r) &= {1 \over \sqrt{nh_1 \cdots h_d}}\sum_{i=1}^{n}K_h\left(R_i - r\right)H^{-1}
\left(
\begin{array}{c}
1 \\
\check{\R}_i
\end{array}
\right)\varepsilon_i =: (V_{n,j_1\dots j_L}(r))'_{1 \leq j_1\leq \dots \leq j_L \leq d, 0 \leq L \leq p},\\
B_n(r) &= {1 \over \sqrt{nh_1 \cdots h_d}}\sum_{i=1}^{n}K_h\left(R_i - r\right)H^{-1}
\left(
\begin{array}{c}
1 \\
\check{\R}_i
\end{array}
\right)\\
&\quad \times \left\{\sum_{1 \leq j_1 \leq \dots \leq j_{p+1}\leq d}{1 \over \bm{s}_{j_1\dots j_{p+1}}!}\partial_{j_1,\dots,j_{p+1}}m(r)\prod_{\ell=1}^{p+1}(R_{i,j_\ell} - r_{j_\ell}) \right. \\
& \left. \quad \quad  + \sum_{1 \leq j_1 \leq \dots \leq j_{p+2}\leq d}{1 \over \bm{s}_{j_1\dots j_{p+2}}!}\partial_{j_1,\dots,j_{p+2}}m(\tilde{R}_i)\prod_{\ell=1}^{p+2}(R_{i,j_\ell} - r_{j_\ell})\right\} \\
&=: (B_{n,j_1\dots j_L}(\tilde{R}))'_{1 \leq j_1\leq \dots \leq j_L \leq d, 0 \leq L \leq p}.
\end{align*}
Now we focus on $B_{n,j_1 \dots j_L}(\tilde{R})$. 
\begin{align*}
&B_{n,j_1 \dots j_L}(\tilde{R}) \\
&= {1 \over \sqrt{nh_1 \cdots h_d}}\sum_{i=1}^{n}K_h\left(R_i - r\right) \left(\prod_{\ell=1}^{L}{R_{i,j_\ell} - r_{j_\ell} \over h_{j_\ell}}\right) \\
&\quad \times \left\{\sum_{1 \leq j_{1,1} \leq \dots \leq j_{1,p+1}\leq d}{1 \over \bm{s}_{j_{1,1}\dots j_{1,p+1}}!}\partial_{j_{1,1},\dots,j_{1,p+1}}m(r)\prod_{\ell_1=1}^{p+1}(R_{i,j_{1,\ell_1}} - r_{j_{1,\ell_1}}) \right. \\
& \left. \quad \quad  + \sum_{1 \leq j_{1,1} \leq \dots \leq j_{1,p+2}\leq d}{1 \over \bm{s}_{j_{1,1}\dots j_{1,p+2}}!}\partial_{j_{1,1},\dots,j_{1,p+2}}m(\tilde{R}_i)\prod_{\ell_1=1}^{p+2}(R_{i,j_{1,\ell_1}} - r_{j_{1,\ell_1}})\right\}\\
&=: \mathbb{B}_{n,1}(r) + \mathbb{B}_{n,2}(\tilde{R}).
\end{align*}
For $\mathbb{B}_{n,1}(r)$, 
\begin{align}
E[\mathbb{B}_{n,1}(r)] &= \sqrt{n \over h_1 \cdots h_d}E\left[K_h(R_1 - r)\left(\prod_{\ell=1}^{L}{R_{1,j_\ell} - r_{j_\ell} \over h_{j_\ell}}\right) \right. \nonumber \\
&\left. \quad \times \sum_{1 \leq j_{1,1} \leq \dots \leq j_{1,p+1}\leq d}{1 \over \bm{s}_{j_{1,1}\dots j_{1,p+1}}!}\partial_{j_{1,1},\dots,j_{1,p+1}}m(r)\prod_{\ell_1=1}^{p+1}(R_{1,j_{1,\ell_1}} - r_{j_{1,\ell_1}})\right] \nonumber \\
&= \sqrt{nh_1\cdots h_d}\sum_{1 \leq j_{1,1} \leq \dots \leq j_{1,p+1}\leq d}{1 \over \bm{s}_{j_{1,1}\dots j_{1,p+1}}!}\partial_{j_{1,1},\dots,j_{1,p+1}}m(r)\prod_{\ell_1=1}^{p+1}h_{j_{1,\ell_1}} \nonumber \\
&\quad \times \int  \prod_{\ell=1}^{L}z_{j_\ell}\prod_{\ell_1=1}^{p+1}z_{j_{1,\ell_1}}K(z)f(r + h \circ z)dz \nonumber \\
&=  \sqrt{nh_1\cdots h_d}\sum_{1 \leq j_{1,1} \leq \dots \leq j_{1,p+1}\leq d}{1 \over \bm{s}_{j_{1,1}\dots j_{1,p+1}}!}\partial_{j_{1,1},\dots,j_{1,p+1}}m(r)\prod_{\ell_1=1}^{p+1}h_{j_{1,\ell_1}} \nonumber \\
&\quad \times \left(f(r)\int  \prod_{\ell=1}^{L}z_{j_\ell}\prod_{\ell_1=1}^{p+1}z_{j_{1,\ell_1}}K(z)dz \right.\nonumber \\ 
&\left. \quad \quad + \sum_{k=1}^{d}\partial_k f(r) h_k\int  z_k\prod_{\ell=1}^{L}z_{j_\ell}\prod_{\ell_1=1}^{p+1}z_{j_{1,\ell_1}}K(z)dz\right)(1 + o(1)). \label{higher-order-B11}
\end{align}
\begin{align}
&\Var(\mathbb{B}_{n,1}(r)) \nonumber \\
&\leq \sum_{1 \leq j_{1,1} \leq \dots \leq j_{1,p+1} \leq d, 1 \leq j_{2,1} \leq \dots \leq j_{2,p+1} \leq d}\partial_{j_{1,1}\dots j_{1,p+1}}m(r)\partial_{j_{2,1}\dots j_{2,p+1}}m(r) \nonumber \\
&\quad \quad \times \prod_{\ell_1=1}^{p+1}h_{j_{1,\ell_1}}\prod_{\ell_2=1}^{p+1}h_{j_{2,\ell_2}} \int \left(\prod_{\ell=1}^{L}z_{j_\ell}^2\prod_{\ell_1=1}^{p+1}|z_{j_{1,\ell_1}}|\prod_{\ell_2=1}^{p+1}|z_{j_{2,\ell_2}}|\right)K^2(z)f(r + h \circ z)dz \nonumber \\
&= O\left(\left(\sum_{1 \leq j_1 \leq \dots \leq j_{p+1}\leq d}\prod_{\ell=1}^{p+1}h_{j_\ell}\right)^2\right). \label{higher-order-B12}
\end{align}
For $\mathbb{B}_{n,2}(\tilde{R})$, 
\begin{align*}
\mathbb{B}_{n,2}(\tilde{R}) &= \left\{\mathbb{B}_{n,2}(\tilde{R}) - \mathbb{B}_{n,2}(r) - E[\mathbb{B}_{n,2}(\tilde{R}) - \mathbb{B}_{n,2}(r)]\right\} \\
&\quad + E[\mathbb{B}_{n,2}(\tilde{R}) - \mathbb{B}_{n,2}(r)]\\
&\quad + \mathbb{B}_{n,2}(r) - E[\mathbb{B}_{n,2}(r)]\\
&\quad + E[\mathbb{B}_{n,2}(r)]\\
&=: \sum_{\ell = 1}^4\mathbb{B}_{n,2\ell}.
\end{align*}
Define $N_r(h):= \prod_{j=1}^{d}[r_j-C_Kh_j, r_j + C_Kh_j]$. For $\mathbb{B}_{n,21}$, 
\begin{align}\label{B_n21}
&\Var(\mathbb{B}_{n,21})  \nonumber \\ 
&\leq {1 \over h_1 \cdots h_d}E\left[K_h^2\left(R_i - r \right)\prod_{\ell=1}^{L}\left({R_{i,j_\ell} - r_{j_\ell} \over h_{j_\ell}}\right)^2 \right. \nonumber \\
&\left. \quad \times \sum_{1 \leq j_{1,1} \leq \dots \leq j_{1,p+2} \leq d, 1 \leq j_{2,1} \leq \dots \leq j_{2,p+2} \leq d}{1 \over \bm{s}_{j_{1,1} \dots j_{1,p+2}}! }{1 \over \bm{s}_{j_{2,1}\dots j_{2,p+2}}!} \right. \nonumber  \\
&\left. \quad \times (\partial_{j_{1,1}\dots j_{1,p+2}}m(\tilde{R}_i) - \partial_{j_{1,1}\dots j_{1,p+2}}m(r))(\partial_{j_{2,1}\dots j_{2,p+2}}m(\tilde{R}_i) - \partial_{j_{2,1}\dots j_{2,p+2}}m(r)) \right.  \nonumber \\
&\left. \quad  \times \prod_{\ell_1=1}^{p+2}(R_{i,j_{1,\ell_1}} - r_{j_{1\ell_1}})\prod_{\ell_2=1}^{p+2}(R_{i,j_{2,\ell_2}}-r_{j_{2\ell_2}})\right]  \nonumber \\ 
&\leq \max_{1 \leq j_1 \leq \dots \leq j_{p+2} \leq d}\sup_{y \in N_r(h)}|\partial_{j_1\dots j_{p+2}}m(y) - \partial_{j_1\dots j_{p+2}}m(r)|^2 \nonumber \\
&\quad \quad \times  \sum_{1 \leq j_{1,1} \leq \dots \leq j_{1,p+2} \leq d, 1 \leq j_{2,1} \leq \dots \leq j_{2,p+2} \leq d}\prod_{\ell_1=1}^{p+2}h_{j_{1,\ell_1}}\prod_{\ell_2=1}^{p+2}h_{j_{2,\ell_2}} \nonumber \\
&\quad \quad \times \int \left(\prod_{\ell=1}^{L}|z_{j_\ell}|\prod_{\ell_1=1}^{p+2}|z_{j_{1,\ell_1}}|\prod_{\ell_2=1}^{p+2}|z_{j_{2,\ell_2}}|\right)K^2(z)f(r + h \circ z)dz \nonumber \\
&= o\left( \left(\sum_{1 \leq j_1 \leq \dots \leq j_{p+2} \leq d}\prod_{\ell=1}^{p+2}h_{j_\ell}\right)^2\right).
\end{align}
For $\mathbb{B}_{n,22}$, 
\begin{align}\label{B_n22}
&|\mathbb{B}_{n,22}| \nonumber \\
&\leq \max_{1 \leq j_1,\dots,j_{p+2} \leq d}\sup_{y \in N_r(h)}|\partial_{j_1\dots j_{p+2}}m(y) - \partial_{j_1\dots j_{p+2}}m(r)| \nonumber \\
&\quad \times \sqrt{nh_1 \cdots h_d}\sum_{1 \leq j_{1,1} \leq \dots \leq j_{1,p+2} \leq d}\prod_{\ell_1=1}^{p+2}h_{j_{1,\ell_1}}\int \left(\prod_{\ell=1}^{L}|z_{j_\ell}|\prod_{\ell_1=1}^{p+2}|z_{j_{1,\ell_1}}|\right)|K(z)|f(r + h \circ z)dz \nonumber \\
&= o\left(\sqrt{nh_1 \cdots h_d}\sum_{1 \leq j_{1,1} \leq \dots \leq j_{1,p+2} \leq d}\prod_{\ell_1=1}^{p+2}h_{j_{1,\ell_1}}\right). 
\end{align}
For $\mathbb{B}_{n,23}$, 
\begin{align}\label{B_n23}
&\Var(\mathbb{B}_{n,23}) \nonumber \\ 
&\leq  \sum_{1 \leq j_{1,1} \leq \dots \leq j_{1,p+2} \leq d, 1 \leq j_{2,1} \leq \dots \leq j_{2,p+2} \leq d}\partial_{j_{1,1}\dots j_{1,p+2}}m(r)\partial_{j_{2,1}\dots j_{2,p+2}}m(r) \nonumber \\
&\quad \quad \times \prod_{\ell_1=1}^{p+2}h_{j_{1,\ell_1}}\prod_{\ell_2=1}^{p+2}h_{j_{2,\ell_2}} \int \left(\prod_{\ell=1}^{L}z_{j_\ell}^2\prod_{\ell_1}^{p+2}|z_{j_{1,\ell_1}}|\prod_{\ell_2=1}^{p+2}|z_{j_{2,\ell_2}}|\right)K^2(z)f(r + h \circ z)dz \nonumber  \\ 
&= O\left(\left(\sum_{1 \leq j_1 \leq \dots \leq j_{p+2}  \leq d}\prod_{\ell=1}^{p+2}h_{j_\ell}\right)^2\right).
\end{align}
For $\mathbb{B}_{n,24}$, 
\begin{align}\label{B_n24}
\mathbb{B}_{n,24} &= \sqrt{nh_1 \cdots h_d}\sum_{1 \leq j_{1,1} \leq \dots \leq j_{1,p+2} \leq d}{\partial_{j_{1,1}\dots j_{1,p+2}}m(r) \over \bm{s}_{j_{1,1} \dots j_{1,p+2}}!}\nonumber \\&\quad \times \prod_{\ell_1=1}^{p+2}h_{j_{1,\ell_1}} \int \left(\prod_{\ell=1}^{L}z_{j_\ell} \prod_{\ell_1=1}^{p+2}z_{j_{1,\ell_1}}\right) K(z)f(r + h \circ z)dz \nonumber \\
&= f(r) \sqrt{nh_1 \cdots h_d} \nonumber \\
&\times \left(\sum_{1 \leq j_{1,1} \leq \dots \leq j_{1,p+2} \leq d}{\partial_{j_{1,1}\dots j_{1,p+2}}m(r) \over \bm{s}_{j_{1,1} \dots j_{1,p+2}}!}\prod_{\ell_1=1}^{p+2}h_{j_{1,\ell_1}} \int \left(\prod_{\ell=1}^{L}z_{j_\ell} \prod_{\ell_1=1}^{p+2}z_{j_{1,\ell_1}}\right) K(z)dz\right)(1 + o(1)). 
\end{align} 
Combining (\ref{higher-order-B11})-(\ref{B_n24}), 
\begin{align*}
&B_{n,j_1\dots j_L}(\tilde{R})\\ 
& =  \sqrt{nh_1\cdots h_d}\sum_{1 \leq j_{1,1} \leq \dots \leq j_{1,p+1}\leq d}{1 \over \bm{s}_{j_{1,1}\dots j_{1,p+1}}!}\partial_{j_{1,1},\dots,j_{1,p+1}}m(r)\prod_{\ell_1=1}^{p+1}h_{j_{1,\ell_1}} \nonumber \\
& \times \left(f(r)\int  \prod_{\ell=1}^{L}z_{j_\ell}\prod_{\ell_1=1}^{p+1}z_{j_{1,\ell_1}}K(z)dz + \sum_{k=1}^{d}\partial_k f(r) h_k\int  \left(z_k\prod_{\ell=1}^{L}z_{j_\ell}\prod_{\ell_1=1}^{p+1}z_{j_{1,\ell_1}}\right)K(z)dz\right)(1 + o(1)). \\  
& \quad + \sqrt{nh_1 \cdots h_d}\\
& \times \left(f(r)\sum_{1 \leq j_{1,1} \leq \dots \leq j_{1,p+2} \leq d}{\partial_{j_{1,1}\dots j_{1,p+2}}m(r) \over \bm{s}_{j_{1,1} \dots j_{1,p+2}}!}\prod_{\ell_1=1}^{p+2}h_{j_{1,\ell_1}} \int \left(\prod_{\ell=1}^{L}z_{j_\ell} \prod_{\ell_1=1}^{p+2}z_{j_{1,\ell_1}}\right) K(z)dz\right)(1 + o(1)).
\end{align*}

\subsubsection{Higher-order bias of the local-linear estimator} \label{sec.asymptotic_higher_order}
For local-linear estimators (i.e., $d=2, p=1$), we have
\begin{align*}
b_{n,0} &= {f(r) \over 2}\sum_{j,k=1}^{2}\partial_{jk}m(r)h_jh_k \int z_kz_j K(z)dz\\
&\quad + \sum_{\ell=1}^{2}{\partial_\ell f(r) \over 2}\sum_{j,k=1}^{2}\partial_{jk}m(r)h_j h_k h_\ell \int z_jz_kz_\ell K(z)dz\\
&\quad + {f(r) \over 6}\sum_{j,k,\ell=1}^{2}\partial_{jk\ell}m(r)h_j h_k h_\ell \int z_j z_k z_\ell K(z)dz,
\end{align*}
\begin{align*}
b_{n,1} &= {f(r) \over 2}\sum_{j,k=1}^{2}\partial_{jk}m(r)h_jh_k \int z_1z_kz_j K(z)dz\\
&\quad + \sum_{\ell=1}^{2}{\partial_\ell f(r) \over 2}\sum_{j,k=1}^{2}\partial_{jk}m(r)h_j h_k h_\ell \int z_1z_jz_kz_\ell K(z)dz\\
&\quad + {f(r) \over 6}\sum_{j,k,\ell=1}^{2}\partial_{jk\ell}m(r)h_j h_k h_\ell \int z_1z_j z_k z_\ell K(z)dz,
\end{align*}
\begin{align*}
b_{n,2} &= {f(r) \over 2}\sum_{j,k=1}^{2}\partial_{jk}m(r)h_jh_k \int z_2z_kz_j K(z)dz\\
&\quad + \sum_{\ell=1}^{2}{\partial_\ell f(r) \over 2}\sum_{j,k=1}^{2}\partial_{jk}m(r)h_j h_k h_\ell \int z_2z_jz_kz_\ell K(z)dz\\
&\quad + {f(r) \over 6}\sum_{j,k,\ell=1}^{2}\partial_{jk\ell}m(r)h_j h_k h_\ell \int z_2z_j z_k z_\ell K(z)dz.
\end{align*}
When $K(z) = K_1(z_1)K_2(z_2)$ where $K_1(z_1) = (1-|z_1|)1_{\{|z_1| \leq 1\}}$ and $K_2(z_2) = 2(1-z_2)1_{\{0 \leq z_2 \leq 1\}}$, we have
\begin{align*}
b_{n,0} &=  {f(r) \over 2}\left\{h_1^2\partial_{11}m(r)\kappa_1^{(2,1)} + h_2^2 \partial_{22}m(r)\kappa_2^{(2,1)} \right\}\\
&\quad + {\partial_1 f(r) \over 2}\left(2h_1^2h_2\partial_{12}m(r)\kappa_{1,2}^{(2,1,1)}\right)\\
&\quad + {\partial_2 f(r) \over 2}\left(h_1^2h_2\partial_{11}m(r)\kappa_{1,2}^{(2,1,1)} + h_2^3\partial_{22}m(r)\kappa_2^{(3,1)}\right)\\
&\quad + {f(r) \over 6}\left(3h_1^2h_2\partial_{112}m(r)\kappa_{1,2}^{(2,1,1)} + h_2^3\partial_{222}m(r)\kappa_2^{(3,1)}\right),
\end{align*}
\begin{align*}
b_{n,1} &= {f(r) \over 2}\left(2h_1h_2\partial_{12}m(r)\kappa_{1,2}^{(2,1,1)}\right)\\
&\quad + {\partial_1 f(r) \over 2}\left(h_2^3 \partial_{11}m(r)\kappa_1^{(4,1)} + h_1^2h_2 \partial_{22}m(r)\kappa_{1,2}^{(2,2,1)}\right)\\
&\quad + {\partial_2 f(r) \over 2}\left(2h_1h_2^2 \partial_{12}m(r)\kappa_{1,2}^{(2,2,1)}\right)\\
&\quad + {f(r) \over 6}\left(h_1^3\partial_{111}m(r)\kappa_{1}^{(4,1)} + 3h_1h_2^2\partial_{122}m(r)\kappa_{1,2}^{(2,2,1)}\right),
\end{align*}
\begin{align*}
b_{n,2} &= {f(r) \over 2}\left(h_1^2\partial_{11}m(r)\kappa_{1,2}^{(2,1,1)} + h_2^2 \partial_{22}m(r)\kappa_2^{(3,1)}\right)\\
&\quad + {\partial_1 f(r) \over 2}\left(2h_1^2h_2\partial_{12}m(r)\kappa_{1,2}^{(2,2,1)} \right)\\
&\quad + {\partial_2 f(r) \over 2}\left(h_1^2h_2 \partial_{11}m(r)\kappa_{1,2}^{(2,2,1)} + h_2^3 \partial_{22}m(r)\kappa_2^{(4,1)}\right)\\
&\quad + {f(r) \over 6}\left(3h_1^2h_2\partial_{112}m(r)\kappa_{1,2}^{(2,2,1)} + h_2^3\partial_{222}m(r)\kappa_{2}^{(4,1)}\right).
\end{align*}
Therefore, 
\begin{align*}
&\text{Bias}(\hat{m}(r)) \\
&= \tilde{s}_1b_{n,0} + \tilde{s}_3 b_{n,2}\\
&= \left\{{h_1^2 \over 2}\partial_{11}m(r)(\tilde{s}_1\kappa_1^{(2,1)} + \tilde{s}_3 \kappa_{1,2}^{(2,1,1)}) + {h_2^2 \over 2}\partial_{22}m(r)(\tilde{s}_1\kappa_2^{(2,1)} + \tilde{s}_3 \kappa_{2}^{(3,1)})\right\}\\
&\quad + h_1^2h_2\left({\partial_{11}m(r) \over 2}{\partial_2 f(r) \over f(r)} + \partial_{12}m(r){\partial_1 f(r) \over f(r)} + {\partial_{112}m(r) \over 2}\right)(\tilde{s}_1\kappa_{1,2}^{(2,1,1)} + \tilde{s}_3 \kappa_{1,2}^{(2,2,1)})\\
&\quad + h_2^3\left({1 \over 2}\partial_{22}m(r){\partial_2 f(r) \over f(r)} + {1 \over 6}\partial_{222}m(r)\right)(\tilde{s}_1\kappa_2^{(3,1)} + \tilde{s}_3\kappa_2^{(4,1)}).
\end{align*}

\section{Shapes of the polynomial fits used in simulations}\label{sec.polynomialShapes}
\footnotesize

The following equations are the polynomials estimated from the real data at each evaluation points. In the estimation of coefficients, we rotate the axis so that the sign of $Y$ values determine the treatment status. Specifically, in Design 1 and Design 2 evaluate points where the boundary SABER 11 is 0, $X$ is SISBEN and $Y$ is SABER 11; in Design 3 and Design 4 evaluate points where the boundary SISBEN is 0, $X$ is SABER 11 and $Y$ is SISBEN. The $X$ values (and $Y$ values by construction) are re-centered to the evaluation point as the origin. The e notation ($k \text{e-n}$) represents $k \times 10^{-n}$.

\subsection{Design 1 at point 7}
\begin{align*}
    \mbox{ Control: } &  \\
    \input{control_side_polynomial_7}\\
    \mbox{ Treated: } & \\
    \input{treatment_side_polynomial_7}
\end{align*}

\subsection{Design 2 at point 13}
\begin{align*}
    \mbox{ Control: } & \\ \input{control_side_polynomial_13}\\
    \mbox{ Treated: } & \\
    \input{treatment_side_polynomial_13}
\end{align*}

\subsection{Design 3 at point 19}
\begin{align*}
    \mbox{ Control: } &  \\
    \input{control_side_polynomial_19}\\
    \mbox{ Treated: } & \\
    \input{treatment_side_polynomial_19}
\end{align*}

\subsection{Design 4 at point 25}
\begin{align*}
    \mbox{ Control: } &  \\
    \input{control_side_polynomial_25}\\
    \mbox{ Treated: } & \\
    \input{treatment_side_polynomial_25}
\end{align*}
\normalsize

\subsection{Supports for four designs}
\begin{figure}[H]
    \centering 
    \begin{minipage}{0.49\hsize}
    (a) Design 1
    \includegraphics[width=\columnwidth]{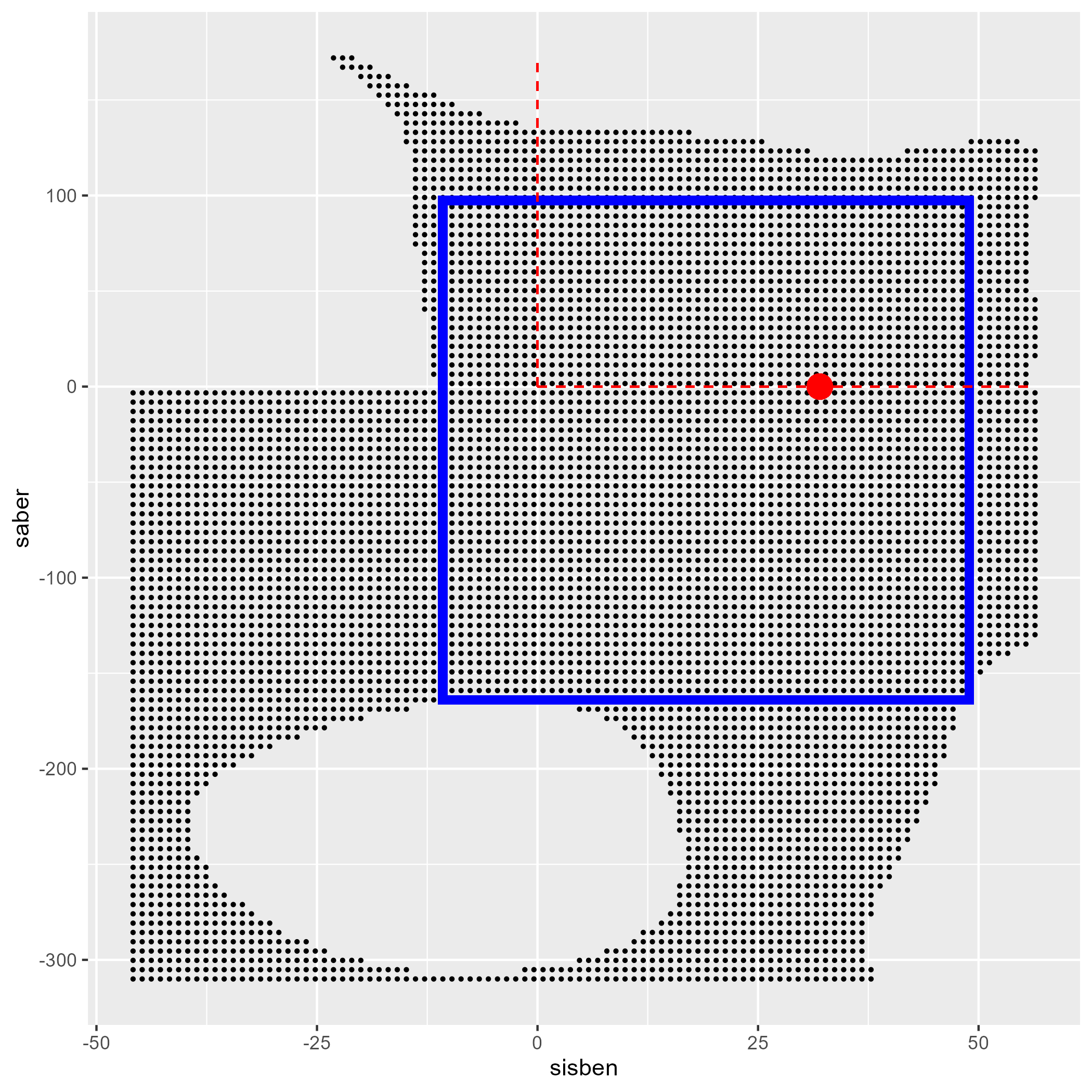}
    \centering 
    \end{minipage}    
    \begin{minipage}{0.49\hsize}
    (b) Design 2
    \includegraphics[width=\columnwidth]{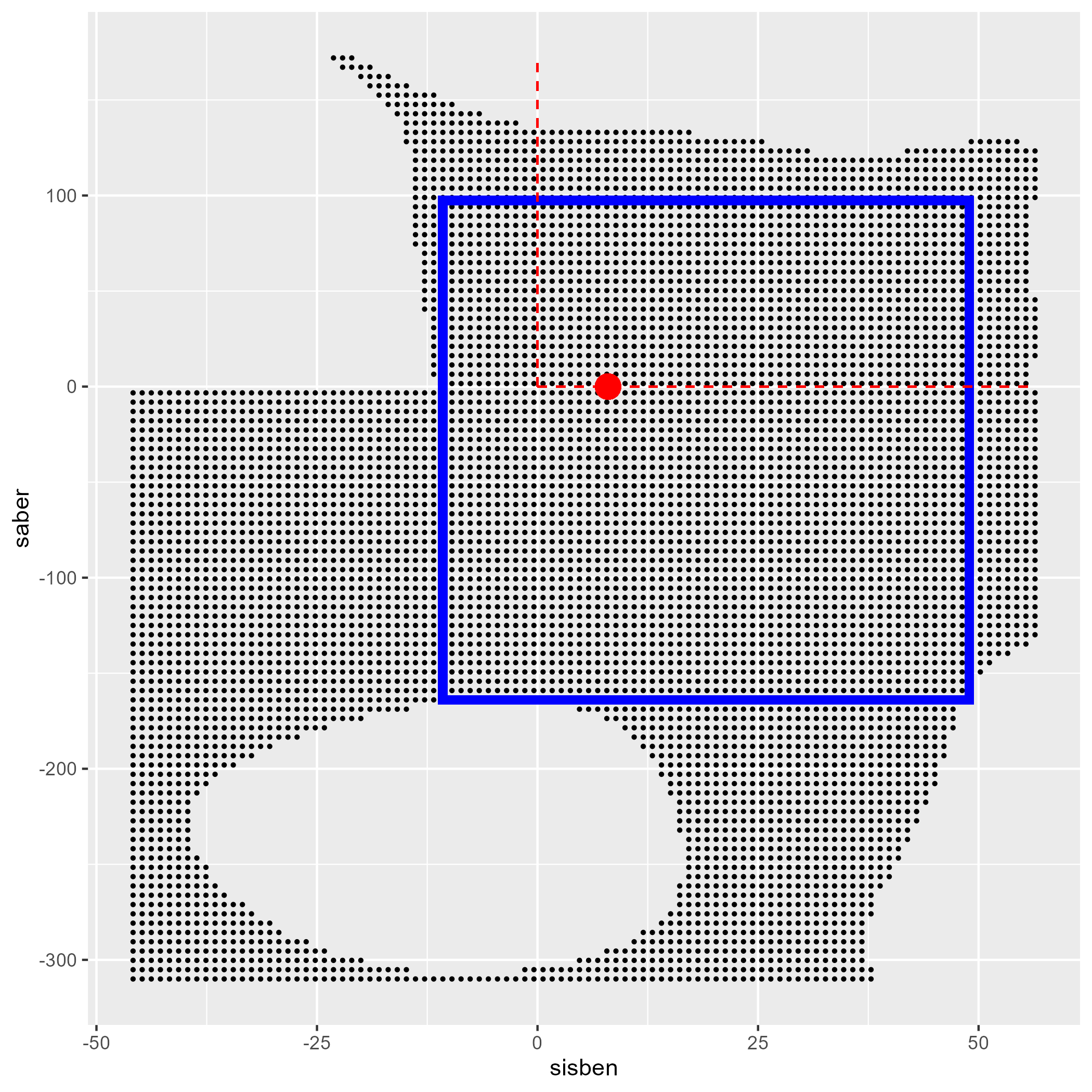}
    \centering
    \end{minipage}
    \begin{minipage}{0.49\hsize}
    (c) Design 3
    \includegraphics[width=\columnwidth]{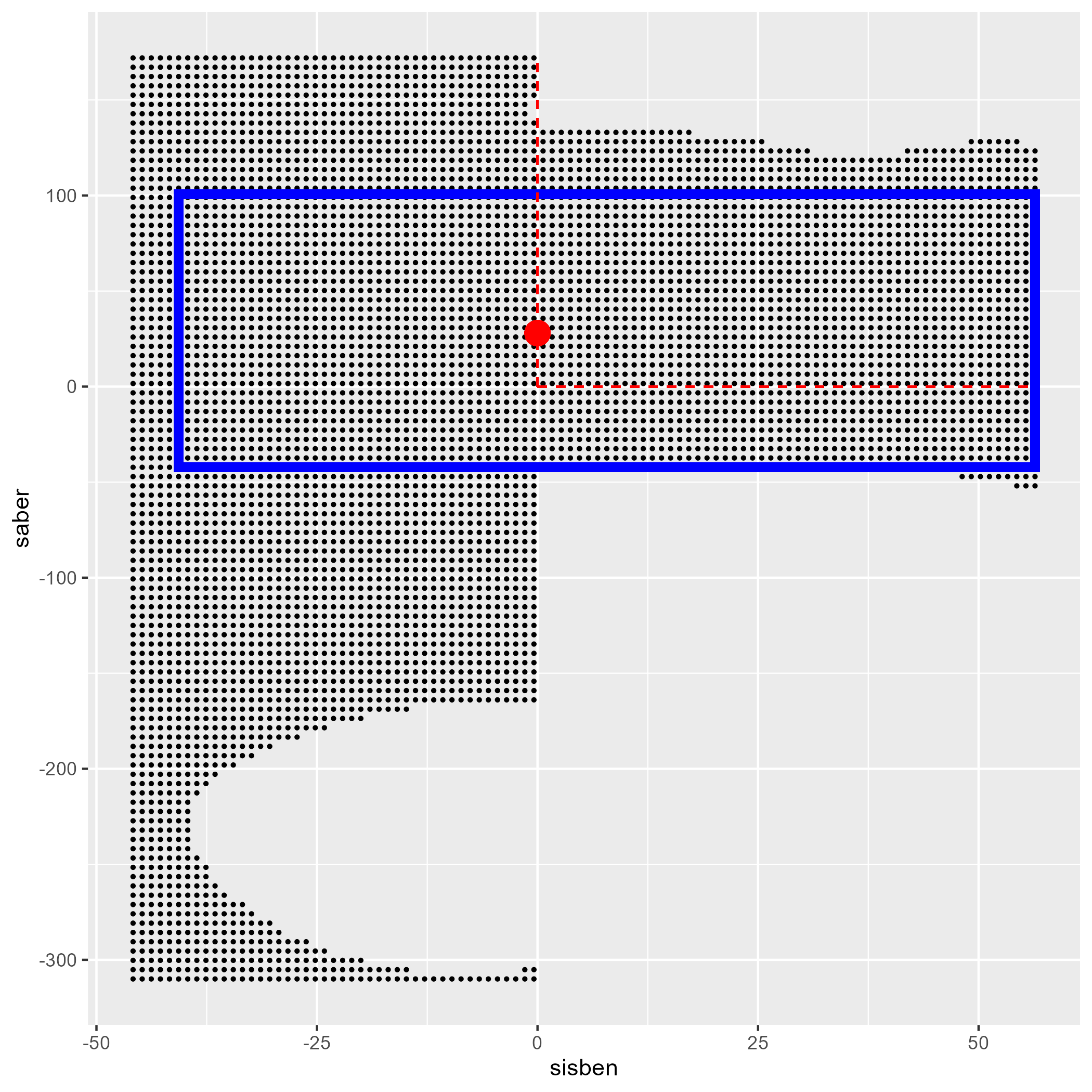}
    \centering 
    \end{minipage}
    \begin{minipage}{0.49\hsize}
    (d) Design 4
    \includegraphics[width=\columnwidth]{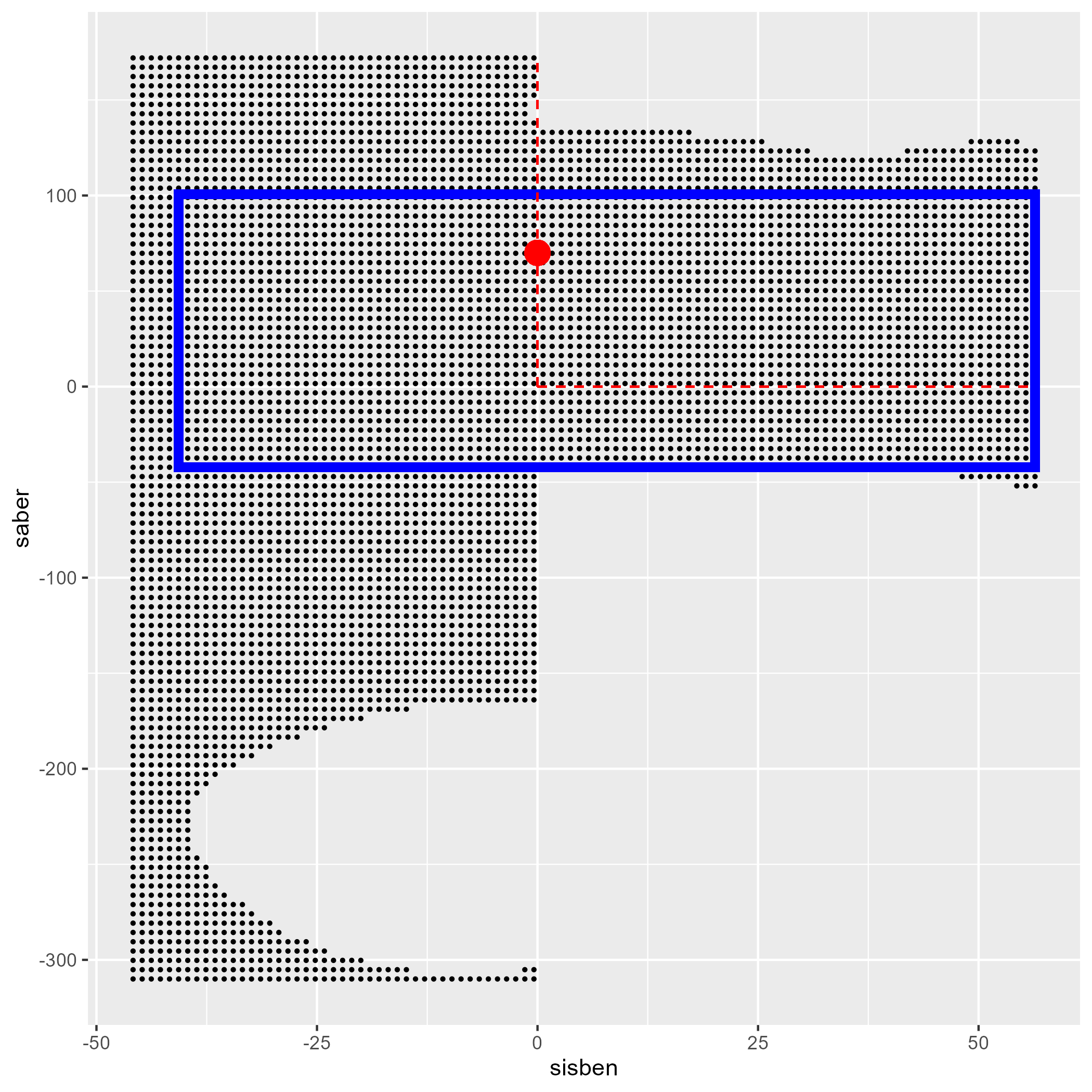}
    \centering 
    \end{minipage}
    \caption{The red circles represent each evaluation point on the boundary. Black dotted areas are points that have their global polynomial predictions from each evaluation point fall within $[0,1]$. The blue rectangles are the largest rectangle ares which falls within the black dotted areas. In the numerical simulations, observations are drawn from the blue rectangle supports.}
    \label{fig:supports}
\end{figure}
\section{Implementation details} \label{sec.implementation}
In section \ref{sec.estimator}, we propose our optimal bandwidth selection from the following formula:
\[
  \frac{h_1}{h_2} = \left(\frac{B_2(c)}{B_1(c)}\right)^{1/2}
\]
and
\[
 h_1 =  \left[\frac{(\sigma^2_+(c) + \sigma_-^2(c))}{2n} e'_1 S^{-1}\mathcal{K}S^{-1}e_1 B_1(c)^{-5/2} B_2(c)^{-1/2}\right]^{1/6}
\]
and our RD estimate prior to the bias correction is $\hat{\beta}^+_0(c) - \hat{\beta}^-_0(c)$ where these intercept terms of the local-polynomial estimates $\{\hat{\beta}^+_0(c), \hat{\beta}^-_0(c)\}$ are computed with the bandwidths specified above. Nevertheless, to compute the optimal bandwidth, we need to estimate the bias terms $B_1(c)$ and $B_2(c)$ as well as the residual variances $\{\sigma^2_+(c),  \sigma_-^2(c)\}$. We follow \citealp{Calonico.Cattaneo.Titiunik2014}, Section 5) in estimation of the residual variances at the boundary point $c$. For the bias terms, as in \cite{Calonico.Cattaneo.Titiunik2014}, we set a pair of pilot bandwidths with the local-quadratic regression. The key complication of our study is that the local-quadratic regression is also multivariate. 

The expression of the bias terms involve a pair of partial derivatives $(\partial_{11} m_+(c),\partial_{22} m_+(c))$ for the treated and $(\partial_{11} m_-(c),\partial_{22} m_-(c))$ for the control. Given a pair of pilot bandwidths $b_+$ and $b_-$ for the treated and the control, we run the local-quadratic estimation
\begin{align*}
\hat{\gamma}^+(c) = \argmin_{(\gamma_0,\dots,\gamma_5)' \in \mathbb{R}^{6}}\sum_{i=1}^{n} \left(\right. &Y_i - \gamma_0 - \gamma_1(R_{i,1} - c_1) \\
&- \gamma_{2}(R_{i,2}-c_2) - \gamma_{3}(R_{i,1}-c_2)^2 \\
&- \gamma_{4}(R_{i,1}-c_1)(R_{i,2} - c_2)\\
&- \gamma_{5}(R_{i,2}-c_2)^2)^2K_b\left(R_i - c\right)1\{R_i \in \mathcal{T}\}    
\end{align*}
and
\begin{align*}
\hat{\gamma}^-(c) = \argmin_{(\gamma_0,\dots,\gamma_5)' \in \mathbb{R}^{6}}\sum_{i=1}^{n} \left(\right. &Y_i - \gamma_0 - \gamma_1(R_{i,1} - c_1) \\
&- \gamma_{2}(R_{i,2}-c_2) - \gamma_{3}(R_{i,1}-c_2)^2 \\
&- \gamma_{4}(R_{i,1}-c_1)(R_{i,2} - c_2)\\
&- \gamma_{5}(R_{i,2}-c_2)^2)^2K_b\left(R_i - c\right)1\{R_i \in \mathcal{T}^C\}    
\end{align*}
where
$
K_{b}(R_i - c) = K\left({R_{i,1} - c_1 \over b},{R_{i,2} - c_2 \over b}\right)
$
to obtain these partial derivatives. These pilot bandwidths $(b_+,b_-)$ are chosen from minimizing the mean squared error of estimating the bias term, which involves the local cubic regression. \footnote{Furthermore, we choose the preliminary bandwidth for the local cubic regression from minimizing the mean squared error of estimating the bias term for the pilot bandwidth. This preliminary bandwidth selection involves the global 4th order polynomial regressions.}

Given the pilot bandwidths, we estimate the bias terms $B_1(c)$ and $B_2(c)$. Let $\hat{B}_1(c)$ and $\hat{B}_2(c)$ be their estimates. In the optimal bandwidth selection, we follow \cite{Imbens.Kalyanaraman2012} to regularize the bias term which appears in the denominator. Specifically, we employ their result that the inverse of bias term estimation error is approximated by $3$ times of their variance. We choose the optimal bandwidths from the first-order condition: we set
\[
 h_1 = \left[\frac{(\hat{\sigma}^2_+(c) + \hat{\sigma}_-^2(c))}{2n} e'_1 S^{-1}\mathcal{K}S^{-1}e_1 (\hat{B}_1(c)^2 + 3\hat{\V}(\hat{B}_1(c))^{-1} \left(\frac{\hat{B}_2(c)^2}{\hat{B}_1(c)^2 + 3\hat{\V}(\hat{B}_1(c))}\right)^{1/4}\right]^{1/6}
\]
and
\[
 h_2 = \left[\frac{(\hat{\sigma}^2_+(c) + \hat{\sigma}_-^2(c))}{2n} e'_1 S^{-1}\mathcal{K}S^{-1}e_1 (\hat{B}_2(c)^2 + 3\hat{\V}(\hat{B}_2(c))^{-1} \left(\frac{\hat{B}_1(c)^2}{\hat{B}_2(c)^2 + 3\hat{\V}(\hat{B}_2(c))}\right)^{1/4}\right]^{1/6}
\]
separately for each subsample of the treated and control, where $\hat{\V}(\hat{B}_1(c))$ and $\hat{\V}(\hat{B}_2(c))$ are variance estimates from the bias estimation with the pilot bandwidths.

\section{Consequence of converting two-dimensional data to one dimension.} \label{sec.density}

Let $Z_i = \|R_i\|$ and $K_1(r) = 2(1-r)1_{\{0 \leq r \leq 1\}}$. Define 
\[
\check{f}(\mathbf{0}) = {1 \over \check{n}h}\sum_{i=1}^{n}K_1(Z_i/h)1_{\{R_{i,2} \geq 0\}},\ \check{n} = \sum_{i=1}^{n}1_{\{R_{i,2} \geq 0\}}. 
\]
Note that ${\check{n} \over n} = P(R_{1,2} \geq 0) + O_p(n^{-1/2})$ and 
\begin{align*}
\check{f}(\mathbf{0}) &= \left({1 \over (\check{n}/n)} - {1 \over P(R_{1,2} \geq 0)} + {1 \over P(R_{1,2} \geq 0)}\right){1 \over nh}\sum_{i=1}^{n}K_1(Z_i/h)1_{\{R_{i,2} \geq 0\}}\\
&= {1 \over P(R_{1,2} \geq 0)}{1 \over nh}\sum_{i=1}^{n}K_1(Z_i/h)1_{\{R_{i,2} \geq 0\}} + O_p(n^{-1/2})\\
&=: {1 \over P(R_{1,2} \geq 0)}\tilde{f}(\mathbf{0}) + O_p(n^{-1/2}). 
\end{align*}
Further, 
\begin{align*}
E[\tilde{f}(\mathbf{0})] &= {2 \over h}E[K_1(Z_1/h)1_{\{R_{1,2} \geq 0\}}]\\
&= {2 \over h}\int (1-\|(r_1/h, r_2/h)\|)1\{\|(r_1/h,r_2/h)\| \leq 1\}1_{\{r_2 \geq 0\}}f(r)dr\\
&= {2 \over h}\int (1-\|(r_1/h, r_2/h)\|)1\{\|(r_1/h,r_2/h)\| \leq 1\}1_{\{r_2/h \geq 0\}}f(r)dr\\
&= 2h\int (1-\|z\|)1_{\{\|z\| \leq 1, z_2 \geq 0\}}f(hz_1, hz_2)dz\\
&= 2h\left(f(\mathbf{0})\int (1-\|z\|)1_{\{\|z\| \leq 1, z_2 \geq 0\}}dz + o(1)\right) \\
&= 2h\left(f(\mathbf{0})\int_{0}^{1}(1-r)rdr\int_{0}^{\pi}d\theta + o(1)\right)\\
&= 2h\left({\pi \over 6}f(\mathbf{0}) + o(1)\right)
\end{align*}
where we used the dominated convergence theorem for the fifth equation, and 
\begin{align*}
\Var(\tilde{f}(\mathbf{0})) & \leq {1 \over nh^2}E\left[K_1^2(Z_1/h)1_{\{R_{1,2} \geq 0\}}\right]\\
&= {4 \over n}\int (1-\|z\|)^21_{\{\|z\| \leq 0, z_2 \geq 0\}}f(hz_1, hz_2)dz\\
&= {4 \over n}\left(f(\mathbf{0})\int (1-\|z\|)^21_{\{\|z\| \leq 1, z_2 \geq 0\}}dz + o(1)\right)\\
&= {4 \over n}\left(f(\mathbf{0})\int_0^1 (1-r)^2rdr \int_0^{\pi} d\theta + o(1)\right)\\
&= {4 \over n}\left({\pi \over 12}f(\mathbf{0}) + o(1)\right)
\end{align*} 
where we used the dominated convergence theorem for the second equation. Then we have
\begin{align*}
\check{f}(\mathbf{0}) &= {\pi h \over 3P(R_{1,2} \geq 0)}f(\mathbf{0}) + o(h) + O_p(n^{-1/2}). 
\end{align*}

\section{The \textit{rdrobust} bandwidth for the \textit{distance} strategy}
In this section, we show that the rate of convergence for the \textit{rdrobust} bandwidth for the \textit{distance} strategy depends on the pilot bandwidth. Let $|\cdot|$ denote the Euclidean matrix norm, that is, $|A|^2 = \mathrm{trace}(A'A)$ for scalar, vector, or matrix $A$. We write $a_n \precsim b_n$ to mean that $a_n \leq C b_n$ for some positive constant $C$ independent of $n$. Letting $Z_i = D_i\|R_i\| - (1-D_i) \|R_i\|$, we define
\begin{align*}
    r_1(z) &= (1,z)', \ \ \ e_1 = (1,0)', \\
    \Gamma_{+}(h) &= \frac{1}{n h} \sum_{i=1}^n 1\{Z_i \geq 0\} K(Z_i/h) r_1(Z_i/h) r_1(Z_i/h)', \\
    \Gamma_{-}(h) &= \frac{1}{n h} \sum_{i=1}^n 1\{Z_i < 0\} K(Z_i/h) r_1(Z_i/h) r_1(Z_i/h)', \\
    \hat{\Psi}_{+}(h) &= \frac{1}{nh^2}  \sum_{i=1}^n 1\{Z_i \geq 0\} K(Z_i/h)^2  r_1(Z_i/h) r_1(Z_i/h)' \hat{\sigma}^2(Z_i), \\
    \hat{\Psi}_{-}(h) &= \frac{1}{nh^2}  \sum_{i=1}^n 1\{Z_i < 0\} K(Z_i/h)^2  r_1(Z_i/h) r_1(Z_i/h)' \hat{\sigma}^2(Z_i), \\
    \hat{V}_{+}(h) &= e_1' \Gamma_{+}(h)^{-1} \hat{\Psi}_{+}(h) \Gamma_{+}(h)^{-1} e_1 /n, \ \ \ \hat{V}_{-}(h) = e_1' \Gamma_{-}(h)^{-1} \hat{\Psi}_{-}(h) \Gamma_{-}(h)^{-1} e_1 /n,
\end{align*}
where $\hat{\sigma}^2(z)$ is an estimator of $\sigma^2(z) = Var(Y_i|Z_i=z)$. Then $\tilde{V}_{CCT}$ in
\[
    \hat{h}_{CCT} = C \cdot \left( \frac{\tilde{V}_{CCT}}{\tilde{B}_{CCT}}  \right)^{1/5} n^{-1/5} \label{CCT_band_distance}
\]
is written as follows:
\begin{align*}
    \tilde{V}_{CCT} &= n h_{\mathrm{initial}} \left\{ \hat{V}_{+}(h_{\mathrm{initial}}) + \hat{V}_{-}(h_{\mathrm{initial}}) \right\},
\end{align*}
where $h_{\mathrm{initial}}$ is the initial bandwidth. For simplicity of the discussion, we assume that $\sigma^2(z)$ is known and replace $\hat{\sigma}^2(Z_i)$ with $\sigma^2(Z_i)$. Hence, in what follows, we define $\Psi_{+}(h)$ and $\Psi_{-}(h)$ as follows:
\begin{align*}
    \Psi_{+}(h) &= \frac{1}{nh^2}  \sum_{i=1}^n 1\{Z_i \geq 0\} K(Z_i/h)^2  r_1(Z_i/h) r_1(Z_i/h)' \sigma^2(Z_i), \\
    \Psi_{-}(h) &= \frac{1}{nh^2}  \sum_{i=1}^n 1\{Z_i < 0\} K(Z_i/h)^2  r_1(Z_i/h) r_1(Z_i/h)' \sigma^2(Z_i).
\end{align*}

To show the convergence of $h_{\mathrm{initial}}\tilde{V}_{CCT}$, we first consider the convergence of $\Gamma_{+}(h)$ and $\Psi_{+}(h)$.

\begin{proposition}\label{prop.Gamma_Psi}
Assume that $Z_1, \ldots, Z_n$ are independent and identically distributed, $Z_i$ has the probability density function $f_Z$ with $f_Z(0)=0$, there exists $\delta > 0$ such that $f_Z(z)$ is continuously differentiable on $[0,\delta]$, $\sigma^2(z)$ is bounded and right-continuous at $z=0$, and $K$ satisfies Assumption \ref{Ass2} for $d=1$. If $n \to \infty$ and $n h^2 \to \infty$, then
\[
h^{-1}\Gamma_{+}(h) = C_{\Gamma,+} + o_p(1) \ \ \text{and} \ \  \Psi_{+}(h) = C_{\Psi,+} + o_p(1),
\]
where $C_{\Gamma,+} := f_{Z}'(0)\int_0^{\infty} z K(z) r_1(z) r_1(z)' dz$, and $C_{\Psi,+} := f_{Z}'(0) \sigma^2(0) \int_0^{\infty} zK(z)^2 r_1(z) r_1(z)'dz$.
\end{proposition}
\begin{proof}
From Lemma S.A.1 in \cite{Calonico.Cattaneo.Titiunik2014}, we obtain
\begin{align*}
    E[\Gamma_{+}(h)] &= \int_0^{\infty} K(z) r_1(z) r_1(z)' f_Z(hz) dz, \\
    E\left[ \left| \Gamma_{+}(h) - E[\Gamma_{+}(h)] \right|^2 \right] &\precsim \frac{1}{nh} \int_0^{\infty} K(z)^2 |r_1(z)|^4 f_Z(hz) dz.
\end{align*}
Because $f_Z(z)$ is continuously differentiable and $f_Z(0)=0$, we have
\begin{align*}
    E[h^{-1}\Gamma_{+}(h)] &= h^{-1} \int_0^{\infty} K(z) r_1(z) r_1(z)' \left\{ f_{Z,+}'(0) h z + o(h) \right\}  dz \\
    &= f_{Z,+}'(0)\int_0^{\infty} z K(z) r_1(z) r_1(z)' dz + o(1), \\
    E\left[ h^{-2} \left| \Gamma_{+}(h) - E[\Gamma_{+}(h)] \right|^2 \right] &\precsim \frac{1}{nh^3} \int_0^{\infty} K(z)^2 |r_1(z)|^4 \left\{ f_{Z,+}'(0) h z + o(h) \right\} dz \\
    &= \frac{1}{nh^2} f_{Z,+}'(0) \int_0^{\infty} z K(z)^2 |r_1(z)|^4  dz + o\left( \frac{1}{nh^2} \right) = o(1).
\end{align*}
Hence, we obtain $h^{-1}\Gamma_{+}(h) = C_{\Gamma,+} + o_p(1)$. Similar to $\Gamma_{+}(h)$, $\Psi_{+}(h)$ satisfies
\begin{align*}
    E[\Psi_{+}(h) ] &= h^{-1} \int_0^{\infty} K(z)^2 r_1(z) r_1(z)' \sigma^2(hz) f_Z(hz) dz, \\
    &= f_{Z}'(0) \int_0^{\infty} z K(z)^2 r_1(z) r_1(z)' \sigma^2(hz) dz + o(1),  \\
    E\left[ \left| \Psi_{+}(h) - E[\Psi_{+}(h) ] \right|^2 \right] &\precsim \frac{1}{nh^3} \int_0^{\infty} K(z)^4 |r_1(z)|^2  \sigma^4(hz) f_Z(hz) dz, \\
    &= \frac{1}{nh^2} f_{Z}'(0) \int_0^{\infty} K(z)^4 |r_1(z)|^2  \sigma^4(hz) dz + o\left( \frac{1}{nh^2} \right) = o(1),
\end{align*}
which implies $\Psi_{+}(h) = C_{\Psi,+} + o_p(1)$.
\end{proof}

Proposition \ref{prop.Gamma_Psi} implies that if $C_{\Gamma,+}$ and $C_{\Psi,+}$ are nonsingular and $h_{\mathrm{initial}}$ satisfies $nh_{\mathrm{initial}} \to \infty$, then we obtain
\begin{align*}
    n h_{\mathrm{initial}}^2 \hat{V}_{+}(h_{\mathrm{initial}}) &= e_1' \left\{ h_{\mathrm{initial}}^{-1}\Gamma_{+}(h_{\mathrm{initial}}) \right\}^{-1} \Psi_{+}(h_{\mathrm{initial}}) \left\{ h_{\mathrm{initial}}^{-1}\Gamma_{+}(h_{\mathrm{initial}}) \right\}^{-1} e_1 \\
    &= e_1' C_{\Gamma,+}^{-1} C_{\Psi,+} C_{\Gamma,+}^{-1} e_1 + o_p(1),
\end{align*}
where $e_1' C_{\Gamma,+}^{-1} C_{\Psi,+} C_{\Gamma,+}^{-1} e_1  > 0$. A similar result holds for $n h_{\mathrm{initial}}^2 \hat{V}_{-}(h_{\mathrm{initial}})$ as well. As a result, $h_{\mathrm{initial}} \tilde{V}_{CCT}$ converges to a positive constant. If $\tilde{B}_{CCT}$ also converges to a positive constant, we obtain
\begin{align*}
    \hat{h}_{CCT} &= C \cdot \left( \frac{h_{\mathrm{initial}} \tilde{V}_{CCT}}{\tilde{B}_{CCT}} \right)^{1/5} h_{\mathrm{initial}}^{-1/5} n^{-1/5} = O_p\left( h_{\mathrm{initial}}^{-1/5} n^{-1/5} \right).
\end{align*}
Hence, the rate of convergence for $\hat{h}_{CCT}$ depends on the initial bandwidth. As with $\hat{h}_{IK}$, if $h_{\mathrm{initial}} = O_p(n^{-1/5})$, then the the convergence rate of $\hat{h}_{CCT}$ is $n^{-4/25}$, which is suboptimal for the two-dimensional problem.

\section{Additional numerical results tables and figures}
\label{sec.standardized_plots}

\begin{table}[H]
\centering
\addtocounter{table}{-4}
\caption{Summary For Bandwidths And Effective Sample Sizes For The Base Designs.}
\input{figures/rd2dim_revision/table_simulation_bands_5000}

\label{tab:MC_result_bands}
\vspace{0.5cm}
\begin{minipage}{0.9\hsize}\footnotesize
 \textit{Notes:} Results are from $10,000$ replication draws of $5,000$ observation samples. pilot represents the pilot bandwidth, $h_1$ is the bandwidth for the axis along with the boundary, and $h_2$ is the bandwidth for the axis orthogonal to the boundary if presented. eff. sample is the effective sample size.
\end{minipage}
\end{table}

\begin{table}[H]
\centering
\addtocounter{table}{-1}
\caption{Simulation Results For Linear Probability Models.}
\input{figures/rd2dim_revision/table_simulation_5000_linprob.tex}

\label{tab:MC_result_LPM}
\vspace{0.5cm}
\begin{minipage}{0.9\hsize}\footnotesize
 \textit{Notes:} Results are from $10,000$ replication draws of $5,000$ observation samples. \textit{rdrobust} is the estimator with the Euclidean distance from the boundary point as the running variable using \textit{rdrobust}; \textit{2D local poly} refers to our preferred different bandwidth estimator \textit{diff bw} and with imposing common bandwidth \textit{common bw}. All the implementations are in \textit{R}. \textit{length} and \textit{coverage} are of generated confidence interval length and coverage rate.
\end{minipage}
\end{table}

\begin{table}[H]
\centering
\addtocounter{table}{-1}
\caption{Summary For Bandwidths And Effective Sample Sizes For LPM.}
\input{figures/rd2dim_revision/table_simulation_bands_5000_linprob}

\label{tab:MC_result_LPM_bands}
\vspace{0.5cm}
\begin{minipage}{0.9\hsize}\footnotesize
 \textit{Notes:} Results are from $10,000$ replication draws of $5,000$ observation samples. pilot represents the pilot bandwidth, $h_1$ is the bandwidth for the axis along with the boundary, and $h_2$ is the bandwidth for the axis orthogonal to the boundary if presented. eff. sample is the effective sample size.
\end{minipage}
\end{table}

\begin{figure}[H]
    \centering 
    \begin{minipage}{0.49\hsize}
    (a) Black Percentage
    \includegraphics[width=\columnwidth]{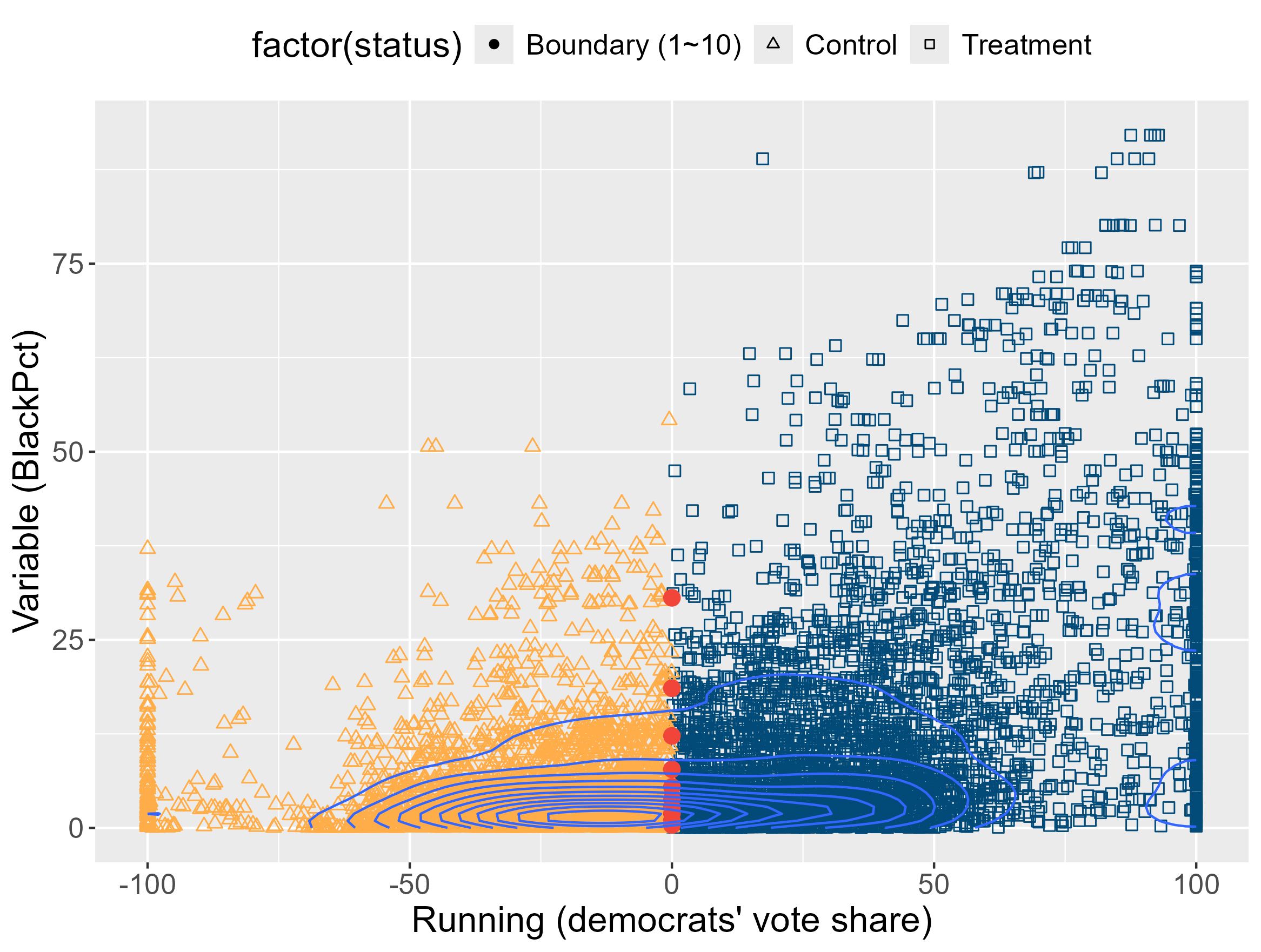}
    \centering 
    \end{minipage}    
    \begin{minipage}{0.49\hsize}
    (b) Foreign Born Percentage
    \includegraphics[width=\columnwidth]{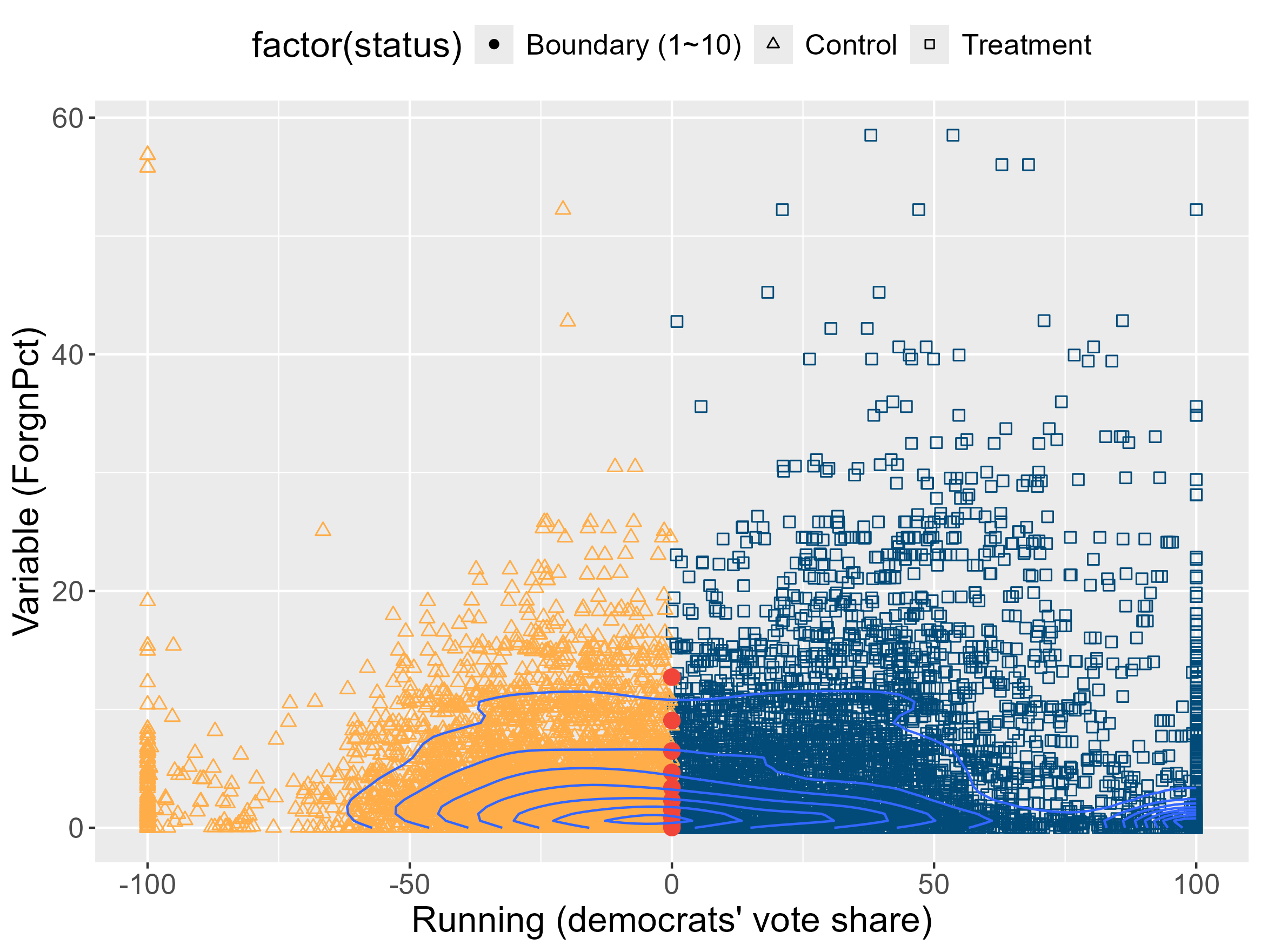}
    \centering
    \end{minipage}
    \begin{minipage}{0.49\hsize}
    (c) Government Worker Percentage
    \includegraphics[width=\columnwidth]{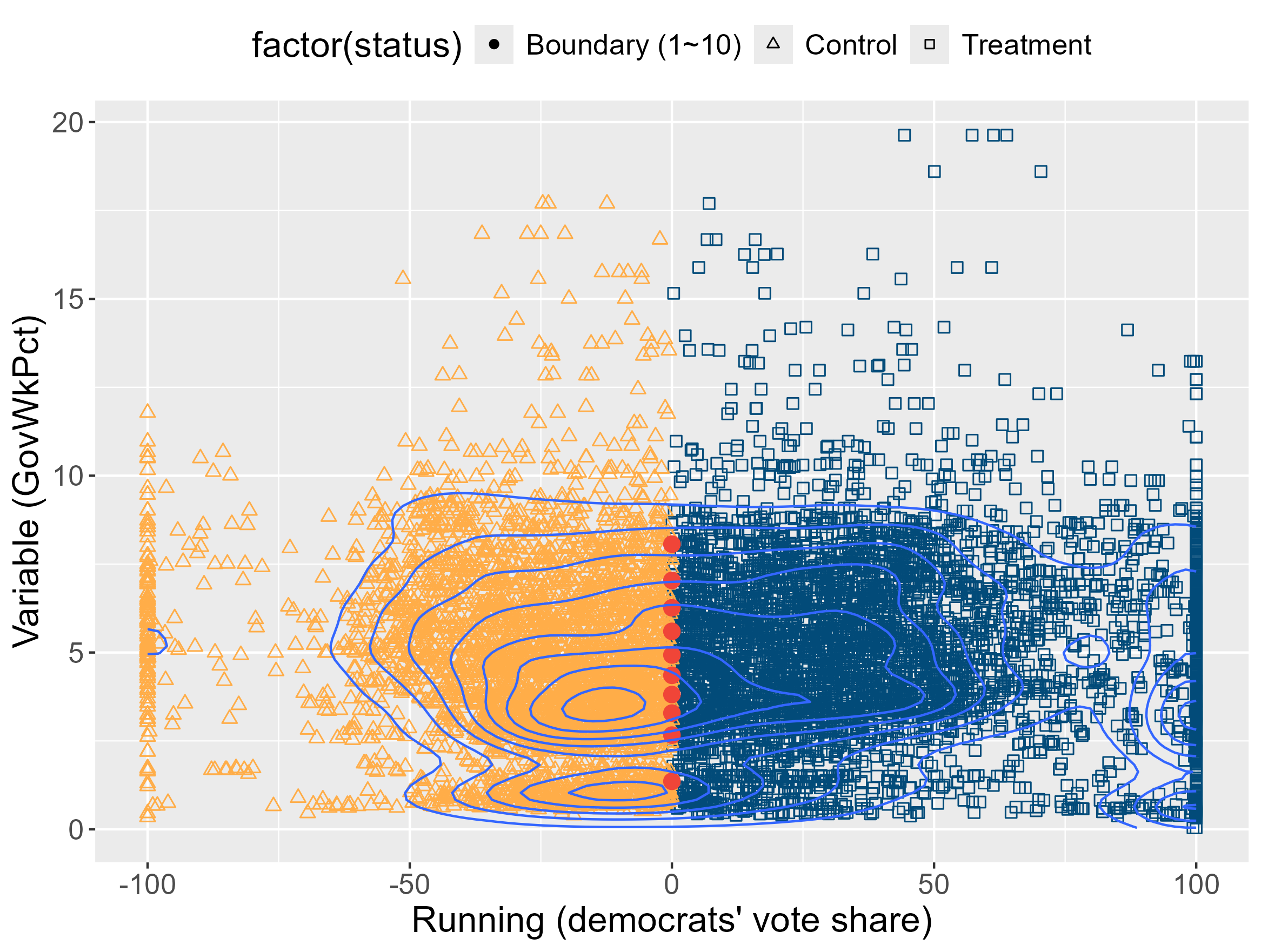}
    \centering 
    \end{minipage}
    \begin{minipage}{0.49\hsize}
    (d) Urban Percentage
    \includegraphics[width=\columnwidth]{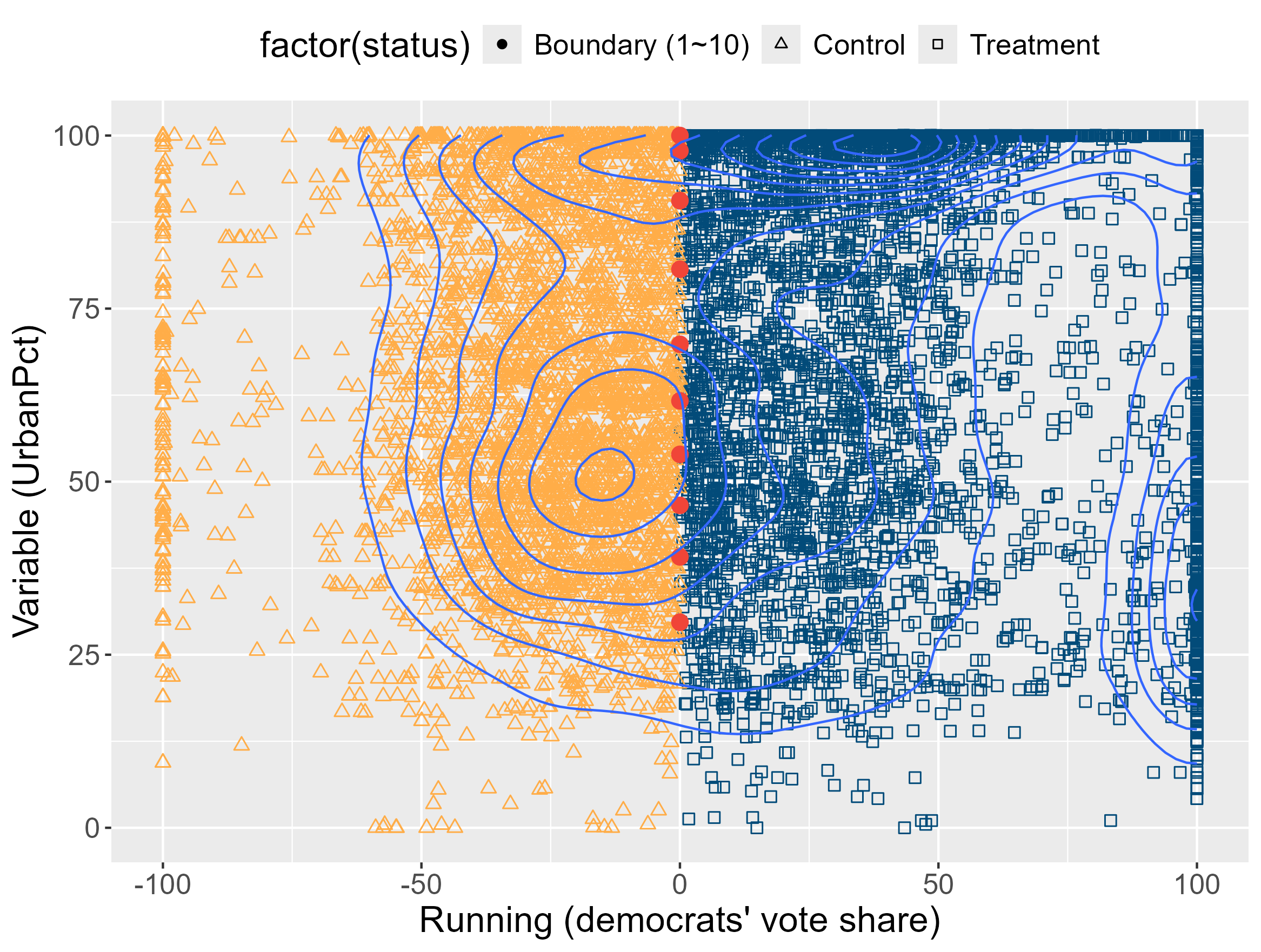}
    \centering 
    \end{minipage}
    \caption{Scatter plots of voting district covariates across different values of the running variable with ten evaluation points over the policy \textit{boundary} for each covariate. (Source: our calculation using \citealp{caugheyDataCitation})}
    \label{fig:shapes_lee_data}
\end{figure}

\begin{figure}[H]
    \centering
    \subfigure[SABER = 0, rdrobust with and without standardizing]{
        \includegraphics[width=0.47\hsize]{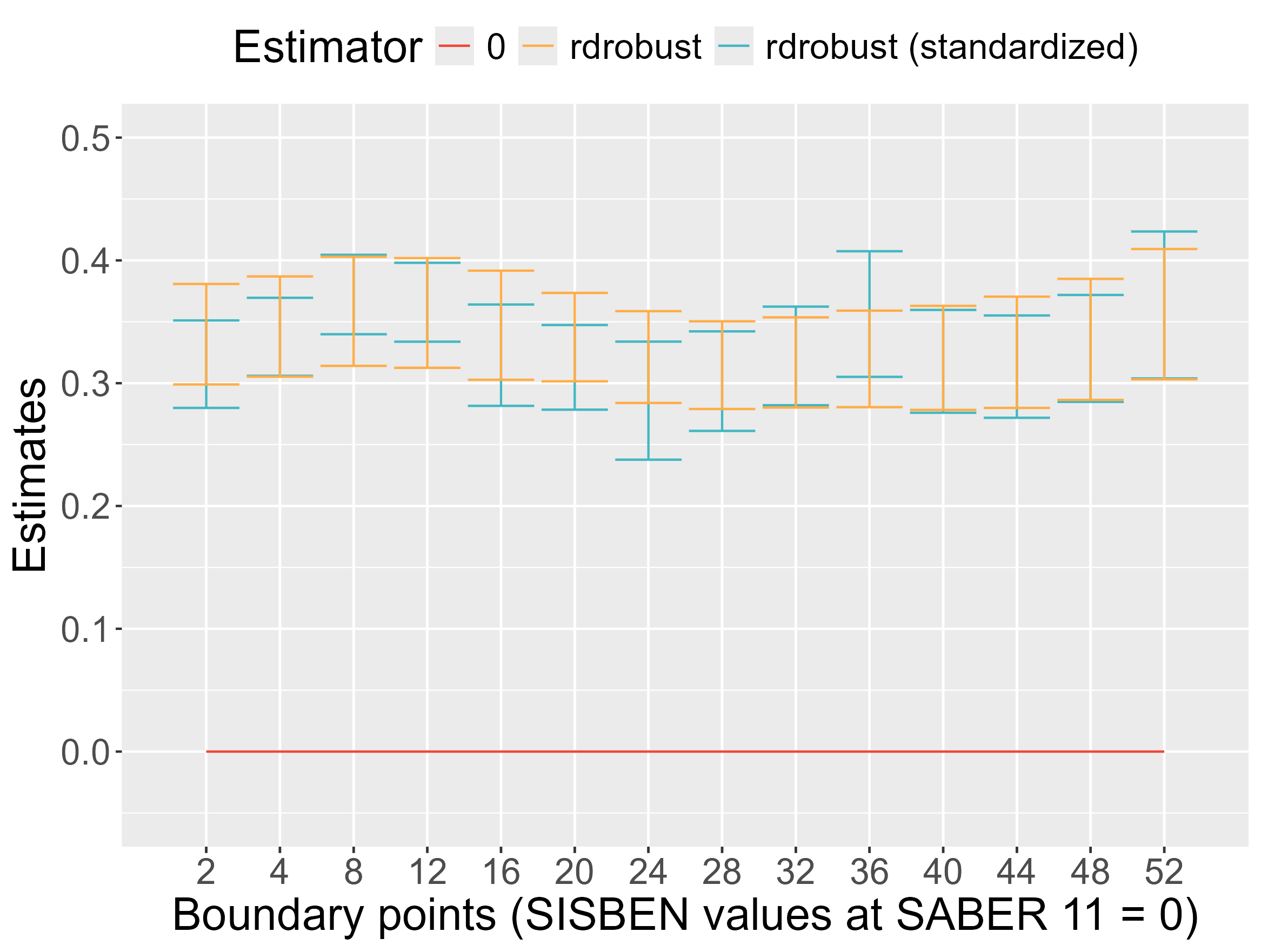}
    }
    \subfigure[SABER = 0, rd2dim with and without standardizing]{
        \includegraphics[width=0.47\hsize]{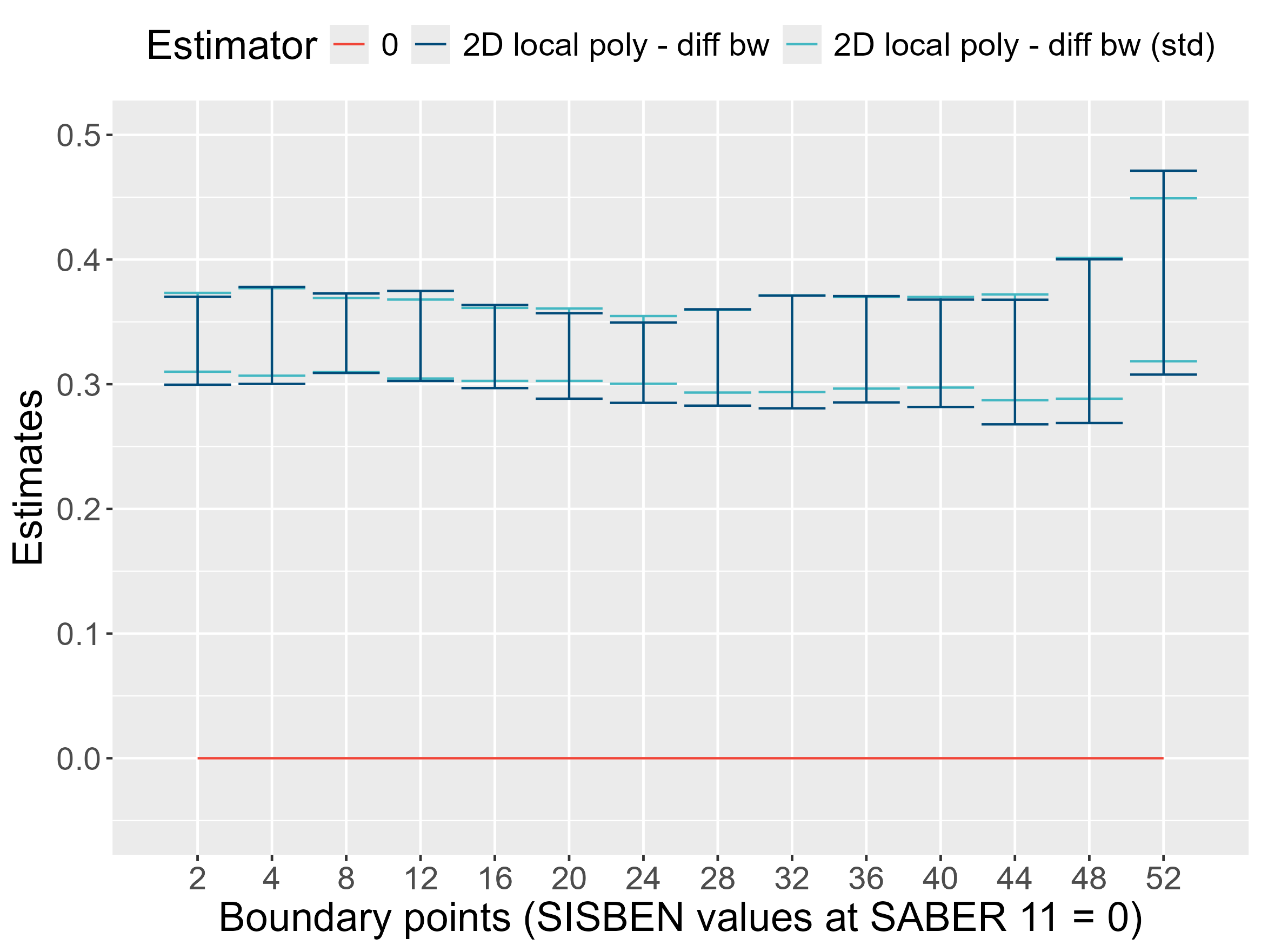}
    }
    \subfigure[SISBEN = 0, rdrobust with and without standardizing]{
        \includegraphics[width=0.47\hsize]{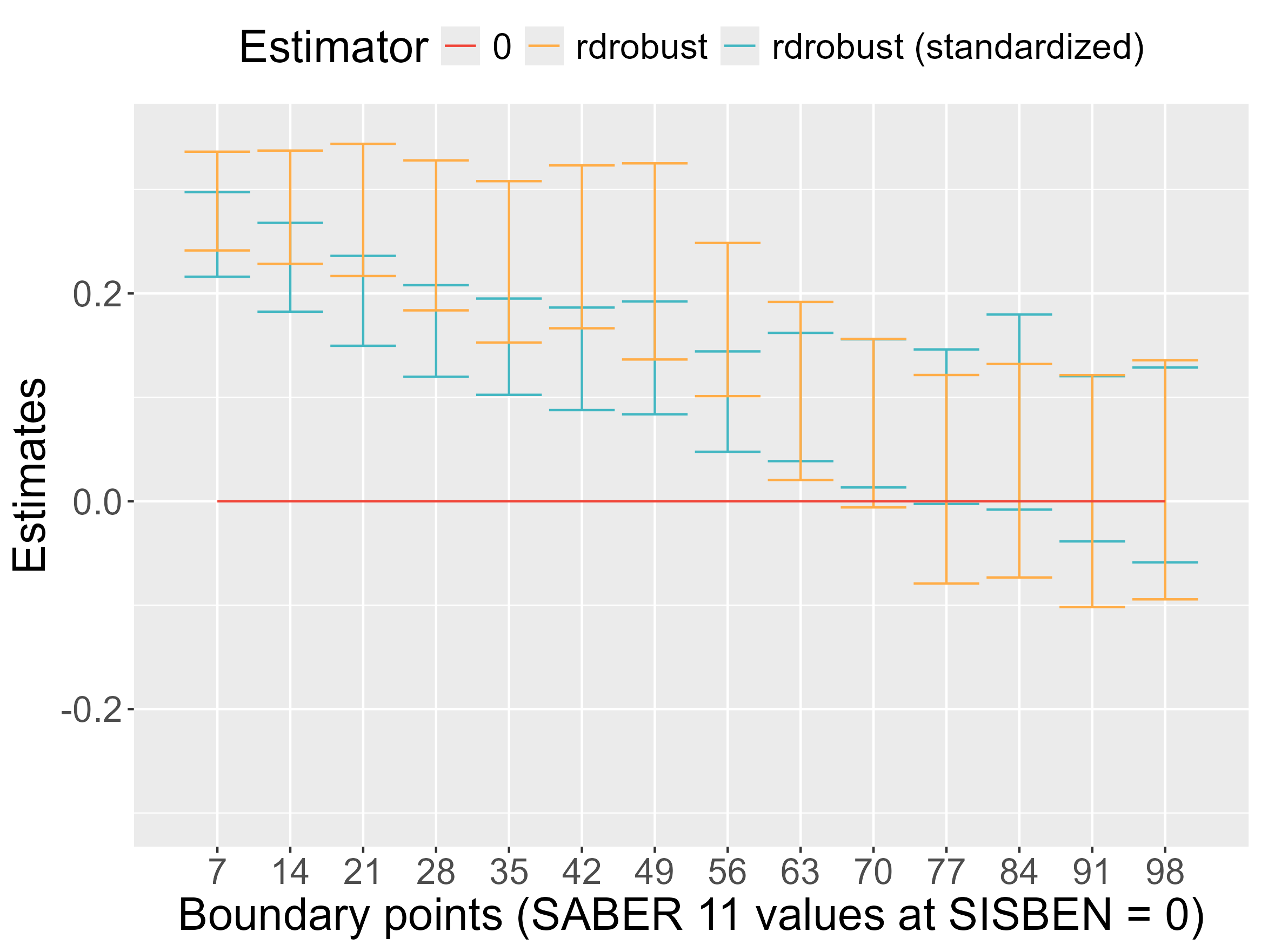}
    }
    \subfigure[SISBEN = 0, rd2dim with and without standardizing]{
        \includegraphics[width=0.47\hsize]{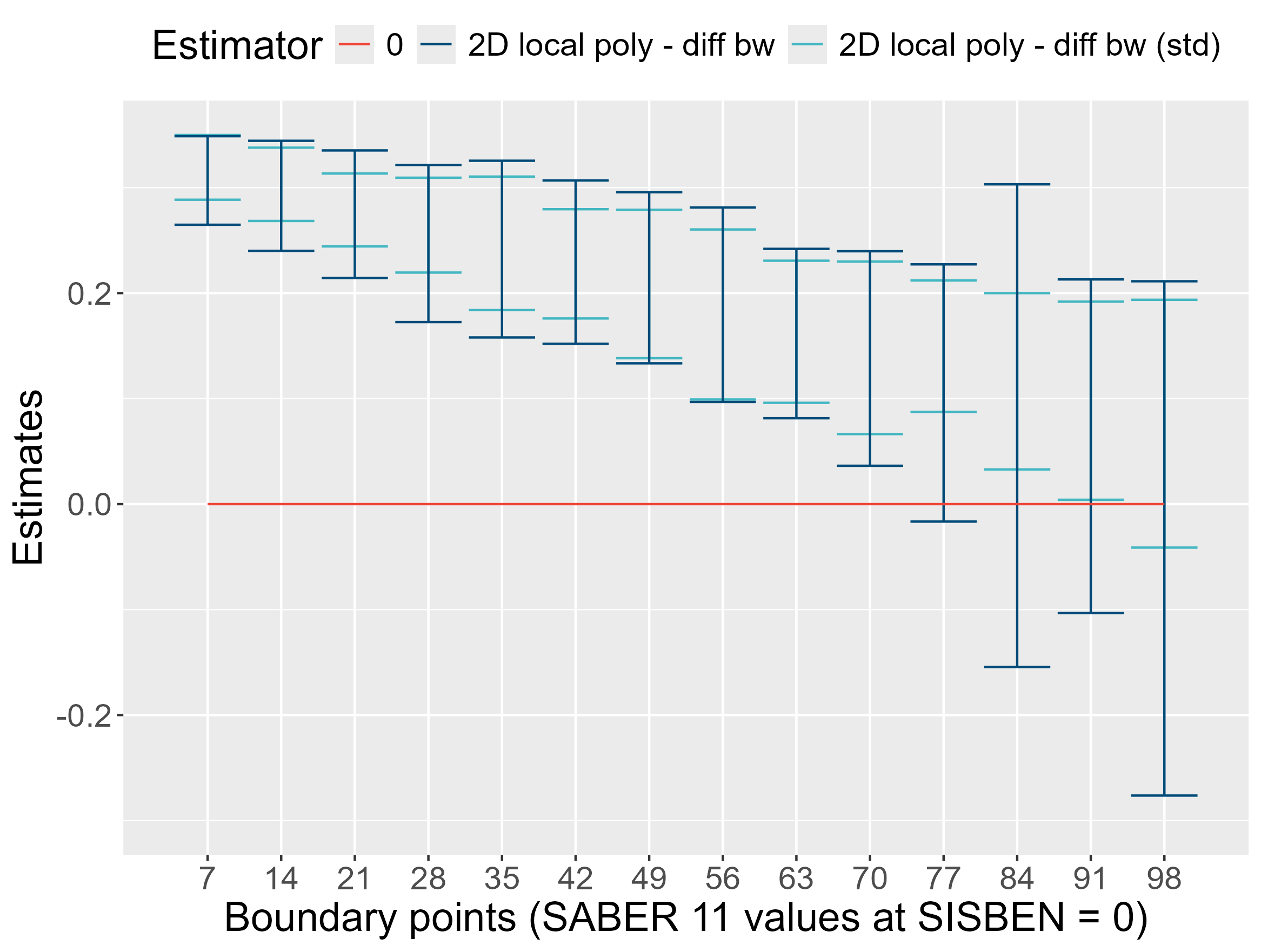}
    }
    \caption{Estimation results over the 28 boundary points comparing two \textit{rdrobust} estimates with and without standardizing scales by their standard deviations for each two axes (Panel (a) and (c)) and our estimator (Panel (b) and (d)).}
    \label{fig:empirical_results_scaled_distance_standardized}
\end{figure}

\addtocounter{table}{-1} 
\input{figures/rd2dim_revision/table_simulation_5000_all_points.tex}
\addtocounter{table}{-8} 

\input{figures/rd2dim_revision/table_simulation_5000_all_points_linprob.tex}

\addtocounter{table}{-8}
\input{figures/rd2dim_revision/table_application_Colombia_bands}

\addtocounter{table}{-4}
\begin{table}[H]
\centering
\caption{Bandwidths for the Lee study (Black Percentage).}
\input{figures/rd2dim_revision/table_application_lee_BlackPct_bands}

\label{tab:MC_result_bands_lee_black}
\begin{minipage}{0.9\hsize}\footnotesize
 \textit{Notes:} The results are for the Lee study with Black percentage variable. Pilot represents the pilot bandwidth, h1 is the
bandwidth for the axis along with the boundary, and h2 is the bandwidth for the axis orthogonal
to the boundary if presented. eff. sample is the effective sample size.
\end{minipage}
\end{table}
\addtocounter{table}{-1}
\begin{table}[H]
\centering
\caption{Bandwidths for the Lee study (Foreign Percentage).}
\input{figures/rd2dim_revision/table_application_lee_ForgnPct_bands}

\label{tab:MC_result_bands_lee_forgn}
\begin{minipage}{0.9\hsize}\footnotesize
 \textit{Notes:} The results are for the Lee study with Foreign percentage variable. Pilot represents the pilot bandwidth, h1 is the
bandwidth for the axis along with the boundary, and h2 is the bandwidth for the axis orthogonal
to the boundary if presented. eff. sample is the effective sample size.
\end{minipage}
\end{table}

\addtocounter{table}{-1}
\begin{table}[H]
\centering
\caption{Bandwidths for the Lee study (Government Worker Percentage).}
\input{figures/rd2dim_revision/table_application_lee_GovWkPct_bands}

\label{tab:MC_result_bands_lee_gvt}
\begin{minipage}{0.9\hsize}\footnotesize
 \textit{Notes:} The results are for the Lee study with Government Worker percentage variable. Pilot represents the pilot bandwidth, h1 is the
bandwidth for the axis along with the boundary, and h2 is the bandwidth for the axis orthogonal
to the boundary if presented. eff. sample is the effective sample size.
\end{minipage}
\end{table}
\addtocounter{table}{-1}
\begin{table}[H]
\centering
\caption{Bandwidths for the Lee study (Urban Percentage).}
\input{figures/rd2dim_revision/table_application_lee_UrbanPct_bands}

\label{tab:MC_result_bands_lee_urban}
\begin{minipage}{0.9\hsize}\footnotesize
 \textit{Notes:} The results are for the Lee study with Urban percentage variable. Pilot represents the pilot bandwidth, h1 is the
bandwidth for the axis along with the boundary, and h2 is the bandwidth for the axis orthogonal
to the boundary if presented. eff. sample is the effective sample size.
\end{minipage}
\end{table}
\end{document}

%% file: figures/rd2dim_revision/table_simulation_5000.tex
\begingroup
\fontsize{12.0pt}{14.4pt}\selectfont
\begin{longtable}{llrrrr}
\toprule
& Estimator & length & bias & coverage & rmse \\ 
\midrule\addlinespace[2.5pt]
Design 1 & rdrobust & 0.413 & -0.003 & 0.931 & 0.113 \\ 
Design 1 & 2D local poly - common bw & 0.211 & 0.029 & 0.942 & 0.054 \\ 
Design 1 & 2D local poly - diff bw & 0.266 & 0.020 & 0.982 & 0.046 \\ 
\midrule\addlinespace[2.5pt]
Design 2 & rdrobust & 0.180 & 0.033 & 0.930 & 0.054 \\ 
Design 2 & 2D local poly - common bw & 0.138 & 0.002 & 0.989 & 0.027 \\ 
Design 2 & 2D local poly - diff bw & 0.133 & 0.006 & 0.985 & 0.026 \\ 
\midrule\addlinespace[2.5pt]
Design 3 & rdrobust & 0.179 & 0.027 & 0.903 & 0.057 \\ 
Design 3 & 2D local poly - common bw & 0.167 & 0.019 & 0.960 & 0.040 \\ 
Design 3 & 2D local poly - diff bw & 0.166 & 0.017 & 0.970 & 0.038 \\ 
\midrule\addlinespace[2.5pt]
Design 4 & rdrobust & 0.292 & 0.023 & 0.934 & 0.086 \\ 
Design 4 & 2D local poly - common bw & 0.353 & 0.032 & 0.980 & 0.080 \\ 
Design 4 & 2D local poly - diff bw & 0.359 & 0.028 & 0.975 & 0.081 \\ 
\bottomrule
\end{longtable}
\endgroup

%% file: figures/rd2dim_revision/table_simulation_rmse_5000_all_points.tex
\begingroup
\fontsize{12.0pt}{14.4pt}\selectfont
\begin{longtable}{lrrrrr}
\toprule
Estimator (rmse) & min & 25\% & 50\% & 75\% & max \\ 
\midrule\addlinespace[2pt]
rdrobust & 0.035 & 0.044 & 0.056 & 0.060 & 0.095 \\ 
2D local poly - common bw & 0.021 & 0.029 & 0.039 & 0.040 & 0.146 \\ 
2D local poly - diff bw & 0.023 & 0.030 & 0.036 & 0.038 & 0.105 \\ 
\bottomrule
\end{longtable}
\endgroup

%% file: figures/rd2dim_revision/table_simulation_coverage_5000_all_points.tex
\begingroup
\fontsize{12.0pt}{14.4pt}\selectfont
\begin{longtable}{lrrrrr}
\toprule
Estimator (coverage) & min & 25\% & 50\% & 75\% & max \\ 
\midrule\addlinespace[2pt]
rdrobust & 0.742 & 0.920 & 0.928 & 0.944 & 0.963 \\ 
2D local poly - common bw & 0.158 & 0.959 & 0.970 & 0.978 & 0.993 \\ 
2D local poly - diff bw & 0.601 & 0.970 & 0.980 & 0.983 & 0.990 \\ 
\bottomrule
\end{longtable}
\endgroup

%% file: figures/rd2dim_revision/table_simulation_length_5000_all_points.tex
\begingroup
\fontsize{12.0pt}{14.4pt}\selectfont
\begin{longtable}{lrrrrr}
\toprule
Estimator (length) & min & 25\% & 50\% & 75\% & max \\ 
\midrule\addlinespace[2pt]
rdrobust & 0.143 & 0.158 & 0.175 & 0.194 & 0.262 \\ 
2D local poly - common bw & 0.111 & 0.122 & 0.163 & 0.166 & 0.200 \\ 
2D local poly - diff bw & 0.125 & 0.137 & 0.163 & 0.169 & 0.218 \\ 
\bottomrule
\end{longtable}
\endgroup

%% file: figures/rd2dim_revision/table_simulation_bias_5000_all_points.tex
\begingroup
\fontsize{12.0pt}{14.4pt}\selectfont
\begin{longtable}{lrrrrr}
\toprule
Estimator (bias) & min & 25\% & 50\% & 75\% & max \\ 
\midrule\addlinespace[2pt]
rdrobust & -0.060 & 0.007 & 0.019 & 0.027 & 0.057 \\ 
2D local poly - common bw & -0.141 & -0.003 & 0.019 & 0.020 & 0.023 \\ 
2D local poly - diff bw & -0.097 & 0.001 & 0.017 & 0.018 & 0.026 \\ 
\bottomrule
\end{longtable}
\endgroup

%% file: control_side_polynomial_7.tex
0.351330594& + (0.0016345305) X + (0.0001058476) X^2 + (8.255\text{e-07}) X^3 + (5.9\text{e-09}) X^4 + (1 \text{e-10}) X^5 \\ + (0.0053400898) Y &  + (2.4132 \text{e-05}) Y^2 - (1.83\text{e-08}) Y^3 - (4\text{e-10}) Y^4 - 0 Y^5 \\ + (4.50874\text{e-05}) X Y &  + (1.0092\text{e-06}) X^2 Y + (3.368\text{e-07}) X Y^2 + (2\text{e-10}) X^2 Y^2 + (8\text{e-10}) X^3 Y + (1.07\text{e-08}) X Y^3

%% file: treatment_side_polynomial_7.tex
0.6585339043& + (0.000775413) X + (5.94362\text{e-05}) X^2 - (1.3635\text{e-06}) X^3 + (4.988\text{e-07}) X^4 + (1.69\text{e-08}) X^5\\ + (0.0032217053) Y &  - (6.65157\text{e-05})Y^2 + (2.97\text{e-06}) Y^3 - (3.79\text{e-08}) Y^4 + (1\text{e-10}) Y^5 \\ - (1.03557\text{e-05}) X Y &  - (4.2481\text{e-06}) X^2 Y + (3.884\text{e-07}) X Y^2 + (4.4\text{e-09}) X^2 Y^2 - (6\text{e-10}) X^3 Y - (1.027\text{e-07}) X Y^3

%% file: control_side_polynomial_13.tex
0.36273926& - (0.0021631216) X + (5.15506\text{e-05}) X^2 + (8.953\text{e-07}) X^3 - (7.4\text{e-09}) X^4 + (1\text{e-10}) X^5 \\ + (0.0046917496) Y &  + (1.61902\text{e-05}) Y^2 - (3.67\text{e-08}) Y^3 - (4\text{e-10}) Y^4 - 0 Y^5 \\ + (1.50884\text{e-05}) X Y &  + (2.408\text{e-07}) X^2 Y + (3.25\text{e-07}) X Y^2 + (2\text{e-10}) X^2 Y^2 + (8\text{e-10}) X^3 Y + (1.07\text{e-08}) X Y^3

%% file: treatment_side_polynomial_13.tex
0.7242674163& - (0.0040502435) X - (0.0004489873) X^2 + (4.78549\text{e-05}) X^3 - (1.5242\text{e-06}) X^4 + (1.69\text{e-08}) X^5\\ + (0.0024425863) Y &  - (7.33327\text{e-05}) Y^2 + (2.9837\text{e-06}) Y^3 - (3.79\text{e-08}) Y^4 + (1\text{e-10}) Y^5 \\ + (1.61465\text{e-05}) X Y &  + (3.1439\text{e-06}) X^2 Y + (1.796\text{e-07}) X Y^2 + (4.4\text{e-09}) X^2 Y^2 - (6\text{e-10}) X^3 Y - (1.027\text{e-07}) X Y^3

%% file: control_side_polynomial_19.tex
0.5206142027& + (0.0052087349) X + (8.183\text{e-06}) X^2 - (8.79\text{e-08}) X^3 - (4\text{e-10}) X^4 - 0 X^5 \\ - (0.0021581664) Y &  + (2.64291\text{e-05}) Y^2 + (1.5009\text{e-06}) Y^3 - (1.18\text{e-08}) Y^4 + (1\text{e-10}) Y^5 \\ + (3.3066\text{e-05}) X Y &  + (3.854\text{e-07}) X^2 Y - (1.5\text{e-09}) X Y^2 + (2\text{e-10}) X^2 Y^2 + (1.07\text{e-08}) X^3 Y + (8\text{e-10}) X Y^3

%% file: treatment_side_polynomial_19.tex
0.7549214382& + (0.0025430669) X + (3.01802\text{e-05}) X^2 - (1.152\text{e-07}) X^3 - (1.75\text{e-08}) X^4 + (1\text{e-10}) X^5\\ + (0.014353943) Y &  - (0.0021086853) Y^2 + (0.0001045443) Y^3 - (2.1986\text{e-06}) Y^4 + (1.69\text{e-08}) Y^5 \\ - (4.90521\text{e-05}) X Y &  + (6.19\text{e-08}) X^2 Y + (5.8515\text{e-06}) X Y^2 + (4.4\text{e-09}) X^2 Y^2 - (1.027\text{e-07}) X^3 Y - (6\text{e-10}) X Y^3

%% file: control_side_polynomial_25.tex
0.7458374267& + (0.0052893523) X - (8.065\text{e-06}) X^2 - (1.737\text{e-07}) X^3 - (6\text{e-10}) X^4 - 0 X^5 \\ - (3.26995\text{e-05}) Y &  + (2.68002\text{e-05}) Y^2 + (1.9491\text{e-06}) Y^3 - (1.18\text{e-08}) Y^4 + (1\text{e-10}) Y^5 \\ + (6.94992\text{e-05}) X Y &  + (4.82\text{e-07}) X^2 Y + (1.92\text{e-08}) X Y^2 + (2\text{e-10}) X^2 Y^2 + (1.07\text{e-08}) X^3 Y + (8\text{e-10}) X Y^3

%% file: treatment_side_polynomial_25.tex
0.8710000105& + (0.0015475707) X - (6.16581\text{e-05}) X^2 - (4.855\text{e-07}) X^3 + (1.31\text{e-08}) X^4 + (1\text{e-10}) X^5\\ + (0.0123605658) Y &  - (0.0018552507) Y^2 + (0.0001002323) Y^3 - (2.1986\text{e-06}) Y^4 + (1.69\text{e-08}) Y^5 \\ - (4.68808\text{e-05}) X Y &  - (1.02\text{e-08}) X^2 Y + (6.2169\text{e-06}) X Y^2 + (4.4\text{e-09}) X^2 Y^2 - (1.027\text{e-07}) X^3 Y - (6\text{e-10}) X Y^3

%% file: figures/rd2dim_revision/table_simulation_bands_5000.tex
\begingroup
\fontsize{12.0pt}{14.4pt}\selectfont
\begin{longtable}{lrrrr}
\toprule
Estimator & pilot & h1 & h2 & eff. sample \\ 
\midrule\addlinespace[2.5pt]
\multicolumn{5}{l}{1} \\[2.5pt] 
\midrule\addlinespace[2.5pt]
rdrobust & 33.1 & 16.1 & - & 125.0 \\ 
2D local poly - common bw & 54.1 & 24.0 & - & 474.2 \\ 
2D local poly - diff bw & 54.1 & 15.7 & 64.6 & 188.1 \\ 
\midrule\addlinespace[2.5pt]
\multicolumn{5}{l}{2} \\[2.5pt] 
\midrule\addlinespace[2.5pt]
rdrobust & 36.4 & 16.7 & - & 481.2 \\ 
2D local poly - common bw & 52.0 & 18.7 & - & 664.0 \\ 
2D local poly - diff bw & 52.0 & 11.1 & 62.8 & 293.3 \\ 
\midrule\addlinespace[2.5pt]
\multicolumn{5}{l}{3} \\[2.5pt] 
\midrule\addlinespace[2.5pt]
rdrobust & 35.9 & 20.3 & - & 616.1 \\ 
2D local poly - common bw & 26.1 & 22.4 & - & 972.7 \\ 
2D local poly - diff bw & 26.1 & 27.1 & 18.0 & 1,418.2 \\ 
\midrule\addlinespace[2.5pt]
\multicolumn{5}{l}{4} \\[2.5pt] 
\midrule\addlinespace[2.5pt]
rdrobust & 55.9 & 32.6 & - & 319.7 \\ 
2D local poly - common bw & 31.6 & 29.9 & - & 395.5 \\ 
2D local poly - diff bw & 31.6 & 32.9 & 25.3 & 581.9 \\ 
\bottomrule
\end{longtable}
\endgroup

%% file: figures/rd2dim_revision/table_simulation_5000_linprob.tex
\begingroup
\fontsize{12.0pt}{14.4pt}\selectfont
\begin{longtable}{lrrrr}
\toprule
Estimator & length & bias & coverage & rmse \\ 
\midrule\addlinespace[2pt]
\multicolumn{5}{l}{1} \\[2pt] 
\midrule\addlinespace[2pt]
rdrobust & 0.786 & 0.007 & 0.935 & 0.248 \\ 
2D local poly - common bw & 0.473 & 0.014 & 0.969 & 0.106 \\ 
2D local poly - diff bw & 0.537 & 0.016 & 0.974 & 0.115 \\ 
\midrule\addlinespace[2pt]
\multicolumn{5}{l}{2} \\[2pt] 
\midrule\addlinespace[2pt]
rdrobust & 0.431 & 0.020 & 0.935 & 0.128 \\ 
2D local poly - common bw & 0.344 & -0.012 & 0.983 & 0.077 \\ 
2D local poly - diff bw & 0.338 & -0.015 & 0.986 & 0.071 \\ 
\midrule\addlinespace[2pt]
\multicolumn{5}{l}{3} \\[2pt] 
\midrule\addlinespace[2pt]
rdrobust & 0.574 & 0.042 & 0.929 & 0.171 \\ 
2D local poly - common bw & 0.408 & 0.032 & 0.989 & 0.083 \\ 
2D local poly - diff bw & 0.410 & 0.031 & 0.989 & 0.083 \\ 
\midrule\addlinespace[2pt]
\multicolumn{5}{l}{4} \\[2pt] 
\midrule\addlinespace[2pt]
rdrobust & 0.865 & 0.016 & 0.926 & 0.258 \\ 
2D local poly - common bw & 0.821 & 0.041 & 0.988 & 0.177 \\ 
2D local poly - diff bw & 0.840 & 0.037 & 0.979 & 0.185 \\ 
\bottomrule
\end{longtable}
\endgroup

%% file: figures/rd2dim_revision/table_simulation_bands_5000_linprob.tex
\begingroup
\fontsize{12.0pt}{14.4pt}\selectfont
\begin{longtable}{lrrrr}
\toprule
Estimator & pilot & h1 & h2 & eff. sample \\ 
\midrule\addlinespace[2.5pt]
\multicolumn{5}{l}{1} \\[2.5pt] 
\midrule\addlinespace[2.5pt]
rdrobust & 56.0 & 30.3 & - & 635.0 \\ 
2D local poly - common bw & 56.3 & 41.7 & - & 1,453.5 \\ 
2D local poly - diff bw & 56.3 & 32.6 & 56.9 & 946.1 \\ 
\midrule\addlinespace[2.5pt]
\multicolumn{5}{l}{2} \\[2.5pt] 
\midrule\addlinespace[2.5pt]
rdrobust & 52.4 & 28.4 & - & 992.9 \\ 
2D local poly - common bw & 54.7 & 39.2 & - & 1,493.3 \\ 
2D local poly - diff bw & 54.7 & 30.8 & 53.2 & 1,155.8 \\ 
\midrule\addlinespace[2.5pt]
\multicolumn{5}{l}{3} \\[2.5pt] 
\midrule\addlinespace[2.5pt]
rdrobust & 39.3 & 22.8 & - & 783.5 \\ 
2D local poly - common bw & 37.0 & 31.7 & - & 1,888.5 \\ 
2D local poly - diff bw & 37.0 & 34.2 & 27.1 & 2,165.8 \\ 
\midrule\addlinespace[2.5pt]
\multicolumn{5}{l}{4} \\[2.5pt] 
\midrule\addlinespace[2.5pt]
rdrobust & 54.7 & 32.3 & - & 328.0 \\ 
2D local poly - common bw & 39.0 & 36.0 & - & 644.3 \\ 
2D local poly - diff bw & 39.0 & 37.5 & 31.3 & 773.9 \\ 
\bottomrule
\end{longtable}
\endgroup

%% file: figures/rd2dim_revision/table_simulation_5000_all_points.tex
\begin{table}[H]
\centering
\caption{Simulation Results For All 30 Points.}
\label{tab:MC_result_all_30_points}
\vspace{0.5cm}
\begingroup
\fontsize{11.0pt}{13.4pt}\selectfont
\begin{longtable}{lrrrr}
\toprule
Estimator & length & bias & coverage & rmse \\ 
\midrule\addlinespace[2pt]
\multicolumn{5}{l}{1} \\[2pt] 
\midrule\addlinespace[2pt]
1. rdrobust & 0.243 & -0.060 & 0.742 & 0.095 \\ 
2. 2D local poly - common bw & 0.200 & -0.141 & 0.158 & 0.146 \\ 
3. 2D local poly - diff bw & 0.218 & -0.097 & 0.601 & 0.105 \\ 
\midrule\addlinespace[2pt]
\multicolumn{5}{l}{2} \\[2pt] 
\midrule\addlinespace[2pt]
1. rdrobust & 0.196 & 0.020 & 0.929 & 0.058 \\ 
2. 2D local poly - common bw & 0.163 & -0.027 & 0.959 & 0.041 \\ 
3. 2D local poly - diff bw & 0.178 & -0.005 & 0.989 & 0.032 \\ 
\midrule\addlinespace[2pt]
\multicolumn{5}{l}{3} \\[2pt] 
\midrule\addlinespace[2pt]
1. rdrobust & 0.170 & 0.031 & 0.898 & 0.055 \\ 
2. 2D local poly - common bw & 0.142 & 0.013 & 0.979 & 0.030 \\ 
3. 2D local poly - diff bw & 0.158 & 0.024 & 0.967 & 0.036 \\ 
\midrule\addlinespace[2pt]
\multicolumn{5}{l}{4} \\[2pt] 
\midrule\addlinespace[2pt]
1. rdrobust & 0.153 & 0.018 & 0.951 & 0.042 \\ 
2. 2D local poly - common bw & 0.132 & 0.020 & 0.952 & 0.032 \\ 
3. 2D local poly - diff bw & 0.148 & 0.026 & 0.953 & 0.036 \\ 
\bottomrule
\end{longtable}
\endgroup
\end{table}
\begin{table}[H]
\begingroup
\begin{longtable}{lrrrr}
\toprule
Estimator & length & bias & coverage & rmse \\ 
\midrule\addlinespace[2pt]
\multicolumn{5}{l}{5} \\[2pt] 
\midrule\addlinespace[2pt]
1. rdrobust & 0.143 & 0.001 & 0.963 & 0.035 \\ 
2. 2D local poly - common bw & 0.125 & 0.013 & 0.974 & 0.027 \\ 
3. 2D local poly - diff bw & 0.141 & 0.016 & 0.973 & 0.029 \\ 
\midrule\addlinespace[2pt]
\multicolumn{5}{l}{6} \\[2pt] 
\midrule\addlinespace[2pt]
1. rdrobust & 0.149 & -0.009 & 0.940 & 0.040 \\ 
2. 2D local poly - common bw & 0.119 & 0.006 & 0.986 & 0.023 \\ 
3. 2D local poly - diff bw & 0.136 & 0.006 & 0.985 & 0.025 \\ 
\midrule\addlinespace[2pt]
\multicolumn{5}{l}{7} \\[2pt] 
\midrule\addlinespace[2pt]
1. rdrobust & 0.158 & -0.004 & 0.929 & 0.044 \\ 
2. 2D local poly - common bw & 0.115 & 0.000 & 0.991 & 0.022 \\ 
3. 2D local poly - diff bw & 0.132 & 0.000 & 0.988 & 0.024 \\ 
\midrule\addlinespace[2pt]
\multicolumn{5}{l}{8} \\[2pt] 
\midrule\addlinespace[2pt]
1. rdrobust & 0.163 & 0.007 & 0.945 & 0.044 \\ 
2. 2D local poly - common bw & 0.112 & -0.002 & 0.993 & 0.021 \\ 
3. 2D local poly - diff bw & 0.129 & 0.001 & 0.990 & 0.023 \\ 
\bottomrule
\end{longtable}
\endgroup
\end{table}
\begin{table}[H]
\begingroup
\begin{longtable}{lrrrr}
\toprule
Estimator & length & bias & coverage & rmse \\ 
\midrule\addlinespace[2pt]
\multicolumn{5}{l}{9} \\[2pt] 
\midrule\addlinespace[2pt]
1. rdrobust & 0.158 & 0.012 & 0.941 & 0.043 \\ 
2. 2D local poly - common bw & 0.111 & -0.003 & 0.993 & 0.021 \\ 
3. 2D local poly - diff bw & 0.127 & 0.002 & 0.990 & 0.023 \\ 
\midrule\addlinespace[2pt]
\multicolumn{5}{l}{10} \\[2pt] 
\midrule\addlinespace[2pt]
1. rdrobust & 0.158 & 0.015 & 0.949 & 0.043 \\ 
2. 2D local poly - common bw & 0.111 & -0.007 & 0.989 & 0.022 \\ 
3. 2D local poly - diff bw & 0.125 & 0.001 & 0.990 & 0.023 \\ 
\midrule\addlinespace[2pt]
\multicolumn{5}{l}{11} \\[2pt] 
\midrule\addlinespace[2pt]
1. rdrobust & 0.168 & 0.014 & 0.956 & 0.044 \\ 
2. 2D local poly - common bw & 0.113 & -0.014 & 0.981 & 0.025 \\ 
3. 2D local poly - diff bw & 0.126 & -0.005 & 0.989 & 0.024 \\ 
\midrule\addlinespace[2pt]
\multicolumn{5}{l}{12} \\[2pt] 
\midrule\addlinespace[2pt]
1. rdrobust & 0.172 & 0.008 & 0.955 & 0.045 \\ 
2. 2D local poly - common bw & 0.116 & -0.024 & 0.945 & 0.033 \\ 
3. 2D local poly - diff bw & 0.127 & -0.014 & 0.969 & 0.029 \\ 
\bottomrule
\end{longtable}
\endgroup
\end{table}
\begin{table}[H]
\begingroup
\begin{longtable}{lrrrr}
\toprule
Estimator & length & bias & coverage & rmse \\ 
\midrule\addlinespace[2pt]
\multicolumn{5}{l}{13} \\[2pt] 
\midrule\addlinespace[2pt]
1. rdrobust & 0.171 & 0.003 & 0.950 & 0.045 \\ 
2. 2D local poly - common bw & 0.121 & -0.030 & 0.900 & 0.039 \\ 
3. 2D local poly - diff bw & 0.132 & -0.021 & 0.940 & 0.034 \\ 
\midrule\addlinespace[2pt]
\multicolumn{5}{l}{14} \\[2pt] 
\midrule\addlinespace[2pt]
1. rdrobust & 0.180 & 0.007 & 0.948 & 0.049 \\ 
2. 2D local poly - common bw & 0.129 & -0.021 & 0.952 & 0.034 \\ 
3. 2D local poly - diff bw & 0.141 & -0.013 & 0.961 & 0.032 \\ 
\midrule\addlinespace[2pt]
\multicolumn{5}{l}{15} \\[2pt] 
\midrule\addlinespace[2pt]
1. rdrobust & 0.194 & 0.020 & 0.925 & 0.057 \\ 
2. 2D local poly - common bw & 0.133 & -0.003 & 0.981 & 0.028 \\ 
3. 2D local poly - diff bw & 0.147 & 0.003 & 0.972 & 0.031 \\ 
\midrule\addlinespace[2pt]
\multicolumn{5}{l}{16} \\[2pt] 
\midrule\addlinespace[2pt]
1. rdrobust & 0.155 & 0.056 & 0.757 & 0.067 \\ 
2. 2D local poly - common bw & 0.191 & 0.019 & 0.964 & 0.045 \\ 
3. 2D local poly - diff bw & 0.197 & 0.017 & 0.970 & 0.046 \\ 
\bottomrule
\end{longtable}
\endgroup
\end{table}
\begin{table}[H]
\begingroup
\begin{longtable}{lrrrr}
\toprule
Estimator & length & bias & coverage & rmse \\ 
\midrule\addlinespace[2pt]
\multicolumn{5}{l}{17} \\[2pt] 
\midrule\addlinespace[2pt]
1. rdrobust & 0.158 & 0.057 & 0.768 & 0.068 \\ 
2. 2D local poly - common bw & 0.176 & 0.021 & 0.959 & 0.043 \\ 
3. 2D local poly - diff bw & 0.179 & 0.019 & 0.963 & 0.043 \\ 
\midrule\addlinespace[2pt]
\multicolumn{5}{l}{18} \\[2pt] 
\midrule\addlinespace[2pt]
1. rdrobust & 0.158 & 0.041 & 0.868 & 0.056 \\ 
2. 2D local poly - common bw & 0.171 & 0.019 & 0.959 & 0.040 \\ 
3. 2D local poly - diff bw & 0.173 & 0.018 & 0.970 & 0.040 \\ 
\midrule\addlinespace[2pt]
\multicolumn{5}{l}{19} \\[2pt] 
\midrule\addlinespace[2pt]
1. rdrobust & 0.194 & 0.027 & 0.922 & 0.058 \\ 
2. 2D local poly - common bw & 0.170 & 0.021 & 0.964 & 0.040 \\ 
3. 2D local poly - diff bw & 0.170 & 0.019 & 0.978 & 0.039 \\ 
\midrule\addlinespace[2pt]
\multicolumn{5}{l}{20} \\[2pt] 
\midrule\addlinespace[2pt]
1. rdrobust & 0.239 & 0.020 & 0.927 & 0.070 \\ 
2. 2D local poly - common bw & 0.170 & 0.020 & 0.966 & 0.040 \\ 
3. 2D local poly - diff bw & 0.169 & 0.018 & 0.981 & 0.038 \\ 
\bottomrule
\end{longtable}
\endgroup
\end{table}
\begin{table}[H]
\begingroup
\begin{longtable}{lrrrr}
\toprule
Estimator & length & bias & coverage & rmse \\ 
\midrule\addlinespace[2pt]
\multicolumn{5}{l}{21} \\[2pt] 
\midrule\addlinespace[2pt]
1. rdrobust & 0.262 & 0.017 & 0.932 & 0.076 \\ 
2. 2D local poly - common bw & 0.169 & 0.020 & 0.966 & 0.040 \\ 
3. 2D local poly - diff bw & 0.169 & 0.018 & 0.980 & 0.038 \\ 
\midrule\addlinespace[2pt]
\multicolumn{5}{l}{22} \\[2pt] 
\midrule\addlinespace[2pt]
1. rdrobust & 0.225 & 0.023 & 0.936 & 0.065 \\ 
2. 2D local poly - common bw & 0.167 & 0.020 & 0.970 & 0.039 \\ 
3. 2D local poly - diff bw & 0.167 & 0.018 & 0.980 & 0.037 \\ 
\midrule\addlinespace[2pt]
\multicolumn{5}{l}{23} \\[2pt] 
\midrule\addlinespace[2pt]
1. rdrobust & 0.179 & 0.032 & 0.906 & 0.058 \\ 
2. 2D local poly - common bw & 0.166 & 0.020 & 0.971 & 0.039 \\ 
3. 2D local poly - diff bw & 0.166 & 0.018 & 0.980 & 0.037 \\ 
\midrule\addlinespace[2pt]
\multicolumn{5}{l}{24} \\[2pt] 
\midrule\addlinespace[2pt]
1. rdrobust & 0.188 & 0.030 & 0.919 & 0.058 \\ 
2. 2D local poly - common bw & 0.165 & 0.019 & 0.973 & 0.038 \\ 
3. 2D local poly - diff bw & 0.165 & 0.018 & 0.982 & 0.037 \\ 
\bottomrule
\end{longtable}
\endgroup
\end{table}
\begin{table}[H]
\begingroup
\begin{longtable}{lrrrr}
\toprule
Estimator & length & bias & coverage & rmse \\ 
\midrule\addlinespace[2pt]
\multicolumn{5}{l}{25} \\[2pt] 
\midrule\addlinespace[2pt]
1. rdrobust & 0.193 & 0.026 & 0.906 & 0.061 \\ 
2. 2D local poly - common bw & 0.165 & 0.020 & 0.970 & 0.040 \\ 
3. 2D local poly - diff bw & 0.165 & 0.018 & 0.978 & 0.038 \\ 
\midrule\addlinespace[2pt]
\multicolumn{5}{l}{26} \\[2pt] 
\midrule\addlinespace[2pt]
1. rdrobust & 0.229 & 0.022 & 0.925 & 0.067 \\ 
2. 2D local poly - common bw & 0.164 & 0.021 & 0.971 & 0.039 \\ 
3. 2D local poly - diff bw & 0.164 & 0.018 & 0.981 & 0.037 \\ 
\midrule\addlinespace[2pt]
\multicolumn{5}{l}{27} \\[2pt] 
\midrule\addlinespace[2pt]
1. rdrobust & 0.201 & 0.027 & 0.933 & 0.059 \\ 
2. 2D local poly - common bw & 0.163 & 0.020 & 0.970 & 0.039 \\ 
3. 2D local poly - diff bw & 0.162 & 0.018 & 0.982 & 0.036 \\ 
\midrule\addlinespace[2pt]
\multicolumn{5}{l}{28} \\[2pt] 
\midrule\addlinespace[2pt]
1. rdrobust & 0.194 & 0.029 & 0.928 & 0.058 \\ 
2. 2D local poly - common bw & 0.162 & 0.019 & 0.973 & 0.039 \\ 
3. 2D local poly - diff bw & 0.162 & 0.017 & 0.984 & 0.036 \\ 
\bottomrule
\end{longtable}
\endgroup
\end{table}
\begin{table}[H]
\begingroup
\begin{longtable}{lrrrr}
\toprule
Estimator & length & bias & coverage & rmse \\ 
\midrule\addlinespace[2pt]
\multicolumn{5}{l}{29} \\[2pt] 
\midrule\addlinespace[2pt]
1. rdrobust & 0.184 & 0.018 & 0.925 & 0.056 \\ 
2. 2D local poly - common bw & 0.161 & 0.020 & 0.969 & 0.038 \\ 
3. 2D local poly - diff bw & 0.163 & 0.018 & 0.982 & 0.037 \\ 
\midrule\addlinespace[2pt]
\multicolumn{5}{l}{30} \\[2pt] 
\midrule\addlinespace[2pt]
1. rdrobust & 0.149 & -0.004 & 0.924 & 0.047 \\ 
2. 2D local poly - common bw & 0.165 & 0.023 & 0.956 & 0.041 \\ 
3. 2D local poly - diff bw & 0.170 & 0.022 & 0.971 & 0.041 \\ 
\bottomrule
\end{longtable}
\endgroup
\vspace{0.5cm}
\centering
\begin{minipage}{0.9\hsize}\footnotesize
 \textit{Notes:} Results are from $5,000$ replication draws of $5,000$ observation samples. \textit{rdrobust} is the estimator with the Euclidean distance from the boundary point as the running variable using \textit{rdrobust}; \textit{2D local poly} refers to our preferred different bandwidth estimator \textit{diff bw} and with imposing common bandwidth \textit{common bw}. All the implementations are in \textit{R}. \textit{length} and \textit{coverage} are of generated confidence interval length and coverage rate.
\end{minipage}
\end{table}

%% file: figures/rd2dim_revision/table_simulation_5000_all_points_linprob.tex
\begin{table}[H]
\centering
\caption{Simulation Results For All 30 Points With Linear Probability Models.}
\label{tab:MC_result_all_30_points_lin_prob}
\vspace{0.5cm}
\begingroup
\fontsize{12.0pt}{14.4pt}\selectfont
\begin{longtable}{lrrrr}
\toprule
Estimator & length & bias & coverage & rmse \\ 
\midrule\addlinespace[2.5pt]
\multicolumn{5}{l}{1} \\[2.5pt] 
\midrule\addlinespace[2.5pt]
1. rdrobust & 0.560 & -0.087 & 0.884 & 0.180 \\ 
2. 2D local poly - common bw & 0.396 & -0.166 & 0.610 & 0.187 \\ 
3. 2D local poly - diff bw & 0.413 & -0.126 & 0.757 & 0.162 \\ 
\midrule\addlinespace[2.5pt]
\multicolumn{5}{l}{2} \\[2.5pt] 
\midrule\addlinespace[2.5pt]
1. rdrobust & 0.484 & 0.018 & 0.937 & 0.136 \\ 
2. 2D local poly - common bw & 0.363 & -0.047 & 0.958 & 0.091 \\ 
3. 2D local poly - diff bw & 0.374 & -0.019 & 0.963 & 0.090 \\ 
\midrule\addlinespace[2.5pt]
\multicolumn{5}{l}{3} \\[2.5pt] 
\midrule\addlinespace[2.5pt]
1. rdrobust & 0.448 & 0.046 & 0.941 & 0.128 \\ 
2. 2D local poly - common bw & 0.345 & 0.002 & 0.980 & 0.076 \\ 
3. 2D local poly - diff bw & 0.354 & 0.021 & 0.972 & 0.084 \\ 
\midrule\addlinespace[2.5pt]
\multicolumn{5}{l}{4} \\[2.5pt] 
\midrule\addlinespace[2.5pt]
1. rdrobust & 0.429 & 0.039 & 0.939 & 0.122 \\ 
2. 2D local poly - common bw & 0.330 & 0.014 & 0.969 & 0.077 \\ 
3. 2D local poly - diff bw & 0.337 & 0.028 & 0.960 & 0.083 \\ 
\bottomrule
\end{longtable}
\endgroup
\end{table}
\begin{table}[H]
\begingroup
\begin{longtable}{lrrrr}
\toprule
Estimator & length & bias & coverage & rmse \\ 
\midrule\addlinespace[2.5pt]
\multicolumn{5}{l}{5} \\[2.5pt] 
\midrule\addlinespace[2.5pt]
1. rdrobust & 0.429 & 0.019 & 0.947 & 0.118 \\ 
2. 2D local poly - common bw & 0.321 & 0.012 & 0.968 & 0.077 \\ 
3. 2D local poly - diff bw & 0.328 & 0.023 & 0.966 & 0.081 \\ 
\midrule\addlinespace[2.5pt]
\multicolumn{5}{l}{6} \\[2.5pt] 
\midrule\addlinespace[2.5pt]
1. rdrobust & 0.434 & 0.000 & 0.949 & 0.120 \\ 
2. 2D local poly - common bw & 0.314 & 0.006 & 0.972 & 0.076 \\ 
3. 2D local poly - diff bw & 0.322 & 0.014 & 0.973 & 0.079 \\ 
\midrule\addlinespace[2.5pt]
\multicolumn{5}{l}{7} \\[2.5pt] 
\midrule\addlinespace[2.5pt]
1. rdrobust & 0.438 & -0.011 & 0.939 & 0.124 \\ 
2. 2D local poly - common bw & 0.308 & -0.002 & 0.972 & 0.076 \\ 
3. 2D local poly - diff bw & 0.316 & 0.005 & 0.971 & 0.078 \\ 
\midrule\addlinespace[2.5pt]
\multicolumn{5}{l}{8} \\[2.5pt] 
\midrule\addlinespace[2.5pt]
1. rdrobust & 0.439 & -0.006 & 0.941 & 0.125 \\ 
2. 2D local poly - common bw & 0.308 & -0.004 & 0.971 & 0.076 \\ 
3. 2D local poly - diff bw & 0.315 & 0.003 & 0.972 & 0.078 \\ 
\bottomrule
\end{longtable}
\endgroup
\end{table}
\begin{table}[H]
\begingroup
\begin{longtable}{lrrrr}
\toprule
Estimator & length & bias & coverage & rmse \\ 
\midrule\addlinespace[2.5pt]
\multicolumn{5}{l}{9} \\[2.5pt] 
\midrule\addlinespace[2.5pt]
1. rdrobust & 0.439 & 0.006 & 0.938 & 0.124 \\ 
2. 2D local poly - common bw & 0.308 & -0.006 & 0.974 & 0.077 \\ 
3. 2D local poly - diff bw & 0.315 & 0.002 & 0.973 & 0.078 \\ 
\midrule\addlinespace[2.5pt]
\multicolumn{5}{l}{10} \\[2.5pt] 
\midrule\addlinespace[2.5pt]
1. rdrobust & 0.440 & 0.021 & 0.939 & 0.128 \\ 
2. 2D local poly - common bw & 0.311 & -0.011 & 0.972 & 0.078 \\ 
3. 2D local poly - diff bw & 0.318 & -0.001 & 0.973 & 0.079 \\ 
\midrule\addlinespace[2.5pt]
\multicolumn{5}{l}{11} \\[2.5pt] 
\midrule\addlinespace[2.5pt]
1. rdrobust & 0.441 & 0.024 & 0.944 & 0.122 \\ 
2. 2D local poly - common bw & 0.317 & -0.020 & 0.973 & 0.079 \\ 
3. 2D local poly - diff bw & 0.324 & -0.009 & 0.973 & 0.079 \\ 
\midrule\addlinespace[2.5pt]
\multicolumn{5}{l}{12} \\[2.5pt] 
\midrule\addlinespace[2.5pt]
1. rdrobust & 0.451 & 0.025 & 0.944 & 0.124 \\ 
2. 2D local poly - common bw & 0.326 & -0.032 & 0.967 & 0.082 \\ 
3. 2D local poly - diff bw & 0.333 & -0.019 & 0.967 & 0.081 \\ 
\bottomrule
\end{longtable}
\endgroup
\end{table}
\begin{table}[H]
\begingroup
\begin{longtable}{lrrrr}
\toprule
Estimator & length & bias & coverage & rmse \\ 
\midrule\addlinespace[2.5pt]
\multicolumn{5}{l}{13} \\[2.5pt] 
\midrule\addlinespace[2.5pt]
1. rdrobust & 0.466 & 0.018 & 0.947 & 0.122 \\ 
2. 2D local poly - common bw & 0.337 & -0.039 & 0.955 & 0.088 \\ 
3. 2D local poly - diff bw & 0.346 & -0.026 & 0.960 & 0.086 \\ 
\midrule\addlinespace[2.5pt]
\multicolumn{5}{l}{14} \\[2.5pt] 
\midrule\addlinespace[2.5pt]
1. rdrobust & 0.495 & 0.031 & 0.939 & 0.135 \\ 
2. 2D local poly - common bw & 0.353 & -0.028 & 0.962 & 0.085 \\ 
3. 2D local poly - diff bw & 0.366 & -0.015 & 0.956 & 0.089 \\ 
\midrule\addlinespace[2.5pt]
\multicolumn{5}{l}{15} \\[2.5pt] 
\midrule\addlinespace[2.5pt]
1. rdrobust & 0.520 & 0.044 & 0.938 & 0.146 \\ 
2. 2D local poly - common bw & 0.366 & -0.010 & 0.970 & 0.084 \\ 
3. 2D local poly - diff bw & 0.383 & 0.002 & 0.959 & 0.093 \\ 
\midrule\addlinespace[2.5pt]
\multicolumn{5}{l}{16} \\[2.5pt] 
\midrule\addlinespace[2.5pt]
1. rdrobust & 0.451 & 0.061 & 0.924 & 0.140 \\ 
2. 2D local poly - common bw & 0.579 & 0.025 & 0.974 & 0.133 \\ 
3. 2D local poly - diff bw & 0.596 & 0.027 & 0.972 & 0.140 \\ 
\bottomrule
\end{longtable}
\endgroup
\end{table}
\begin{table}[H]
\begingroup
\begin{longtable}{lrrrr}
\toprule
Estimator & length & bias & coverage & rmse \\ 
\midrule\addlinespace[2.5pt]
\multicolumn{5}{l}{17} \\[2.5pt] 
\midrule\addlinespace[2.5pt]
1. rdrobust & 0.455 & 0.067 & 0.926 & 0.140 \\ 
2. 2D local poly - common bw & 0.525 & 0.026 & 0.979 & 0.117 \\ 
3. 2D local poly - diff bw & 0.536 & 0.026 & 0.981 & 0.119 \\ 
\midrule\addlinespace[2.5pt]
\multicolumn{5}{l}{18} \\[2.5pt] 
\midrule\addlinespace[2.5pt]
1. rdrobust & 0.470 & 0.056 & 0.929 & 0.139 \\ 
2. 2D local poly - common bw & 0.503 & 0.027 & 0.980 & 0.112 \\ 
3. 2D local poly - diff bw & 0.507 & 0.025 & 0.981 & 0.112 \\ 
\midrule\addlinespace[2.5pt]
\multicolumn{5}{l}{19} \\[2.5pt] 
\midrule\addlinespace[2.5pt]
1. rdrobust & 0.511 & 0.045 & 0.945 & 0.148 \\ 
2. 2D local poly - common bw & 0.493 & 0.028 & 0.981 & 0.109 \\ 
3. 2D local poly - diff bw & 0.490 & 0.025 & 0.984 & 0.105 \\ 
\midrule\addlinespace[2.5pt]
\multicolumn{5}{l}{20} \\[2.5pt] 
\midrule\addlinespace[2.5pt]
1. rdrobust & 0.594 & 0.047 & 0.937 & 0.172 \\ 
2. 2D local poly - common bw & 0.484 & 0.028 & 0.982 & 0.107 \\ 
3. 2D local poly - diff bw & 0.479 & 0.025 & 0.986 & 0.103 \\ 
\bottomrule
\end{longtable}
\endgroup
\end{table}
\begin{table}[H]
\begingroup
\begin{longtable}{lrrrr}
\toprule
Estimator & length & bias & coverage & rmse \\ 
\midrule\addlinespace[2.5pt]
\multicolumn{5}{l}{21} \\[2.5pt] 
\midrule\addlinespace[2.5pt]
1. rdrobust & 0.631 & 0.041 & 0.946 & 0.177 \\ 
2. 2D local poly - common bw & 0.474 & 0.026 & 0.984 & 0.103 \\ 
3. 2D local poly - diff bw & 0.468 & 0.023 & 0.987 & 0.099 \\ 
\midrule\addlinespace[2.5pt]
\multicolumn{5}{l}{22} \\[2.5pt] 
\midrule\addlinespace[2.5pt]
1. rdrobust & 0.554 & 0.048 & 0.943 & 0.162 \\ 
2. 2D local poly - common bw & 0.464 & 0.027 & 0.984 & 0.102 \\ 
3. 2D local poly - diff bw & 0.457 & 0.023 & 0.986 & 0.099 \\ 
\midrule\addlinespace[2.5pt]
\multicolumn{5}{l}{23} \\[2.5pt] 
\midrule\addlinespace[2.5pt]
1. rdrobust & 0.516 & 0.037 & 0.943 & 0.151 \\ 
2. 2D local poly - common bw & 0.449 & 0.026 & 0.987 & 0.097 \\ 
3. 2D local poly - diff bw & 0.443 & 0.024 & 0.990 & 0.094 \\ 
\midrule\addlinespace[2.5pt]
\multicolumn{5}{l}{24} \\[2.5pt] 
\midrule\addlinespace[2.5pt]
1. rdrobust & 0.509 & 0.038 & 0.932 & 0.156 \\ 
2. 2D local poly - common bw & 0.434 & 0.025 & 0.983 & 0.095 \\ 
3. 2D local poly - diff bw & 0.429 & 0.023 & 0.985 & 0.092 \\ 
\bottomrule
\end{longtable}
\endgroup
\end{table}
\begin{table}[H]
\begingroup
\begin{longtable}{lrrrr}
\toprule
Estimator & length & bias & coverage & rmse \\ 
\midrule\addlinespace[2.5pt]
\multicolumn{5}{l}{25} \\[2.5pt] 
\midrule\addlinespace[2.5pt]
1. rdrobust & 0.468 & 0.023 & 0.936 & 0.137 \\ 
2. 2D local poly - common bw & 0.421 & 0.022 & 0.986 & 0.091 \\ 
3. 2D local poly - diff bw & 0.416 & 0.020 & 0.988 & 0.088 \\ 
\midrule\addlinespace[2.5pt]
\multicolumn{5}{l}{26} \\[2.5pt] 
\midrule\addlinespace[2.5pt]
1. rdrobust & 0.482 & 0.024 & 0.937 & 0.136 \\ 
2. 2D local poly - common bw & 0.407 & 0.023 & 0.989 & 0.088 \\ 
3. 2D local poly - diff bw & 0.403 & 0.021 & 0.990 & 0.085 \\ 
\midrule\addlinespace[2.5pt]
\multicolumn{5}{l}{27} \\[2.5pt] 
\midrule\addlinespace[2.5pt]
1. rdrobust & 0.437 & 0.035 & 0.944 & 0.121 \\ 
2. 2D local poly - common bw & 0.393 & 0.022 & 0.987 & 0.084 \\ 
3. 2D local poly - diff bw & 0.389 & 0.020 & 0.988 & 0.081 \\ 
\midrule\addlinespace[2.5pt]
\multicolumn{5}{l}{28} \\[2.5pt] 
\midrule\addlinespace[2.5pt]
1. rdrobust & 0.403 & 0.040 & 0.944 & 0.113 \\ 
2. 2D local poly - common bw & 0.382 & 0.022 & 0.985 & 0.083 \\ 
3. 2D local poly - diff bw & 0.379 & 0.020 & 0.989 & 0.081 \\ 
\bottomrule
\end{longtable}
\endgroup
\end{table}
\begin{table}[H]
\begingroup
\begin{longtable}{lrrrr}
\toprule
Estimator & length & bias & coverage & rmse \\ 
\midrule\addlinespace[2.5pt]
\multicolumn{5}{l}{29} \\[2.5pt] 
\midrule\addlinespace[2.5pt]
1. rdrobust & 0.370 & 0.019 & 0.942 & 0.105 \\ 
2. 2D local poly - common bw & 0.373 & 0.023 & 0.986 & 0.082 \\ 
3. 2D local poly - diff bw & 0.375 & 0.023 & 0.988 & 0.081 \\ 
\midrule\addlinespace[2.5pt]
\multicolumn{5}{l}{30} \\[2.5pt] 
\midrule\addlinespace[2.5pt]
1. rdrobust & 0.334 & -0.009 & 0.940 & 0.097 \\ 
2. 2D local poly - common bw & 0.378 & 0.024 & 0.980 & 0.086 \\ 
3. 2D local poly - diff bw & 0.386 & 0.024 & 0.981 & 0.087 \\ 
\bottomrule
\end{longtable}
\endgroup
\vspace{0.5cm}
\centering
\begin{minipage}{0.9\hsize}\footnotesize
 \textit{Notes:} Results are from $5,000$ replication draws of $5,000$ observation samples. \textit{rdrobust} is the estimator with the Euclidean distance from the boundary point as the running variable using \textit{rdrobust}; \textit{2D local poly} refers to our preferred different bandwidth estimator \textit{diff bw} and with imposing common bandwidth \textit{common bw}. All the implementations are in \textit{R}. \textit{length} and \textit{coverage} are of generated confidence interval length and coverage rate.
\end{minipage}
\end{table}

%% file: figures/rd2dim_revision/table_application_Colombia_bands.tex
\begin{table}[H]
\centering
\caption{Bandwidths and effective sample sizes for the Colombian Data.}
\begingroup
\fontsize{12.0pt}{14.4pt}\selectfont
\begin{longtable}{rlrrrr}
\toprule
point & Estimator & pilot & h1 & h2 & eff. sample \\ 
\midrule\addlinespace[2.5pt]
2 & 2D local poly - diff bw & 63.1 & 12.7 & 56.5 & 1,673 \\ 
2 & rdrobust & 74.5 & 63.9 & - & 111,160 \\ 
3 & 2D local poly - diff bw & 66.9 & 14.5 & 56.0 & 3,634 \\ 
3 & rdrobust & 67.2 & 59.5 & - & 101,745 \\ 
4 & 2D local poly - diff bw & 63.0 & 14.3 & 52.6 & 5,222 \\ 
4 & rdrobust & 62.4 & 53.9 & - & 87,977 \\ 
5 & 2D local poly - diff bw & 62.1 & 14.4 & 51.2 & 7,046 \\ 
5 & rdrobust & 58.4 & 46.7 & - & 69,331 \\ 
6 & 2D local poly - diff bw & 65.7 & 15.0 & 52.0 & 9,653 \\ 
6 & rdrobust & 57.2 & 43.4 & - & 65,304 \\ 
7 & 2D local poly - diff bw & 69.1 & 17.1 & 63.9 & 14,069 \\ 
7 & rdrobust & 57.9 & 42.1 & - & 67,283 \\ 
8 & 2D local poly - diff bw & 67.2 & 15.2 & 54.5 & 12,588 \\ 
8 & rdrobust & 57.2 & 41.0 & - & 69,059 \\ 
\bottomrule
\end{longtable}
\endgroup
\end{table}
\begin{table}[H]
\centering
\begingroup
\fontsize{12.0pt}{14.4pt}\selectfont
\begin{longtable}{rlrrrr}
\toprule
point & Estimator & pilot & h1 & h2 & eff. sample \\ 
\midrule\addlinespace[2.5pt]
9 & 2D local poly - diff bw & 62.1 & 15.1 & 48.1 & 13,649 \\ 
9 & rdrobust & 55.5 & 36.9 & - & 60,207 \\ 
10 & 2D local poly - diff bw & 67.4 & 15.7 & 51.0 & 15,223 \\ 
10 & rdrobust & 66.1 & 32.1 & - & 49,120 \\ 
11 & 2D local poly - diff bw & 68.3 & 16.1 & 52.4 & 17,026 \\ 
11 & rdrobust & 46.5 & 26.1 & - & 34,659 \\ 
12 & 2D local poly - diff bw & 72.4 & 16.7 & 60.5 & 17,939 \\ 
12 & rdrobust & 45.7 & 24.4 & - & 30,959 \\ 
13 & 2D local poly - diff bw & 69.4 & 16.4 & 53.2 & 17,697 \\ 
13 & rdrobust & 45.0 & 23.4 & - & 28,164 \\ 
14 & 2D local poly - diff bw & 73.8 & 17.4 & 62.6 & 19,660 \\ 
14 & rdrobust & 49.8 & 27.1 & - & 35,579 \\ 
15 & 2D local poly - diff bw & 69.7 & 17.3 & 54.1 & 19,134 \\ 
15 & rdrobust & 52.9 & 29.0 & - & 39,069 \\ 
\bottomrule
\end{longtable}
\endgroup
\end{table}
\begin{table}[H]
\centering
\begingroup
\fontsize{12.0pt}{14.4pt}\selectfont
\begin{longtable}{rlrrrr}
\toprule
point & Estimator & pilot & h1 & h2 & eff. sample \\ 
\midrule\addlinespace[2.5pt]
16 & 2D local poly - diff bw & 34.8 & 32.9 & 18.3 & 46,252 \\ 
16 & rdrobust & 45.6 & 28.1 & - & 28,473 \\ 
17 & 2D local poly - diff bw & 36.3 & 21.0 & 25.9 & 15,890 \\ 
17 & rdrobust & 45.1 & 25.0 & - & 18,361 \\ 
18 & 2D local poly - diff bw & 35.5 & 17.6 & 26.6 & 9,052 \\ 
18 & rdrobust & 43.7 & 25.1 & - & 14,397 \\ 
19 & 2D local poly - diff bw & 38.1 & 19.4 & 26.9 & 8,278 \\ 
19 & rdrobust & 42.0 & 28.5 & - & 13,989 \\ 
20 & 2D local poly - diff bw & 42.1 & 31.9 & 25.9 & 15,660 \\ 
20 & rdrobust & 42.4 & 28.8 & - & 10,661 \\ 
21 & 2D local poly - diff bw & 38.5 & 27.0 & 24.8 & 8,809 \\ 
21 & rdrobust & 50.2 & 30.6 & - & 8,866 \\ 
22 & 2D local poly - diff bw & 40.8 & 26.6 & 25.2 & 6,111 \\ 
22 & rdrobust & 45.6 & 31.4 & - & 6,844 \\ 
\bottomrule
\end{longtable}
\endgroup
\end{table}
\begin{table}[H]
\centering
\begingroup
\fontsize{12.0pt}{14.4pt}\selectfont
\begin{longtable}{rlrrrr}
\toprule
point & Estimator & pilot & h1 & h2 & eff. sample \\ 
\midrule\addlinespace[2.5pt]
23 & 2D local poly - diff bw & 43.3 & 32.2 & 24.8 & 6,611 \\ 
23 & rdrobust & 64.8 & 49.3 & - & 13,769 \\ 
24 & 2D local poly - diff bw & 47.2 & 41.5 & 26.3 & 8,296 \\ 
24 & rdrobust & 63.5 & 50.9 & - & 10,881 \\ 
25 & 2D local poly - diff bw & 49.1 & 49.7 & 28.5 & 8,951 \\ 
25 & rdrobust & 74.3 & 59.8 & - & 12,294 \\ 
26 & 2D local poly - diff bw & 47.0 & 37.2 & 30.3 & 3,318 \\ 
26 & rdrobust & 68.5 & 45.5 & - & 4,265 \\ 
27 & 2D local poly - diff bw & 37.4 & 27.0 & 46.3 & 1,197 \\ 
27 & rdrobust & 78.1 & 58.6 & - & 6,055 \\ 
28 & 2D local poly - diff bw & 36.2 & 65.8 & 14.8 & 7,152 \\ 
28 & rdrobust & 80.7 & 59.3 & - & 4,499 \\ 
29 & 2D local poly - diff bw & 36.5 & 36.5 & 14.6 & 1,071 \\ 
29 & rdrobust & 93.7 & 84.2 & - & 11,015 \\ 
\bottomrule
\end{longtable}
\endgroup
\label{tab:MC_result_bands}
\vspace{0.5cm}
\begin{minipage}{0.9\hsize}\footnotesize
 \textit{Notes:} The bandwidths and effective sample sizes for each evaluation points in the Colombian study. The points from 2 through 15 represents the SABER 11 = 0 boundary from SISBEN values 52 to SISBEN values 2; the points from 16 through 29 represents the SISBEN = 0 boundary from SABER 11 values 7 to SABER 11 values 98. Pilot represents the pilot bandwidth, $h_1$ is the bandwidth for the axis along with the boundary, and $h_2$ is the bandwidth for the axis orthogonal to the boundary if presented. eff. sample is the effective sample size.
\end{minipage}
\end{table}

%% file: figures/rd2dim_revision/table_application_lee_BlackPct_bands.tex
\begingroup
\fontsize{12.0pt}{14.4pt}\selectfont
\begin{longtable}{rlrrrr}
\toprule
point & Estimator & pilot & h1 & h2 & eff. sample \\ 
\midrule
1 & 2D local poly - diff bw & 39.1 & 31.7 & 18.8 & 4,429 \\ 
1 & rdrobust & 37.0 & 17.8 & - & 2,276 \\ 
2 & 2D local poly - diff bw & 37.1 & 28.7 & 17.6 & 4,022 \\ 
2 & rdrobust & 36.9 & 17.9 & - & 2,319 \\ 
3 & 2D local poly - diff bw & 34.6 & 25.0 & 16.1 & 3,536 \\ 
3 & rdrobust & 35.9 & 17.6 & - & 2,304 \\ 
4 & 2D local poly - diff bw & 34.8 & 25.6 & 16.2 & 3,648 \\ 
4 & rdrobust & 36.0 & 17.4 & - & 2,309 \\ 
5 & 2D local poly - diff bw & 35.0 & 25.9 & 16.2 & 3,696 \\ 
5 & rdrobust & 36.5 & 16.6 & - & 2,233 \\ 
6 & 2D local poly - diff bw & 41.1 & 62.4 & 17.0 & 6,961 \\ 
6 & rdrobust & 36.0 & 15.7 & - & 2,127 \\ 
7 & 2D local poly - diff bw & 35.1 & 33.2 & 15.4 & 4,687 \\ 
7 & rdrobust & 36.4 & 15.2 & - & 2,018 \\ 
8 & 2D local poly - diff bw & 38.5 & 58.4 & 15.3 & 6,834 \\ 
8 & rdrobust & 41.7 & 18.9 & - & 2,420 \\ 
9 & 2D local poly - diff bw & 39.5 & 44.8 & 17.4 & 5,961 \\ 
9 & rdrobust & 32.5 & 14.4 & - & 566 \\ 
10 & 2D local poly - diff bw & 41.0 & 35.6 & 18.2 & 5,034 \\ 
10 & rdrobust & 32.0 & 20.3 & - & 379 \\ 
\bottomrule
\end{longtable}
\endgroup

%% file: figures/rd2dim_revision/table_application_lee_ForgnPct_bands.tex
\begingroup
\fontsize{12.0pt}{14.4pt}\selectfont
\begin{longtable}{rlrrrr}
\toprule
point & Estimator & pilot & h1 & h2 & eff. sample \\ 
1 & 2D local poly - diff bw & 42.5 & 13.6 & 21.4 & 1,941 \\ 
1 & rdrobust & 32.2 & 16.9 & - & 2,285 \\ 
2 & 2D local poly - diff bw & 38.5 & 11.6 & 19.4 & 1,650 \\ 
2 & rdrobust & 32.6 & 17.3 & - & 2,361 \\ 
3 & 2D local poly - diff bw & 41.0 & 12.3 & 19.2 & 1,746 \\ 
3 & rdrobust & 32.2 & 17.6 & - & 2,417 \\ 
4 & 2D local poly - diff bw & 41.1 & 12.5 & 19.2 & 1,784 \\ 
4 & rdrobust & 32.3 & 17.5 & - & 2,427 \\ 
5 & 2D local poly - diff bw & 49.4 & 15.3 & 22.2 & 2,249 \\ 
5 & rdrobust & 31.9 & 17.0 & - & 2,376 \\ 
6 & 2D local poly - diff bw & 52.8 & 16.5 & 22.9 & 2,424 \\ 
6 & rdrobust & 32.7 & 16.4 & - & 2,315 \\ 
7 & 2D local poly - diff bw & 49.7 & 14.7 & 20.0 & 2,174 \\ 
7 & rdrobust & 33.1 & 15.7 & - & 2,230 \\ 
8 & 2D local poly - diff bw & 44.0 & 15.8 & 20.9 & 2,338 \\ 
8 & rdrobust & 35.1 & 15.7 & - & 2,209 \\ 
9 & 2D local poly - diff bw & 41.4 & 16.5 & 20.4 & 2,453 \\ 
9 & rdrobust & 42.1 & 16.8 & - & 2,296 \\ 
10 & 2D local poly - diff bw & 39.4 & 17.5 & 20.2 & 2,595 \\ 
10 & rdrobust & 45.2 & 20.5 & - & 2,702 \\ 
\bottomrule
\end{longtable}
\endgroup

%% file: figures/rd2dim_revision/table_application_lee_GovWkPct_bands.tex
\begingroup
\fontsize{12.0pt}{14.4pt}\selectfont
\begin{longtable}{rlrrrr}
\toprule
point & Estimator & pilot & h1 & h2 & eff. sample \\ 
1 & 2D local poly - diff bw & 13.3 & 4.6 & 10.8 & 517 \\ 
1 & rdrobust & 28.4 & 12.7 & - & 1,809 \\ 
2 & 2D local poly - diff bw & 13.7 & 4.9 & 11.7 & 683 \\ 
2 & rdrobust & 28.5 & 14.4 & - & 2,119 \\ 
3 & 2D local poly - diff bw & 12.8 & 4.6 & 9.5 & 667 \\ 
3 & rdrobust & 29.8 & 15.2 & - & 2,249 \\ 
4 & 2D local poly - diff bw & 12.5 & 4.7 & 9.1 & 696 \\ 
4 & rdrobust & 30.6 & 15.9 & - & 2,332 \\ 
5 & 2D local poly - diff bw & 12.8 & 4.9 & 9.3 & 743 \\ 
5 & rdrobust & 31.5 & 16.5 & - & 2,433 \\ 
6 & 2D local poly - diff bw & 13.5 & 5.4 & 10.8 & 821 \\ 
6 & rdrobust & 32.3 & 17.3 & - & 2,542 \\ 
7 & 2D local poly - diff bw & 13.6 & 5.6 & 11.1 & 857 \\ 
7 & rdrobust & 32.8 & 17.8 & - & 2,605 \\ 
8 & 2D local poly - diff bw & 12.9 & 5.4 & 9.4 & 776 \\ 
8 & rdrobust & 33.4 & 18.5 & - & 2,713 \\ 
9 & 2D local poly - diff bw & 12.8 & 5.5 & 9.2 & 715 \\ 
9 & rdrobust & 33.5 & 18.7 & - & 2,721 \\ 
10 & 2D local poly - diff bw & 12.9 & 5.7 & 9.4 & 701 \\ 
10 & rdrobust & 35.1 & 18.6 & - & 2,697 \\ 
\bottomrule
\end{longtable}
\endgroup

%% file: figures/rd2dim_revision/table_application_lee_UrbanPct_bands.tex
\begingroup
\fontsize{12.0pt}{14.4pt}\selectfont
\begin{longtable}{rlrrrr}
\toprule
point & Estimator & pilot & h1 & h2 & eff. sample \\
1 & 2D local poly - diff bw & 42.7 & 33.1 & 19.2 & 2,232 \\ 
1 & rdrobust & 45.6 & 22.2 & - & 846 \\ 
2 & 2D local poly - diff bw & 42.2 & 56.0 & 15.3 & 5,214 \\ 
2 & rdrobust & 42.1 & 20.6 & - & 1,045 \\ 
3 & 2D local poly - diff bw & 41.1 & 40.1 & 16.6 & 3,840 \\ 
3 & rdrobust & 58.9 & 25.7 & - & 1,748 \\ 
4 & 2D local poly - diff bw & 45.5 & 41.5 & 19.0 & 4,463 \\ 
4 & rdrobust & 45.6 & 25.8 & - & 1,863 \\ 
5 & 2D local poly - diff bw & 45.6 & 42.1 & 19.1 & 5,533 \\ 
5 & rdrobust & 42.1 & 25.2 & - & 1,749 \\ 
6 & 2D local poly - diff bw & 41.4 & 40.0 & 16.4 & 5,035 \\ 
6 & rdrobust & 72.1 & 41.3 & - & 4,353 \\ 
7 & 2D local poly - diff bw & 42.5 & 59.7 & 11.6 & 6,651 \\ 
7 & rdrobust & 45.8 & 24.5 & - & 1,691 \\ 
8 & 2D local poly - diff bw & 43.3 & 49.4 & 16.9 & 5,192 \\ 
8 & rdrobust & 56.6 & 23.1 & - & 1,337 \\ 
9 & 2D local poly - diff bw & 55.1 & 80.1 & 22.8 & 7,166 \\ 
9 & rdrobust & 55.3 & 26.5 & - & 1,439 \\ 
10 & 2D local poly - diff bw & 45.2 & 49.7 & 18.7 & 4,457 \\ 
10 & rdrobust & 56.0 & 28.5 & - & 1,507 \\ 
\bottomrule
\end{longtable}
\endgroup

%% file: bibtex_hardcopied_20260107.bib
@article{Arai.Ichimura2018,
  title = {Simultaneous Selection of Optimal Bandwidths for the Sharp Regression Discontinuity Estimator},
  author = {Arai, Yoichi and Ichimura, Hidehiko},
  year = {2018},
  journal = {Quantitative Economics},
  volume = {9},
  number = {1},
  pages = {441--482},
  issn = {1759-7331},
  doi = {10.3982/QE590},
  urldate = {2020-11-09},
  copyright = {Copyright \textcopyright{} 2018 The Authors.},
  langid = {english}
}

@article{Arai.Otsu.Seo2021,
  title = {Regression {{Discontinuity Design}} with {{Potentially Many Covariates}}},
  author = {Arai, Yoichi and Otsu, Taisuke and Seo, Myung Hwan},
  year = {2021},
  month = sep,
  journal = {arXiv:2109.08351 [econ, stat]},
  eprint = {2109.08351},
  primaryclass = {econ, stat},
  urldate = {2021-10-11},
  archiveprefix = {arxiv}
}

@article{Armstrong.Kolesar2018,
  title = {Optimal {{Inference}} in a {{Class}} of {{Regression Models}}},
  author = {Armstrong, Timothy B. and Koles{\'a}r, Michal},
  year = {2018},
  journal = {Econometrica},
  volume = {86},
  number = {2},
  pages = {655--683},
  issn = {1468-0262},
  doi = {10.3982/ECTA14434},
  urldate = {2021-09-29},
  langid = {english}
}

@article{Black1999,
  title = {Do {{Better Schools Matter}}? {{Parental Valuation}} of {{Elementary Education}}*},
  shorttitle = {Do {{Better Schools Matter}}?},
  author = {Black, Sandra E.},
  year = {1999},
  month = may,
  journal = {The Quarterly Journal of Economics},
  volume = {114},
  number = {2},
  pages = {577--599},
  issn = {0033-5533},
  doi = {10.1162/003355399556070},
  urldate = {2023-02-07}
}

@article{Calonico.Cattaneo.Farrell.Titiunik2017,
  title = {Rdrobust: {{Software}} for {{Regression-discontinuity Designs}}},
  shorttitle = {Rdrobust},
  author = {Calonico, Sebastian and Cattaneo, Matias D. and Farrell, Max H. and Titiunik, Roc{\'i}o},
  year = {2017},
  month = jun,
  journal = {The Stata Journal},
  volume = {17},
  number = {2},
  pages = {372--404},
  publisher = {{SAGE Publications}},
  issn = {1536-867X},
  doi = {10.1177/1536867X1701700208},
  urldate = {2023-04-20}
}

@article{Calonico.Cattaneo.Farrell.Titiunik2019,
  title = {Regression {{Discontinuity Designs Using Covariates}}},
  author = {Calonico, Sebastian and Cattaneo, Matias D. and Farrell, Max H. and Titiunik, Roc{\'i}o},
  year = {2019},
  month = jul,
  journal = {The Review of Economics and Statistics},
  volume = {101},
  number = {3},
  pages = {442--451},
  issn = {0034-6535},
  doi = {10.1162/rest_a_00760},
  urldate = {2021-06-04}
}

@article{Calonico.Cattaneo.Farrell2022,
  title = {Coverage Error Optimal Confidence Intervals for Local Polynomial Regression},
  author = {Calonico, Sebastian and Cattaneo, Matias D. and Farrell, Max H.},
  year = {2022},
  month = nov,
  journal = {Bernoulli},
  volume = {28},
  number = {4},
  pages = {2998--3022},
  publisher = {{Bernoulli Society for Mathematical Statistics and Probability}},
  issn = {1350-7265},
  doi = {10.3150/21-BEJ1445},
  urldate = {2023-04-04}
}

@article{Calonico.Cattaneo.Titiunik2014,
  ids = {calonicoRobustNonparametricConfidence2014a,calonicoRobustNonparametricConfidence2014b},
  title = {Robust {{Nonparametric Confidence Intervals}} for {{Regression-Discontinuity Designs}}},
  author = {Calonico, Sebastian and Cattaneo, Matias D. and Titiunik, Rocio},
  year = {2014},
  journal = {Econometrica},
  volume = {82},
  number = {6},
  pages = {2295--2326},
  issn = {1468-0262},
  doi = {10.3982/ECTA11757},
  urldate = {2020-08-16},
  copyright = {\textcopyright{} 2014 The Econometric Society},
  langid = {english}
}

@article{Calonico.Cattaneo.Titiunik2014a,
  title = {Robust {{Data-Driven Inference}} in the {{Regression-Discontinuity Design}}},
  author = {Calonico, Sebastian and Cattaneo, Matias D. and Titiunik, Roc{\'i}o},
  year = {2014},
  month = dec,
  journal = {The Stata Journal},
  volume = {14},
  number = {4},
  pages = {909--946},
  publisher = {{SAGE Publications}},
  issn = {1536-867X},
  doi = {10.1177/1536867X1401400413},
  urldate = {2023-04-20}
}

@article{Cattaneo.Frandsen.Titiunik2015,
  title = {Randomization {{Inference}} in the {{Regression Discontinuity Design}}: {{An Application}} to {{Party Advantages}} in the {{U}}.{{S}}. {{Senate}}},
  shorttitle = {Randomization {{Inference}} in the {{Regression Discontinuity Design}}},
  author = {Cattaneo, Matias D. and Frandsen, Brigham R. and Titiunik, Roc{\'i}o},
  year = {2015},
  month = mar,
  journal = {Journal of Causal Inference},
  volume = {3},
  number = {1},
  pages = {1--24},
  publisher = {{De Gruyter}},
  issn = {2193-3685},
  doi = {10.1515/jci-2013-0010},
  urldate = {2022-01-06},
  langid = {english}
}

@book{Cattaneo.Idrobo.Titiunik2019,
  title = {A {{Practical Introduction}} to {{Regression Discontinuity Designs}}: {{Foundations}}},
  shorttitle = {A {{Practical Introduction}} to {{Regression Discontinuity Designs}}},
  author = {Cattaneo, Matias D. and Idrobo, Nicol{\'a}s and Titiunik, Roc{\'i}o},
  year = {2019},
  month = nov,
  publisher = {{Cambridge University Press}},
  urldate = {2022-02-21},
  isbn = {978-1-108-68460-6 978-1-108-71020-6},
  langid = {english}
}

@book{Cattaneo.Idrobo.Titiunik2023,
  title = {A {{Practical Introduction}} to {{Regression Discontinuity Designs}}: {{Extensions}}},
  shorttitle = {A {{Practical Introduction}} to {{Regression Discontinuity Designs}}},
  author = {Cattaneo, Matias D. and Idrobo, Nicolas and Titiunik, Rocio},
  year = {2024},
  eprint = {2301.08958},
  primaryclass = {econ, stat},
  publisher = {Cambridge University Press},
  urldate = {2023-04-06},
  archiveprefix = {arXiv}
}

@article{Cattaneo.Keele.Titiunik.Vazquez-Bare2016,
  title = {Interpreting {{Regression Discontinuity Designs}} with {{Multiple Cutoffs}}},
  author = {Cattaneo, Matias D. and Keele, Luke and Titiunik, Roc{\'i}o and {Vazquez-Bare}, Gonzalo},
  year = {2016},
  month = may,
  journal = {The Journal of Politics},
  volume = {78},
  number = {4},
  pages = {1229--1248},
  publisher = {{The University of Chicago Press}},
  issn = {0022-3816},
  doi = {10.1086/686802},
  urldate = {2020-09-05}
}

@article{Cattaneo.Keele.Titiunik.Vazquez-Bare2021,
  title = {Extrapolating {{Treatment Effects}} in {{Multi-Cutoff Regression Discontinuity Designs}}},
  author = {Cattaneo, Matias D. and Keele, Luke and Titiunik, Roc{\'i}o and {Vazquez-Bare}, Gonzalo},
  year = {2021},
  month = oct,
  journal = {Journal of the American Statistical Association},
  volume = {116},
  number = {536},
  pages = {1941--1952},
  publisher = {{Taylor \& Francis}},
  issn = {0162-1459},
  doi = {10.1080/01621459.2020.1751646},
  urldate = {2023-04-06}
}

@article{Cattaneo.Titiunik.Vazquez-Bare2016,
  title = {Inference in {{Regression Discontinuity Designs}} under {{Local Randomization}}},
  author = {Cattaneo, Matias D. and Titiunik, Roc{\'i}o and {Vazquez-Bare}, Gonzalo},
  year = {2016},
  month = jun,
  journal = {The Stata Journal},
  volume = {16},
  number = {2},
  pages = {331--367},
  publisher = {{SAGE Publications}},
  issn = {1536-867X},
  doi = {10.1177/1536867X1601600205},
  urldate = {2021-07-12},
  langid = {english}
}

@article{Cattaneo.Titiunik.Vazquez-Bare2017,
  title = {Comparing {{Inference Approaches}} for {{RD Designs}}: {{A Reexamination}} of the {{Effect}} of {{Head Start}} on {{Child Mortality}}},
  shorttitle = {Comparing {{Inference Approaches}} for {{RD Designs}}},
  author = {Cattaneo, Matias D. and Titiunik, Roc{\'i}o and {Vazquez-Bare}, Gonzalo},
  year = {2017},
  journal = {Journal of Policy Analysis and Management},
  volume = {36},
  number = {3},
  pages = {643--681},
  issn = {1520-6688},
  doi = {10.1002/pam.21985},
  urldate = {2021-07-12},
  langid = {english}
}

@article{Cattaneo.Titiunik.Vazquez-Bare2020,
  ids = {Cattaneo.Titiunik.Vazquez-Bare2020a},
  title = {Analysis of Regression-Discontinuity Designs with Multiple Cutoffs or Multiple Scores},
  author = {Cattaneo, Matias D. and Titiunik, Roc{\'i}o and {Vazquez-Bare}, Gonzalo},
  year = {2020},
  month = dec,
  journal = {The Stata Journal: Promoting communications on statistics and Stata},
  volume = {20},
  number = {4},
  pages = {866--891},
  issn = {1536-867X, 1536-8734},
  doi = {10.1177/1536867X20976320},
  urldate = {2022-12-14},
  langid = {english}
}

@article{Caughey.Sekhon2011,
  title = {Elections and the {{Regression Discontinuity Design}}: {{Lessons}} from {{Close U}}.{{S}}. {{House Races}}, 1942\textendash 2008},
  shorttitle = {Elections and the {{Regression Discontinuity Design}}},
  author = {Caughey, Devin and Sekhon, Jasjeet S.},
  year = {2011},
  journal = {Political Analysis},
  volume = {19},
  number = {4},
  pages = {385--408},
  publisher = {{Cambridge University Press}},
  issn = {1047-1987, 1476-4989},
  doi = {10.1093/pan/mpr032},
  urldate = {2020-04-29},
  langid = {english}
}

@incollection{DiNardo.Lee2011,
  title = {Chapter 5 - {{Program Evaluation}} and {{Research Designs}}},
  booktitle = {Handbook of {{Labor Economics}}},
  author = {DiNardo, John and Lee, David S.},
  editor = {Ashenfelter, Orley and Card, David},
  year = {2011},
  month = jan,
  volume = {4},
  pages = {463--536},
  publisher = {{Elsevier}},
  doi = {10.1016/S0169-7218(11)00411-4},
  urldate = {2023-04-17},
  langid = {english}
}

@article{Fan.Gijbels1992,
  title = {Variable {{Bandwidth}} and {{Local Linear Regression Smoothers}}},
  author = {Fan, Jianqing and Gijbels, Irene},
  year = {1992},
  journal = {Annals of Statistics},
  volume = {20},
  number = {4},
  pages = {2008--2036},
  publisher = {{Institute of Mathematical Statistics}},
  issn = {0090-5364, 2168-8966},
  doi = {10.1214/aos/1176348900},
  urldate = {2020-11-19},
  langid = {english},
  mrnumber = {MR1193323},
  zmnumber = {0765.62040}
}

@article{Frolich.Huber2019,
  title = {Including {{Covariates}} in the {{Regression Discontinuity Design}}},
  author = {Fr{\"o}lich, Markus and Huber, Martin},
  year = {2019},
  month = oct,
  journal = {Journal of Business \& Economic Statistics},
  volume = {37},
  number = {4},
  pages = {736--748},
  publisher = {{Taylor \& Francis}},
  issn = {0735-0015},
  doi = {10.1080/07350015.2017.1421544},
  urldate = {2023-04-20}
}

@article{Hahn.Todd.Klaauw2001,
  ids = {hahnIdentificationEstimationTreatment2001a,hahnIdentificationEstimationTreatment2001b,hahnIdentificationEstimationTreatment2001c},
  title = {Identification and {{Estimation}} of {{Treatment Effects}} with a {{Regression-Discontinuity Design}}},
  author = {Hahn, Jinyong and Todd, Petra and der Klaauw, Wilbert Van},
  year = {2001},
  journal = {Econometrica},
  volume = {69},
  number = {1},
  pages = {201--209},
  issn = {1468-0262},
  doi = {10.1111/1468-0262.00183},
  urldate = {2020-04-21},
  copyright = {The Econometric Society, 2001},
  langid = {english}
}

@article{Imbens.Kalyanaraman2012,
  ids = {imbensOptimalBandwidthChoice2012a,imbensOptimalBandwidthChoice2012b,imbensOptimalBandwidthChoice2012c},
  title = {Optimal {{Bandwidth Choice}} for the {{Regression Discontinuity Estimator}}},
  author = {Imbens, Guido and Kalyanaraman, Karthik},
  year = {2012},
  month = jul,
  journal = {The Review of Economic Studies},
  volume = {79},
  number = {3},
  pages = {933--959},
  publisher = {{Oxford Academic}},
  issn = {0034-6527},
  doi = {10.1093/restud/rdr043},
  urldate = {2020-11-19},
  langid = {english}
}

@article{Imbens.Lemieux2008,
  title = {Regression Discontinuity Designs: {{A}} Guide to Practice},
  shorttitle = {Regression Discontinuity Designs},
  author = {Imbens, Guido W. and Lemieux, Thomas},
  year = {2008},
  month = feb,
  journal = {Journal of Econometrics},
  series = {The Regression Discontinuity Design: {{Theory}} and Applications},
  volume = {142},
  number = {2},
  pages = {615--635},
  issn = {0304-4076},
  doi = {10.1016/j.jeconom.2007.05.001},
  urldate = {2021-09-07},
  langid = {english}
}

@article{Imbens.Wager2019,
  title = {Optimized {{Regression Discontinuity Designs}}},
  author = {Imbens, Guido and Wager, Stefan},
  year = {2019},
  month = may,
  journal = {The Review of Economics and Statistics},
  volume = {101},
  number = {2},
  pages = {264--278},
  issn = {0034-6535},
  doi = {10.1162/rest_a_00793},
  urldate = {2021-11-04}
}

@article{Keele.Titiunik.Zubizarreta2015,
  title = {Enhancing a Geographic Regression Discontinuity Design through Matching to Estimate the Effect of Ballot Initiatives on Voter Turnout},
  author = {Keele, Luke and Titiunik, Roc{\'i}o and Zubizarreta, Jos{\'e} R.},
  year = {2015},
  journal = {Journal of the Royal Statistical Society: Series A (Statistics in Society)},
  volume = {178},
  number = {1},
  pages = {223--239},
  issn = {1467-985X},
  doi = {10.1111/rssa.12056},
  urldate = {2023-04-03},
  langid = {english}
}

@article{Keele.Titiunik2015,
  title = {Geographic {{Boundaries}} as {{Regression Discontinuities}}},
  author = {Keele, Luke J. and Titiunik, Roc{\'i}o},
  year = {2015},
  journal = {Political Analysis},
  volume = {23},
  number = {1},
  pages = {127--155},
  publisher = {{Cambridge University Press}},
  issn = {1047-1987, 1476-4989},
  doi = {10.1093/pan/mpu014},
  urldate = {2021-08-14},
  langid = {english}
}

@article{Kreiss.Rothe2021,
  title = {Inference in {{Regression Discontinuity Designs}} with {{High-Dimensional Covariates}}},
  author = {Krei{\ss}, Alexander and Rothe, Christoph},
  year = {2021},
  month = oct,
  journal = {arXiv:2110.13725 [econ, stat]},
  eprint = {2110.13725},
  primaryclass = {econ, stat},
  urldate = {2021-11-19},
  archiveprefix = {arxiv}
}

@article{Kurisu.Matsuda2022,
  title = {Local Polynomial Regression for Spatial Data on $\mathbb{R}^d$},
  author = {Kurisu, Daisuke and Matsuda, Yasumasa},
  year = {2024},
  month = nov,
  journal = {Bernoulli},
  volume = {30},
  number = {4},
  pages = {2770--2794}
}

@article{Kwon.Kwon2020a,
  title = {Adaptive {{Inference}} in {{Multivariate Nonparametric Regression Models Under Monotonicity}}},
  author = {Kwon, Koohyun and Kwon, Soonwoo},
  year = {2020},
  month = nov,
  eprint = {2011.14219},
  primaryclass = {econ, math, stat},
  urldate = {2023-03-02},
  archiveprefix = {arxiv},
  langid = {english}
}

@article{Lee.Lemieux2010,
  title = {Regression {{Discontinuity Designs}} in {{Economics}}},
  author = {Lee, David S and Lemieux, Thomas},
  year = {2010},
  month = jun,
  journal = {Journal of Economic Literature},
  volume = {48},
  number = {2},
  pages = {281--355},
  issn = {0022-0515},
  doi = {10.1257/jel.48.2.281},
  urldate = {2020-05-29},
  langid = {english}
}

@article{Lee2008,
  ids = {leeRandomizedExperimentsNonrandom2008a,leeRandomizedExperimentsNonrandom2008b},
  title = {Randomized Experiments from Non-Random Selection in {{U}}.{{S}}. {{House}} Elections},
  author = {Lee, David S.},
  year = {2008},
  month = feb,
  journal = {Journal of Econometrics},
  series = {The Regression Discontinuity Design: {{Theory}} and Applications},
  volume = {142},
  number = {2},
  pages = {675--697},
  issn = {0304-4076},
  doi = {10.1016/j.jeconom.2007.05.004},
  urldate = {2020-04-27},
  langid = {english}
}

@article{Londono-Velez.Rodriguez.Sanchez2020,
  title = {Upstream and {{Downstream Impacts}} of {{College Merit-Based Financial Aid}} for {{Low-Income Students}}: {{Ser Pilo Paga}} in {{Colombia}}},
  shorttitle = {Upstream and {{Downstream Impacts}} of {{College Merit-Based Financial Aid}} for {{Low-Income Students}}},
  author = {{Londo{\~n}o-V{\'e}lez}, Juliana and Rodr$\acute\imath$guez, Catherine and S{\'a}nchez, Fabio},
  year = {2020},
  month = may,
  journal = {American Economic Journal: Economic Policy},
  volume = {12},
  number = {2},
  pages = {193--227},
  issn = {1945-7731, 1945-774X},
  doi = {10.1257/pol.20180131},
  langid = {english}
}

@article{Masry1996,
  title = {Multivariate Regression Estimation Local Polynomial Fitting for Time Series},
  author = {Masry, Elias},
  year = {1996},
  month = dec,
  journal = {Stochastic Processes and their Applications},
  volume = {65},
  number = {1},
  pages = {81--101},
  issn = {0304-4149},
  doi = {10.1016/S0304-4149(96)00095-6},
  urldate = {2022-09-18},
  langid = {english}
}

@article{Matsudaira2008,
  title = {Mandatory Summer School and Student Achievement},
  author = {Matsudaira, Jordan D.},
  year = {2008},
  month = feb,
  journal = {Journal of Econometrics},
  volume = {142},
  number = {2},
  pages = {829--850},
  issn = {03044076},
  doi = {10.1016/j.jeconom.2007.05.015},
  urldate = {2023-02-07},
  langid = {english}
}

@article{Noack.Olma.Rothe2021,
  ids = {Noack.Olma.Rothe,Noack.Olma.Rothe2021a},
  title = {Flexible {{Covariate Adjustments}} in {{Regression Discontinuity Designs}}},
  author = {Noack, Claudia and Olma, Tomasz and Rothe, Christoph},
  year = {2021},
  month = jul,
  journal = {arXiv:2107.07942 [econ, stat]},
  eprint = {2107.07942},
  primaryclass = {econ, stat},
  urldate = {2021-07-23},
  archiveprefix = {arxiv}
}

@article{Papay.Willett.Murnane2011,
  title = {Extending the Regression-Discontinuity Approach to Multiple Assignment Variables},
  author = {Papay, John P. and Willett, John B. and Murnane, Richard J.},
  year = {2011},
  month = apr,
  journal = {Journal of Econometrics},
  volume = {161},
  number = {2},
  pages = {203--207},
  issn = {0304-4076},
  doi = {10.1016/j.jeconom.2010.12.008},
  urldate = {2022-05-17},
  langid = {english}
}

@article{Ruppert.Wand1994,
  title = {Multivariate {{Locally Weighted Least Squares Regression}}},
  author = {Ruppert, D. and Wand, M. P.},
  year = {1994},
  journal = {The Annals of Statistics},
  volume = {22},
  number = {3},
  eprint = {2242229},
  eprinttype = {jstor},
  pages = {1346--1370},
  publisher = {{Institute of Mathematical Statistics}},
  issn = {0090-5364},
  urldate = {2022-10-05}
}

@techreport {londono_data_2020,
author = {{Londo{\~n}o-V{\'e}lez}, Juliana and Rodr$\acute\imath$guez, Catherine and S{\'a}nchez, Fabio},
year = 2020 ,
title = { Replication package for: Upstream and Downstream Impacts of College Merit-Based Financial Aid for Low-Income Students: Ser Pilo Paga in Colombia. },
institution = { American Economic Association [publisher]. },
type = {} ,
note = {Accessed at \url{https://www.aeaweb.org/journals/dataset?id=10.1257/pol.20180131} on 2024-01-23.}
}

@article{guMultivariateLocalPolynomial2015,
  title = {Multivariate {{Local Polynomial Kernel Estimators}}: {{Leading Bias}} and {{Asymptotic Distribution}}},
  shorttitle = {Multivariate {{Local Polynomial Kernel Estimators}}},
  author = {Gu, Jingping and Li, Qi and Yang, Jui-Chung},
  year = 2015,
  month = may,
  journal = {Econometric Reviews},
  volume = {34},
  number = {6-10},
  pages = {979--1010},
  publisher = {Taylor \& Francis},
  issn = {0747-4938},
  doi = {10.1080/07474938.2014.956615},
  urldate = {2025-12-23}
}

@article{cattaneoEstimationInferenceBoundary2025,
  title = {Estimation and {{Inference}} in {{Boundary Discontinuity Designs}}},
  author = {Cattaneo, Matias D. and Titiunik, Rocio and Yu, Ruiqi Rae},
  year = 2025,
  month = may,
  eprint = {2505.05670},
  journal = {arXiv:2505.05670},
  primaryclass = {econ},
  doi = {10.48550/arXiv.2505.05670},
  urldate = {2025-09-02},
  archiveprefix = {arXiv}
}

@article{wongAnalyzingRegressionDiscontinuityDesigns2013,
  title = {Analyzing {{Regression-Discontinuity Designs With Multiple Assignment Variables}}: {{A Comparative Study}} of {{Four Estimation Methods}}},
  shorttitle = {Analyzing {{Regression-Discontinuity Designs With Multiple Assignment Variables}}},
  author = {Wong, Vivian C. and Steiner, Peter M. and Cook, Thomas D.},
  year = 2013,
  journal = {Journal of Educational and Behavioral Statistics},
  volume = {38},
  number = {2},
  eprint = {41999417},
  eprinttype = {jstor},
  pages = {107--141},
  publisher = {American Educational Research Association},
  issn = {1076-9986},
  urldate = {2026-01-07}
}

@article{reardonRegressionDiscontinuityDesigns2012,
  title = {Regression {{Discontinuity Designs With Multiple Rating-Score Variables}}},
  author = {Reardon, Sean F. and Robinson, Joseph P.},
  year = 2012,
  month = jan,
  journal = {Journal of Research on Educational Effectiveness},
  volume = {5},
  number = {1},
  pages = {83--104},
  publisher = {Routledge},
  issn = {1934-5747},
  doi = {10.1080/19345747.2011.609583},
  urldate = {2026-01-07}
}

@data{caugheyDataCitation,
author = {Devin M. Caughey and Jasjeet S. Sekhon},
publisher = {Harvard Dataverse},
title = {{Replication data for: Elections and the Regression-Discontinuity Design: Lessons from Close U.S. House Races, 1942-2008}},
UNF = {UNF:5:AI9kprv6ytPW1MxsspufoA==},
year = {2011},
version = {V2},
doi = {10.7910/DVN/8EYYA2},
note = {Publisher: Harvard Dataverse. URL: https://doi.org/10.7910/DVN/8EYYA2}
}
